\documentclass[12pt]{article}

\usepackage{amsmath, amsthm, amssymb}
\usepackage{graphicx}
\usepackage[round]{natbib}
\usepackage{enumerate}
\usepackage[a4paper, lmargin=2.7cm, rmargin=2.5cm, bottom=2.5cm, top=2.3cm]{geometry}
\usepackage{float}
\usepackage{subcaption}

\newtheorem{theorem}{Theorem}[section]
\newtheorem{lemma}[theorem]{Lemma}
\newtheorem{remark}[theorem]{Remark}
\newtheorem{corollary}[theorem]{Corollary}
\newtheorem{proposition}[theorem]{Proposition}
\newtheorem{claim}[theorem]{Claim}

\theoremstyle{definition}

\newtheorem{definition}[theorem]{Definition}
\newtheorem{assumption}{Assumption}

\usepackage{xcolor}
\usepackage{hyperref}

\graphicspath{{../Images/}{../Figures/}}

\usepackage{./mymacros}
\usepackage{./packagemacros}

\usepackage{yfonts}
\usepackage{algorithm}
\usepackage{algpseudocode}

\title{FDP control in multivariate linear models using the bootstrap}
\author{Samuel Davenport, Bertrand Thirion, Pierre Neuvial}

\def\biometrika{0}

\if\biometrika1
\def\bpa{_biometrika_1}
\else
\def\bpa{}
\fi

\begin{document}
\maketitle
\begin{abstract}
	In this article we develop a method for performing post hoc inference of the False Discovery Proportion (FDP) over multiple contrasts of interest in the multivariate linear model. To do so we use the bootstrap to simulate from the distribution of the null contrasts. We combine the bootstrap with the post hoc inference bounds of \cite{Blanchard2020} and prove that doing so provides simultaneous asymptotic control of the FDP over all subsets of hypotheses. This requires us to demonstrate consistency of the multivariate bootstrap in the linear model, which we do via the Lindeberg Central Limit Theorem, providing a simpler proof of this result than that of \cite{Eck2018}. We demonstrate, via simulations, that our approach provides simultaneous control of the FDP over all subsets and is typically more powerful than existing, state of the art, parametric methods. We illustrate our approach on functional Magnetic Resonance Imaging data from the Human Connectome project and on a transcriptomic dataset of chronic obstructive pulmonary disease. 

\if\biometrika1
\else
\fi

\end{abstract}

\section{Introduction}

Statistical analysis of functional Magnetic Resonance Imaging data grounds the inference of associations between external conditions (such as disease status and experimental factors) and the signals recorded in brain regions, that is assumed to reflect brain activity.
In particular, practitioners typically aim to uncover associations between local signals and conditions that are specific to a given area; such specificity is essential for interpretation purposes.
The most standard framework is that of mass-univariate inference, in which models are fit separately at each brain location, in order to detect significant associations.
This framework is simple and computationally efficient but, given mm-scale resolution reached by current imaging setups, results in a dire multiple comparison issue.

Statistical analysis of genomic data encounters a similar multiple comparison problem.
In particular, this is the case in Genome-Wide Association Studies that aim to identify Single Nucleotide Polymorphisms that are associated with one or more phenotypes of interest, and in gene expression studies, the goal of which is to identify genes where activity is associated with one or more variables of biological or clinical interest.
In this field, the state-of-the-art framework is also based on univariate tests that are performed for each genomic marker. 
While imaging data typically consist of a smooth volume-domain voxel grid, the dependence structure of genomic data is dictated by the interdependence between genomic markers, which is mediated by haplotypic blocks encountered in Genome-Wide Association Studies and by gene networks or pathways in expression studies. 

In both of these scientific fields (and in others), control of the false discovery rate (FDR) has quickly become a de facto standard, as it yields image or genome-level error control together with acceptable power \citep{genovese2002,Storey2003}.
In practice, most researchers control the FDR using the Benjamini-Hochberg procedure \citep{Benjamini1995}, under the assumption of positive regression dependence \citep{Benjamini2001}. This assumption is generally considered reasonable given the positive correlation that typically exists between voxels or genomic markers.
However, users often interpret FDR-control as a control of false discovery proportion (FDP), which is incorrect, as the FDR is only the expected value of
the FDP.
%
%
Overall, this approach can result in unreliable error control, especially when there is dependence within the data, see \cite{Korn2004} and Figure 2.1 in \cite{Neuvial2020}. It is instead desirable to provide probabilistic control on the \emph{proportion} or \emph{number} of false discoveries.

Genomic and brain imaging datasets frequently involve the simultaneous test of several contrasts \citep{Smyth2004,Alberton2020}.
Such simultaneous tests are important because they can ground double dissociation \citep{henson2006}, ensuring the specificity of discoveries and leading to unambiguous interpretations of the results.
A difficulty arises here as the tests of the different contrasts that are considered at each feature (voxel/gene) are typically dependent and it may no longer be reasonable to assume positive regression dependence. It is thus of interest to consider controlling the FDP under the null hypothesis for each contrast, without making unwanted assumptions.

The notion of \emph{post hoc inference} was introduced by \cite{Goeman2011}, following earlier works by \cite{Genovese2006,Meinshausen2006} on the probabilistic control of the FDP. 
The idea of post hoc inference is to provide confidence bounds on the number or proportion of true/false discoveries among arbitrary and possibly data-driven subsets of variables of interest. 
By construction, such guarantees address the issue of circular inference \citep{Rosenblatt2018}.

Post hoc bounds can be obtained as a by-product of the control of a multiple testing risk called the joint error rate (JER) by a simple interpolation argument \citep{Blanchard2020}.
Using this construction, state-of-the-art post hoc bounds~\citep{Goeman2011,Rosenblatt2018} can be recovered from the Simes inequality, a classical result from the multiple testing literature \citep{Simes1986}.
The resulting bounds are valid under positive regression dependence. They have also been shown to be conservative in genomics and neuroimaging applications \citep{Blanchard2021}. 

Since the joint error rate only depends on the joint distribution of the test statistics under the null hypothesis, joint error rate control can alternatively be obtained by randomization techniques~\citep{Blanchard2020,Blanchard2021,Hemerik2019}.
In particular, sharp data-driven joint error rate control and associated post hoc bounds have been obtained for one-sample tests using sign-flipping, and for two-sample tests using permutations~\citep{Blanchard2020,Blanchard2021}.
%
%
%
%
However, obtaining joint error rate control more generally in linear models, especially when doing inference on multiple contrasts, remains an open question to the best of our knowledge. 

In order to perform post hoc inference over contrasts we will need to be able to obtain the joint null distribution of the test-statistics of multiple contrasts within the framework of the linear model. To do so we use the bootstrap, adjusting the approach of \cite{Westfall2011} to the multivariate setting. Justification for bootstrapping the residuals in a one-dimensional linear model was first provided in \cite{Freedman1981} based on theory proved in \cite{Bickel1981}. These results (and their proofs) were extended to multivariate linear models in \cite{Eck2018}. In this work we provide an alternative, simpler, proof of the validity of the bootstrap in the linear model that relies on the Lindeberg Central Limit Theorem. 
We use these results to show that we can obtain  asymptotically valid simultaneous FDP control. This extends the work of \cite{Blanchard2020} to the setting of the linear model. 

A python package with code to run the methods detailed in this paper is available at: \texttt{http://github.com/sjdavenport/pyperm}, Jupyter notebooks with illustrated examples of its use in practice are also available there. Moreover, code to reproduce the analyses and figures of this paper is available at: \texttt{http://github.com/sjdavenport/lmfdp}. Proofs and further theoretical and simulation results are available in the
\if\biometrika1
 supplementary material - sections of which will be denoted using the suffix S.
 \else
 Appendix.
 \fi


\section{Notation and general framework}\label{S:notation}
\subsection{Random Fields on a Lattice}\label{SS:rfs}
Throughout we will take $ (\Omega, \mathcal{F}, \mathbb{P}) $ to be a probability space, write $ \mathbb{E} $  to denote expectation and will define random variables with respect to this space.  We will also take $ \mathbb{N} $ to be the set of positive integers.  We will primarily be working with random fields, observed at a finite number of points, as our data. These are defined as follows.
\if\biometrika1
\vspace{0.06cm}
\fi
\begin{definition}
	Given $ D, L \in \mathbb{N} $ and a finite set $ \mathcal{V} \subset \mathbb{R}^D $, we define a \textbf{random field} on $ \mathcal{V} $ to be a measurable mapping $ g: \Omega \rightarrow \left\lbrace h:\mathcal{V} \rightarrow \mathbb{R}^L \right\rbrace $. We will say that $ g $ has \textbf{dimension} $ L. $
\end{definition}
\if\biometrika1
\vspace{0.05cm}
\fi
%

Given $ \omega \in \Omega $ and $ v \in \mathcal{V} $ we will write $ g(\omega, v) = g(\omega)(v) $ and will typically drop dependence on $\omega $ and simply refer to the random variable $ g(v): \Omega \rightarrow \mathbb{R}^L $ when indexing $ g$ and say that $ g $ is a random field on $ \mathcal{V} $.  We define the mean of $ g $ to be the function $ \mu:\mathcal{V} \rightarrow \mathbb{R}^L $ sending $ v \in \mathcal{V} $ to $ \mathbb{E}\left[ g(v) \right] $. To each $ g $ we associate a covariance $ \textswab{c}$ and a correlation function $ \rho $ which map $ \mathcal{V} \times \mathcal{V}$ to $\mathbb{R}^{L \times L} $ and are defined as
\begin{equation*}
\textswab{c}(u, v) = \cov(g(u), g(v)) = \mathbb{E}\left[ \left( g(u) - \mu(u)\right) \left( g(v) - \mu(v)\right)^T  \right]
\end{equation*}
and $ \rho(u,v) = \textswab{c}(u,v)\left( \textswab{c}(u,u)\textswab{c}(v,v) \right)^{-1/2} $
for all $ u, v \in \mathcal{V} $. 

For $ 1 \leq j \leq L $, we define the random fields $ g_j:  \Omega \rightarrow \left\lbrace g:\mathcal{V} \rightarrow \mathbb{R}\right\rbrace $ which send $ \omega \in \Omega $ to $ g_j(\omega)(\cdot) = g(\omega)(\cdot)_j = g(\cdot)_j $. We will call $ g_1, \dots, g_L $ the \textbf{components} of $ g $ and will write the combination as $ g = [g_1, \dots, g_L]^T. $ Convergence in distribution and probability (which we will denote by $ \convd $ and $ \overset{\mathbb{P}}{\longrightarrow} $) of random fields is well defined via vectorization, see Section \ref{SS:furtherf} for a formalization of this and for how to define operations on random fields. We also define Gaussian random fields as follows.
\if\biometrika1
\vspace{0.05cm}
\fi
\begin{definition}
	Given functions $ \mu:\mathcal{V} \rightarrow \mathbb{R}^L $ and $ \textswab{c}:\mathcal{V}\times \mathcal{V}$ we write $ g \sim \mathcal{G}(\mu, \textswab{c})$ if $ g $ is a random field with mean $ \mu $ and covariance $ \textswab{c} $ and such that the vector $ (g_j(v): v \in \mathcal{V}, 1 \leq j \leq L) $ has a multivariate Gaussian distribution.
\end{definition}

\subsection{Linear Model Framework}\label{SS:LM}

Let $ \mathcal{V} \subset \mathbb{R}^D$ be a finite set of points corresponding to the domain of interest (this could for instance be the voxels of the brain or points representing genes). Suppose that we observe random fields $ y_i: \mathcal{V} \rightarrow \mathbb{R} $, for $ 1\leq i \leq n $ and some number of subjects $ n \in \mathbb{N}$. At each point $v \in \mathcal{V}$,  we assume that
\begin{equation}\label{eq:LM}
 	Y_n(v) = X_n \beta(v) + E_n(v)
\end{equation}
where for each $ v \in \mathcal{V} $, $ Y_n(v) = [y_1(v), \dots, y_n(v)]^T $ is a vector giving the observed data, $ \beta(v) \in \mathbb{R}^p$ is the vector of parameters (for some $ p \in \mathbb{N} $), $X_n\in \mathbb{R}^{n \times p}$ is the design matrix of the covariates (note that this may include nuisance variables) and $ E_n = [\epsilon_1, \dots, \epsilon_n]^T $ is an $ n $-dimensional random field on $ \mathcal{V} $ which represents the unobserved noise, where $ (\epsilon_n)_{n \in \mathbb{N}} $ is an i.i.d sequence of 1-dimensional random fields on $ \mathcal{V} $. Note that we give the design matrix $ X_n $ a subscript $ n $ as we will allow it to grow with $ n. $ Let $ \textswab{c}_{\epsilon} $ and $ \rho_\epsilon $ be the covariance and correlation functions of $ \epsilon_1 $ respectively.

Then, given contrasts $ c_1, \dots, c_L \in \mathbb{R}^{p}$ for some number of contrasts $ L \in \mathbb{N}$, we are interested in testing the null hypotheses: $ H_{0,l}(v): c_l^T \beta(v) = 0, $
for $ 1 \leq l \leq L $ and each $ v \in \mathcal{V} $. For each $ v \in \mathcal{V} $ we can test the intersection hypothesis
\begin{equation*}
H_0(v): c_l^T \beta(v) = 0 \text{ for } 1 \leq l \leq L
\end{equation*}
using an $ F $-test at each $ v \in \mathcal{V} $ given by 
\begin{equation}\label{eq:F}
F_n(v) = \frac{(C\hat{\beta}_n(v))^T(C(X_n^TX_n)^{-1}C^T)^{-1}(C\hat{\beta}_n(v))/\text{rank}(C)}{\hat{\sigma}_n(v)^2}.
\end{equation}
Here $ \hat\beta_n(v) = (X_n^TX_n)^{-1}X_n^TY_n(v) $ and $ C = (c_1, \dots, c_L)^T  \in \mathbb{R}^{L \times p}$ is the matrix of contrasts. $ \hat{\sigma}_n^2:\mathcal{V}\rightarrow \mathbb{R} $ is the estimate of the variance based on the residuals which sends $ v \in \mathcal{V} $ to 
\begin{equation*}
\hat{\sigma}_n^2(v) = \frac{1}{n - r_n}\left\lVert Y_n(v) - X_n\hat\beta_n(v)\right\rVert^2.
\end{equation*}
where $ I_n $ is the $ n \times n $ identity matrix and $ r_n $ is the rank of $ X_n $. The individual null hypotheses can be tested using test statistics:
\begin{equation}\label{eq:t}
T_{n,l}(v) = \frac{c_l^T\hat{\beta}_n(v)}{\sqrt{\hat{\sigma}_n(v)^2 c_l^T(X_n^TX_n)^{-1}c_l}}.
\end{equation}
Under $ H_{0,l}(v) $ and assuming that the noise is Gaussian, conditional on $ X_n $, $  T_{n,l}(v) $ is distributed as a $ t $-statistic with $ n - r_n $ degrees of freedom. This allows a $ p $-value to be calculated for each contrast $ l $ at each point $ v $ as $ p_{n,l}(v) = 2(1 - \Phi_{n - r_n}(\left| T_{n,l}(v) \right|)) $ where $ \Phi_d $ is the CDF of a $ t $-statistic with $ d \in \mathbb{N} $ degrees of freedom. Dropping the Gaussianity assumption, the $ p $-values are still asymptotically valid under reasonable assumptions (see e.g. Theorem \ref{thm:tstatconv}). Moreover, for each $ 1 \leq l \leq L, $ $T_{n,l}$ is a 1-dimensional random field and we define $ T_n = [T_{n,1}, \dots, T_{n,L}]^T $.\\

\subsection{Bounds on the False Discovery Proportion}\label{SS:JER}
The above framework gives us $ m = LV $ different hypothesis tests (where $ V $ is the number of elements of $ \mathcal{V} $), and results in a multiple testing problem, which can be quite severe e.g. if the size of $ \mathcal{V} $ is large. Let $ \mathcal{H}  = \left\lbrace (l,v): 1\leq l \leq L \text{ and } v \in \mathcal{V}  \right\rbrace $ index the hypotheses.  For $ H \subseteq \mathcal{H}  $, let $ \left| H \right| $ denote the number of elements within $ H $. Finally let $ \mathcal{N} \subset \mathcal{H} $ index the true null hypotheses. Given $ 0 < \alpha < 1 $ we will seek to provide a function $ V: \lbrace H: H \subset \mathcal{H} \rbrace  \rightarrow \mathbb{N} $ such that 
\begin{equation}\label{eq:simbound}
\mathbb{P}\left(  |H \cap \mathcal{N}| \leq V(H),\, \text{for all } H \subset \mathcal{H}  \right) \geq 1-\alpha.
\end{equation}
If \eqref{eq:simbound} holds then simultaneously over all $ H \subset \mathcal{H} $, with probability $ 1-\alpha $, $ V(H) $ provides an upper bound on the number of false positives within $ H. $  Suppose that for some $ K \in \mathbb{N} $ we have sets $ R_1, \dots, R_K \subset \mathcal{H}  $ (which depend on the data) and constants $ \zeta_1, \dots, \zeta_K \in \mathbb{N}$ and define 
\begin{equation}\label{eq:JER}
\text{JER}\left( (R_k, \zeta_k)_{1 \leq k \leq K} \right) := \mathbb{P}\left( |R_k \cap \mathcal{N}| > \zeta_k, \text{ some } 1 \leq k \leq K\right)
\end{equation}
to be the joint error rate of the collection $ (R_k, \zeta_k)_{1 \leq k \leq K} $. \cite{Blanchard2020} showed that if $ \text{JER}\left( (R_k, \zeta_k)_{1 \leq k \leq K} \right) \leq \alpha $, then the bound $ \mybar{V}: \lbrace H: H \subset \mathcal{H} \rbrace  \rightarrow \mathbb{R} $, sending $ H \subset \mathcal{H}  $ to 
\begin{equation}\label{eq:Vbar}
\mybar{V}(H) = \min_{1 \leq k \leq K} (|H \setminus R_k| + \zeta_k) \wedge |H|,
\end{equation}
satisfies \eqref{eq:simbound} and thus provides an $ \alpha $-level bound over the number of false positives within each chosen rejection set. If the sets $ R_1, \dots, R_K $ are nested then $ \mybar{V} $ is in fact optimal: this follows by \cite{Blanchard2020}'s Proposition 2.5. We will follow the approach of \cite{Blanchard2020} and define the collections $ (R_k, \zeta_k)_{1 \leq k \leq K} $, that we will use, using template families. 
In this case, on top of being statistically optimal, an important practical feature of the bound  $ \mybar{V}(H) $ is that it can be computed in linear time in $|H|$, see Algorithm~2 in \cite{Enjalbert-Courrech}.
\if\biometrika1
\vspace{-0.35cm}
\fi
\begin{definition}
	Given $ K \in \mathbb{N} $, we say that a family of functions $ (t_k)_{1\leq k \leq K} $ is a \textbf{template family} if for each $ 1 \leq k \leq K $,  $ t_k: [0,1] \rightarrow \mathbb{R} $, $ t_k(0) = 0 $ and $ t_k $ is strictly increasing and continuous. The parameter $ K $ is called the \textbf{size} of the template.
\end{definition}
\if\biometrika1
\vspace{0.05cm}
\fi
The simplest and most commonly used template family is the linear template which, for $ K \in \mathbb{N} $, is given by $ t_k(x) = \frac{xk}{m} $ for $ 1 \leq k \leq K $ and $ x \in [0,1]$. 
Existing post hoc bounds associated with this template are described in Section~\ref{sec:param-approaches}. 
However other choices are available and the optimal choice of template may depend on the dataset under consideration: we refer to Section 
\if\biometrika1
\ref{S:FD}
\else
\ref{S:discussion}
\fi
for further details and a discussion of the choice of template as well as \cite{Hemerik2019}, \cite{Blanchard2020} and \cite{Blain2022}. Given a template family and $ \lambda \in [0,1] $, for each $ 1 \leq k \leq K $ and $ n \in \mathbb{N} $, we will take $ R_k(\lambda) = \left\lbrace (l,v) \in \mathcal{H} : p_{n,l}(v) \leq t_k(\lambda)\right\rbrace, $
set $ \zeta_k = k - 1, $ and let $ p^n_{(k:\mathcal{N})} $ be the $ k $th smallest $ p $-value in the set $ \left\lbrace p_{n,l}(v): (l, v) \in \mathcal{N} \right\rbrace $ (setting $ p^n_{(k:\mathcal{N})} = 1$ if $ k > |\mathcal{N}| $). We will refer to the collection $  (R_k(\lambda), k-1)_{1 \leq k \leq K} $ as the canonical reference family.
\if\biometrika1
\vspace{-0.05cm}
\fi
\begin{lemma}\label{lem:JERSIM}
	For each $ \lambda \in [0,1] $, 
	\begin{align*}
	\text{JER}\left( (R_k(\lambda),k-1)_{1 \leq k \leq K} \right) = \mathbb{P}\left( \min_{1 \leq k \leq K\wedge |\mathcal{H} |} t_k^{-1}(p^n_{(k:\mathcal{N})})\leq \lambda \right).
	\end{align*}
\end{lemma}
\vspace{-0.1cm}
\noindent Thus for a given template family, in order to obtain an upper bound on the number of false positives we can choose a threshold $ \lambda \in [0,1] $ such that 
\begin{equation}\label{eq:probJER}
	\mathbb{P}\left( \min_{1 \leq k \leq K\wedge |\mathcal{H} |} t_k^{-1}(p^n_{(k:\mathcal{N})})\leq \lambda \right) \leq \alpha.
\end{equation}
Then the joint error rate of the family $ (R_k(\lambda),k-1)_{1 \leq k \leq K} $ is controlled to a level $ \alpha $ and so the corresponding bound: $ \mybar{V} $, provides a $ (1-\alpha) $-level simultaneous upper bound on the number of false positives. 

\cite{Blanchard2020} chose $ \lambda $ via permutation testing, using the fact that under an exchangeability assumption permutation allows the probability in \eqref{eq:probJER} to be controlled exactly. In the linear model, permutation of the response does not satisfy the exchangeability assumption when there are multiple potentially non-zero covariates in the model (see Appendix \ref{A:permNwork} for a discussion of this).
In what follows we take a different approach that proceeds via bootstrapping the data and results in asymptotic control of the error rate.

For $\alpha \in (0,1)$  $ \mybar{V}(H)$ provides an $ (1-\alpha) $-level simultaneous upper bound on the number of false positives within $ H $. From \eqref{eq:simbound} we have 
\begin{equation}
\mathbb{P}\left(  \frac{|H \cap \mathcal{N}|}{|H|} \leq \frac{\mybar{V}(H)}{|H|},\,\, \forall H \subset \mathcal{H}  \right) \geq 1-\alpha.
\end{equation}
It thus follows that for each $ H \in \mathcal{H} $, $\frac{ \overline{V}(H)}{\left| H \right|} $ provides an upper bound on the proportion of false positives within $ H $ also known as the \textbf{false discovery proportion} or \textbf{FDP}. Similarly $\frac{|H| - \overline{V}(H)}{|H|} $ provides a $ (1-\alpha) $-level simultaneous lower bound on the \textbf{true discovery proportion} or \textbf{TDP}.


\section{Bootstrapping in the Linear Model}\label{S:bootstrap}
\subsection{Bootstrapping}\label{SS:bootstrapping}
There are several different ways to bootstrap data in the linear model (see \cite{Freedman1981}). We shall focus on the residual bootstrap. Given $ n \in \mathbb{N} $, this proceeds by calculating the residuals 
\if\biometrika1
$ \hat{E}_n = Y_n - X_n\hat{\beta}_n =  (I_n - X_n(X_n^TX_n)^{-1}X_n^T) E_n $,
\else
\begin{equation}
\hat{E}_n = Y_n - X_n\hat{\beta}_n =  (I_n - X_n(X_n^TX_n)^{-1}X_n^T) E_n,
\end{equation}
\fi
where $ I_n $ is the $ n \times n $ identity matrix and
\if\biometrika1
$ \hat{\beta}_n = (X_n^TX_n)^{-1}X_n^T Y_n  = \beta + (X_n^TX_n)^{-1}X_n^TE_n. $
\else
\begin{equation}
	 \hat{\beta}_n = (X_n^TX_n)^{-1}X_n^T Y_n  = \beta + (X_n^TX_n)^{-1}X_n^TE_n.
\end{equation}
\fi
  Given a number of bootstraps to perform: $ B \in \mathbb{N} $ for each $ 1 \leq b \leq B $, conditional on the data, a selection: $ \hat{\epsilon}_1^b, \dots, \hat{\epsilon}_n^b $ is chosen independently with replacement from $ \left\lbrace \hat{E}_{n,1}, \dots, \hat{E}_{n,n}\right\rbrace$ resulting in a combined random field $E^b_n = [\hat{\epsilon}_1^b, \dots, \hat{\epsilon}_n^b]^T.$ Given this let $ Y^b_n = X_n \hat{\beta}_n + E^b_n $ and define bootstrapped parameter estimates $ \hat{\beta}_n^b = (X_n^TX_n)^{-1}X_n^T Y^b_n. $ For $n \in \mathbb{N}$, let $ (\hat{\sigma}^b_n)^2 = \frac{1}{n}\sum_{i = 1}^n (E^b_{n,i})^2 - \left( \frac{1}{n} \sum_{i = 1}^n E^b_{n,i} \right)^2$ be the estimate of variance using the bootstrap residuals.

\subsection{Convergence Results}\label{SS:convergence}
Given this set-up we will now prove convergence of the bootstrapped test-statistics. To do so we will require the following assumption.
\if\biometrika1
\vspace{0.1cm}
\fi
\begin{assumption}\label{ass:X}
	\phantom{zzzzz}\\
	\phantom{zzzzz}a) For $ n \in \mathbb{N} $, $ X_n = [x_1, \dots, x_n]^T $ for a sequence of i.i.d vectors $ \left( x_n \right)_{n \in \mathbb{N}} $ in $ \mathbb{R}^p $ with bounded density and such that $ \mathbb{E}\left[ \left\lVert x_1 \right\rVert^{2+\delta} \right] < \infty$ for some $ \delta > 0 $. Let $ \Sigma_X = \mathbb{E}[x_1x_1^T]. $\\
	\phantom{zzzzz}b)  $ (\epsilon_n)_{n \in \mathbb{N}} $ is an i.i.d sequence of 1-dimensional random fields on $ \mathcal{V} $ which is independent of $ \left( x_n \right)_{n \in \mathbb{N}} $ and such that $ \max_{v \in \mathcal{V}} \mathbb{E}\left[ \epsilon_1(v)^4 \right] < \infty $ and $ \min_{v \in \mathcal{V}} \var(\epsilon_1(v)) > 0.$
\end{assumption}
\if\biometrika1
\vspace{0.2cm}
\fi
\if\biometrika1
\begin{theorem}\label{thm:bootactual}
	(Bootstrap convergence.) Suppose that $ (X_n)_{n \in \mathbb{N}} $ and $ (\epsilon_n)_{n \in \mathbb{N}} $ satisfy Assumption \ref{ass:X}. Then conditional on $ \left(X_m, Y_m \right)_{m \in \mathbb{N}} $, for almost all sequences $ \left(X_m, Y_m \right)_{m \in \mathbb{N}} $, for each $ 1 \leq b \leq B, $ as $ n \rightarrow \infty, \sqrt{n}(\hat{\beta}_n^b -\hat{\beta}_n) \convd \mathcal{G}(0, \textswab{c}_{\epsilon} \Sigma_X^{-1}) $ and $\hat{\sigma}^b_n \overset{\mathbb{P}}{\longrightarrow} \sigma$.
\end{theorem}
\else
\begin{theorem}\label{thm:bootactual}
	(Bootstrap convergence.) Suppose that $ (X_n)_{n \in \mathbb{N}} $ and $ (\epsilon_n)_{n \in \mathbb{N}} $ satisfy Assumption \ref{ass:X}. Then conditional on $ \left(X_m, Y_m \right)_{m \in \mathbb{N}} $, for almost all sequences $ \left(X_m, Y_m \right)_{m \in \mathbb{N}} $, for each $ 1 \leq b \leq B, $ as $ n \rightarrow \infty$, 
	\begin{equation*}
		\sqrt{n}(\hat{\beta}_n^b -\hat{\beta}_n) \convd \mathcal{G}(0, \textswab{c}_{\epsilon} \Sigma_X^{-1})
	\end{equation*} 
	\begin{equation*}
		\text{ and }\quad \hat{\sigma}^b_n \overset{\mathbb{P}}{\longrightarrow} \sigma.
	\end{equation*}
\end{theorem}
\fi
\if\biometrika1
\vspace{0.2cm}
\fi
\noindent Using this result we obtain the following theorem.
\if\biometrika1
\vspace{0.08cm}
\fi
\begin{theorem}\label{thm:boottstat}
	(Bootstrap test-statistic convergence.) Suppose that $ (X_n)_{n \in \mathbb{N}} $ and $ (\epsilon_n)_{n \in \mathbb{N}} $ satisfy Assumption \ref{ass:X} and, for each $ 1 \leq b \leq B, $ let $ T^b_{n}:\mathcal{V} \rightarrow \mathbb{R} $ be the $ L $-dimensional random field on $ \mathcal{V} $ such that, for $ 1 \leq l \leq L$, $\displaystyle T_{n,l}^b = \frac{c_l^T(\hat{\beta}_n^b-\hat{\beta}_n)}{\hat{\sigma}_n^b\sqrt{ c_l^T(X_n^TX_n)^{-1}c_l}}. $
 Then conditional on $ \left(X_m, Y_m \right)_{m \in \mathbb{N}} $, for almost every sequence $ \left(X_m, Y_m \right)_{m \in \mathbb{N}} $, for each $ 1 \leq b \leq B, $ 
	\begin{equation*}
	T_n^b \convd \mathcal{G}(0, \textswab{c}')
	\end{equation*}
	as $ n \rightarrow \infty $. Here $ \textswab{c}': \mathcal{V} \times \mathcal{V} \rightarrow \mathbb{R} $ takes $ u,v \in \mathcal{V} $ to $ \textswab{c}'(u,v) = \rho_{\epsilon}(u,v) AC\Sigma_X^{-1}C^TA^T $ where $ A \in \mathbb{R}^{L \times L} $ is a diagonal matrix with $ A_{ll} = (c_l^T\Sigma_X^{-1}c_l)^{-1/2} $ for $ 1 \leq l \leq L$.
\end{theorem}
\vspace{0.08cm}
Crucially the limiting distribution in this result is the same as the limiting distribution of the test-statistics under the global null that $ \beta = 0 $, see Theorem \ref{thm:tstatconv}. It follows that the bootstrap provides consistent estimates of the quantiles of functionals of the data, see Theorem \ref{thm:consist}. Our proof of Theorem \ref{thm:bootactual} (see Section \ref{S:Lbootproof}) - which uses the Lindeberg Central Limit Theorem - is substantially simpler than existing proofs that we are aware of. Notably \cite{Eck2018} proved a version of Theorem \ref{thm:bootactual} by extending the results of \cite{Freedman1981} to multiple dimensions. Their approach, while interesting, is rather complex and relies on the fact that convergence in distribution is equivalent to convergence in the Mallows metric \citep{Bickel1981}.

\if\biometrika0

\subsection{Consistency of the bootstrap quantile}\label{SS:quantile}
When bootstrapping we use the bootstrap samples to estimate quantiles of the test-statistic under the null hypothesis. In what follows we will demonstrate that as the number of bootstraps and subjects tends to infinity, the derived quantiles converge to a limit. In order to do so given $ G:\mathbb{R}\rightarrow [0,1] $, define $ G^-:[0,1] \rightarrow \mathbb{R} $, which takes $ y \in [0,1] $ to $ G^{-}(y) = \inf\left\lbrace x: G(x) \geq y \right\rbrace $, to be the generalized inverse of $ G. $ Then we have the following theorem.
\vspace{0.08cm}
\begin{theorem}\label{thm:consist}
	Let $ (f_n)_{n \in \mathbb{N}}, f$ be functions from $\left\lbrace h: \mathcal{V} \rightarrow \mathbb{R}^L \right\rbrace$ to $\mathbb{R}$ such that conditional on $ (X_m, Y_m)_{m\in \mathbb{N}} $ for almost all sequences $ (X_m, Y_m)_{m \in \mathbb{N}} $, for each $ b \in \mathbb{N},$
	\begin{equation*}
	f_n(T^b_{n}) \convd f(\mathcal{G}(0, \textswab{c}')).
	\end{equation*}
	For each $ n, B \in \mathbb{N} $ and $ 0 < \alpha < 1 $, let 
	\begin{equation*}
	\lambda^*_{\alpha, n, B} = \inf\left\lbrace \lambda: \frac{1}{B} \sum_{b = 1}^B 1\left[ f_n(T^b_{n}) \leq \lambda \right] \geq \alpha \right\rbrace.
	\end{equation*}
	Take $ F $ to be the CDF of $ f(\mathcal{G}(0, \textswab{c}')) $ conditional on $ (X_m,Y_m)_{m \in\mathbb{N}} $, i.e. for $ \lambda \in [0,1] $, $ F(\lambda) = \mathbb{P}\left( f(\mathcal{G}(0, \textswab{c}')) \leq \lambda \middle| (X_m,Y_m)_{m \in \mathbb{N}}  \right) $ and assume that $ F $ is strictly increasing and continuous. Then, letting $ \lambda_{\alpha} = F^{-1}(\alpha) $, conditional on $ (X_m, Y_m)_{m\in \mathbb{N}} $ for almost all sequences $ (X_m, Y_m)_{m \in \mathbb{N}} $,
	\begin{equation*}
	\lim_{n \rightarrow \infty} \lim_{B \rightarrow \infty} \lambda^*_{\alpha, n, B} = \lambda_{\alpha}
	\end{equation*}
	almost surely.
\end{theorem}
\if\biometrika1
\begin{proof}
	Let $ F_n:\mathbb{R} \rightarrow [0,1] $ send $ \lambda \in \mathbb{R} $ to $\mathbb{P}\left( f(T^1_{n}) \leq \lambda\middle| (X_m,Y_m)_{m \in \mathbb{N}} \right)$.	Define a sequence $ (\eta_n)_{n \in \mathbb{N}} \geq 0 $ such that $ \alpha \pm \eta_n $ are continuity points of $ F_n^{-} $ and $ \eta_n \rightarrow 0 $ as $ n \rightarrow \infty $. To do so let $ \eta_n = 0 $ if $ \alpha $ is a continuity point of $ F_n^{-} $ and take $ \eta_n = \frac{1}{2n^n} $ otherwise. Note that there are at most $ n^n $ distinct values that $ f(T_n^1) $ can take and each value of $ T_n^1 $  is equally likely by construction. As such $F_n $ is a step function with at most $ n^n $ steps, meaning that the height difference between steps is at least $ \frac{1}{n^n}. $ The points of discontinuity of $ F_n^{-} $ are the values in the range of $ F_n$ and so if $ \alpha $ is not a point of continuity of $ F_n^{-} $ then $ \alpha \pm \frac{1}{2n^n} $ must be. Now, 
	\begin{equation}\label{eq:bound}
	\lambda^*_{\alpha - \eta_n, n, B} \leq \lambda^*_{\alpha, n, B} \leq \lambda^*_{\alpha + \eta_n, n, B}.
	\end{equation}
	
	The values $ \alpha \pm \eta_n $ are continuity points of $ F_n^{-} $ and for $ \lambda \in \mathbb{R} $, conditional on $ (X_m,Y_m)_{m \in \mathbb{N}} $, by the SLLN, $ \frac{1}{B} \sum_{b = 1}^B 1\left[ f(T^b_{n}) \leq \lambda\right] $ converges almost surely to $ F_n(\lambda) $ as $ B \rightarrow \infty $. As such, applying Lemma 1.1.1 from \cite{De2006}, $ \lambda^*_{\alpha \pm \eta_n, n, B} \rightarrow F_n^-(\alpha \pm \eta_n) $ almost surely as $ B \rightarrow \infty. $ Moreover, as $  n \rightarrow \infty,$ $F_n $ converges pointwise to $ F $ (as $ f(T^1_{n}) | (X_m,Y_m)_{m \in \mathbb{N}} \convd f(\mathcal{G}(0, \textswab{c}')))$ which is an increasing invertible function with continuous inverse.	As  such $ F_n^-(\alpha \pm \eta_n) \rightarrow \lambda_{\alpha}$ as $ n \rightarrow \infty. $ To see this note that for all $ \delta > 0, $ there exists $ N \in \mathbb{N} $ such that for all $ n \geq N, \eta_n < \delta$ and 
	\begin{equation*}
	F_n^{-}(\alpha) \leq F_n^{-}(\alpha + \eta_n) \leq F_n^-(\alpha + \delta)
	\end{equation*} 
	and $ F_n^-(\alpha + \delta) $ converges to $ F^{-1}(\alpha+ \delta) $ as $ n \rightarrow \infty $ by applying Lemma 1.1.1 from \cite{De2006} once again (since $ F^{-1} $ is continuous and so $ \alpha+ \delta$ is a point of continuity of $ F^{-1} $). Continuity of $ F^{-1} $ implies that $ F^{-1}(\alpha + \delta) \rightarrow F^{-1}(\alpha) $ as $ \delta \rightarrow 0 $. Arguing similarly for the sequence $ \alpha - \eta_n $ the result follows. Taking limits and using the bound in equation \eqref{eq:bound}, it almost surely follows that 
	\begin{equation*}
	\lim_{n \rightarrow \infty} \lim_{B \rightarrow \infty} \lambda^*_{\alpha, n, B} = \lambda_{\alpha}.
	\end{equation*}
	%
\end{proof}
\fi
We refer to $ \lambda^*_{\alpha, n , B} $ as the $ \alpha $-quantile of the bootstrap distribution of $ f_n(T_n) $ based on $ B $ bootstraps. By Theorem \ref{thm:bootactual}, for suitable choices of $ f_n, f $,  $ f_n(T_n^b) \convd f(\mathcal{G}(0, \textswab{c}')) $ as $ n \rightarrow \infty $ and so Theorem \ref{thm:consist} shows that the bootstrapped $ t $-statistics can be used to provide consistent estimates of the quantiles of the limiting distribution of $ f(\mathcal{G}(0, \textswab{c}')) $. The easiest example of such a suitable sequence of functions is to take $ f $ to be continuous and let $ f_n = f $ for all $ n \in \mathbb{N}$, because of the Continuous Mapping Theorem. A more general result that provides a sufficient condition, based on uniform convergence of $ f_n $ to $ f, $ is given in Lemma \ref{lem:unifconv}.
\fi

\section{Joint error rate control in the linear model}\label{S:LMJER}
In this section we will state and prove our main results. We will show that given $ 0 < \alpha < 1 $, choosing $ \lambda $ to be the $ \alpha $-quantile of the bootstrapped distribution results in asymptotic $ 1-\alpha $ level control of the joint error rate and thus results in simultaneous control of the FDP.

\subsection{Joint error rate control}\label{SS:JERcontrol}
To set this up, given a test-statistic $ T:\mathcal{V} \rightarrow \mathbb{R}^L $, a subset $ H \subset \mathcal{H}  $, and $ n \in \mathbb{N} $, for $ 1 \leq k \leq \left| H \right| $, let $ p^n_{(k:H)}(T) $ be the $ k $th minimum value in the set 
\begin{equation*}
\left\lbrace 2 - 2\Phi_{n-r_n}(\left| T_l(v) \right|): (l,v) \in H \right\rbrace.
\end{equation*}
Using the results we have proved so far we can obtain the following theorem.
\if\biometrika1
\vspace{0.08cm}
\fi
\begin{theorem}\label{thm:JERcontrol}
	For $ H \subset \mathcal{H}  $, let $ f_{n,H}: \left\lbrace g: \mathcal{V} \rightarrow \mathbb{R}^L \right\rbrace \rightarrow \mathbb{R} $ send  
\begin{equation*}
	T \mapsto \min_{1 \leq k \leq K \wedge \left| H \right|} t_k^{-1}(p^n_{(k:H)}(T)) 
\end{equation*}
	and for
	\if\biometrika1
	$ n, B \in \mathbb{N} $ and $\alpha \in (0,1)$, let  $\lambda^*_{\alpha,n,B}(H) = \inf\left\lbrace \lambda: \frac{1}{B} \sum_{b = 1}^B 1\left[ f_{n,H}(T^b_{n}) \leq \lambda \right] \geq \alpha \right\rbrace$.
	\else
	$ n \in \mathbb{N} $ and $\alpha \in (0,1)$ let $ \lambda^*_{\alpha,n,B}(H) $ be the $ \alpha $-quantile of the bootstrap distribution (based on $ B \in \mathbb{N} $ bootstraps) of $ f_{n,H}(T_n) $ conditional on the observed data.
	\fi
	Assume that the conditions of Assumption \ref{ass:X} hold and that $ r_n = o(n) $.	Then for all $ H \subset \mathcal{H} $ such that $ \mathcal{N} \subset H$, 
	\begin{equation*}
	\lim_{n \rightarrow \infty} \lim_{B \rightarrow \infty} \mathbb{P}\left( f_{n, \mathcal{N}}(T_n) \leq \lambda^*_{\alpha,n,B}(H) \right) \leq \alpha.
	\end{equation*} 
	The limit holds with equality if $ H = \mathcal{N} $. In particular, taking $ H = \mathcal{H}  $, it follows that
	\begin{equation*}
	\lim_{n \rightarrow \infty} \lim_{B \rightarrow \infty} \mathbb{P}\left( \min_{1 \leq k \leq K\wedge \left| \mathcal{H}  \right|} t_k^{-1}(p^n_{(k:\mathcal{N})}(T_n)) \leq \lambda^*_{\alpha,n,B}(\mathcal{H} ) \right) \leq \alpha
	\end{equation*}
\end{theorem}
Applying this result and using Claim  $ \ref{lem:JERSIM} $ we are thus able to obtain asymptotic control of the joint error rate of the canonical reference family. Following the discussion in Section \ref{SS:JER} this means that we obtain asymptotic post hoc FDP control. In particular we having the following corollary.
\if\biometrika1
\vspace{-0.35cm}
\fi
\begin{corollary}\label{cor:posthoc}
	Under the assumptions of Theorem \ref{thm:JERcontrol}, for $ 0 < \alpha < 1, $ and $ H \subset \mathcal{H} $, let 
	\begin{equation*}
	\mybar{V}_{\alpha, n,B}(H) = \min_{1 \leq k \leq K} (|H \setminus R_k(\lambda^*_{\alpha,n,B}(\mathcal{H} ))| + k-1) \wedge |H|.
	\end{equation*}
	\begin{equation*}
	\text{Then }
	\lim_{n \rightarrow \infty}\lim_{B \rightarrow \infty}\mathbb{P}\left(|H \cap \mathcal{N}| \leq \mybar{V}_{\alpha, n,B}(H),\,\, \forall H \subset \mathcal{H}  \right) \geq 1-\alpha.
	\end{equation*}
\end{corollary}
Thus in order to provide FDP control, given a number of bootstraps $ B \in \mathbb{N} $, we can calculate $\lambda^*_{\alpha,n,B}(H) $, the $ \alpha $-quantile of the bootstrap distribution of $ f_{n,H}(T_n) $ conditional on the observed data. Then $\mybar{V}_{\alpha, n,B}(H) $ provides a $ (1-\alpha) $ level simultaneous upper bound on the number of false positives in $ H \subset \mathcal{H} $.
\subsection{Bootstrap step-down procedure}\label{SS:stepdown}
It is possible to improve on the power of the above procedure by taking a step-down approach in the spirit of  \citep{Romano2005}. This is based on the idea that joint error rate control implies familywise error rate control, see Section \ref{A:FWER}. As such it is possible to obtain an estimate of the set of null hypotheses and thereby obtain a tighter bound. The procedure, which adapts the step-down procedure of \cite{Blanchard2020} to our setting, can be iterated as follows.

\if\biometrika1
\vspace{0.08cm}
\begin{algo}\label{Alg:stepdown}
Step-down bootstrap
\begin{tabbing}
   \enspace \textbf{Set} $ j \leftarrow 0 $ and $ H_n^{(0)} \leftarrow \mathcal{H} $\\
   \enspace \textbf{Repeatedly set} $ j \leftarrow j + 1 $, $ \lambda_{n,j} \leftarrow \lambda^*_{\alpha, n, B}(H_n^{(j-1)}) $ and  $ H_n^{(j)} \leftarrow \left\lbrace (l,v): p_{n,l}(v) \geq t_1(\lambda_{n,j}) \right\rbrace $ \\
\qquad \enspace \textbf{Until} $ H_n^{(j)}  =  H_n^{(j-1)} $\\
\enspace \textbf{Return} $ \hat{H}_n = H_n^{(j)} $
\end{tabbing}
\end{algo}
\else
\begin{algorithm}[H]
	\caption{step-down bootstrap}\label{Alg:stepdown}
	\begin{algorithmic}[1]
		\State{Set $ j \leftarrow 0 $ and $ H_n^{(0)} \leftarrow \mathcal{H} $}
		\Repeat
		\State{Set $ j \leftarrow j + 1 $, $ \lambda_{n,j} \leftarrow \lambda^*_{\alpha, n, B}(H_n^{(j-1)}) $ and  $ H_n^{(j)} \leftarrow \left\lbrace (l,v): p_{n,l}(v) \geq t_1(\lambda_{n,j}) \right\rbrace $}
		\Until{ $ H_n^{(j)}  =  H_n^{(j-1)} $ }
		\State{Set $ \hat{H}_n \leftarrow H_n^{(j)} $} and 
		\Return{ $ \hat{H}_n $}
	\end{algorithmic}
\end{algorithm}
\fi

As the following theorem demonstrates, the step-down approach controls the joint error rate and therefore provides simultaneous FDP control.
\if\biometrika1
\vspace{0.08cm}
\fi
\begin{theorem}\label{thm:stepdown}
	Under the assumptions of Theorem \ref{thm:JERcontrol}, for $ 0 < \alpha < 1, $ let $ \hat{H}_n $ be the set generated by applying Algorithm \ref{Alg:stepdown}. Then 
	\begin{equation*}
	\lim_{n \rightarrow \infty} \lim_{B \rightarrow \infty} \mathbb{P}\left( f_{n, \mathcal{N}}(T_n) < \lambda^*_{\alpha,n,B}(\hat{H}_n) \right) \leq \alpha.
	\end{equation*}
	Thus, for $ H \subset \mathcal{H}, $ letting $ \mybar{V}_{\alpha, n,B}(H) =  \min_{1 \leq k \leq K} (|H \setminus R_k(\lambda^*_{\alpha,n,B}(\hat{H}_n))| + k - 1) \wedge  |H| , $ it follows that 
	\begin{equation*}
	\lim_{n \rightarrow \infty}\lim_{B \rightarrow \infty}\mathbb{P}\left(|H \cap \mathcal{N}| \leq \mybar{V}_{\alpha, n,B}(H),\,\, \forall H \subset \mathcal{H}  \right) \geq 1-\alpha.
	\end{equation*}
\end{theorem}

In the definition of $ \lambda^*_{\alpha, n, B} $ we require the computation of $ |\mathcal{H}| $ statistics for each bootstrap each of which is based on a sample of size $ n $. As such the complexity of these algorithms is $O(nB|\mathcal{H}|)$.

\begin{remark}
	The results in this subsection and the one previous have been stated for two-sided $ p $-values however they also hold for one-sided $ p $-values without change. All that is required to show this is to re-define $ p^n_{(k:H)}(T) $ as the $ k $th minimum value in the set 
	\begin{equation*}
	\left\lbrace 1 - \Phi_{n-r_n}(T_l(v)): (l,v) \in H\right\rbrace.
	\end{equation*}
\end{remark}
\if\biometrika0
\subsection{Parametric Approaches}\label{SS:parametric}
\label{sec:param-approaches}
In this section we will discuss two parametric approaches to simultaneous FDP inference which are based on the Simes inequality \eqref{eq:simesinequality}.
\if\biometrika0
Here we use the term parametric to indicate that dependency assumptions on the data are required in order for the methods to be valid.
\fi
The first one is the original Simes post hoc bound introduced in \cite{Goeman2011}.
The second one is the method of \cite{Rosenblatt2018} and  \cite{Goeman2019}. 
It corresponds to an improvement on the basic Simes bound that is adaptive to the proportion of true null hypotheses - i.e. it is a step-down version of the Simes bound.
This method has been applied to brain imaging data in \cite{Rosenblatt2018}, and is called \textbf{ARI} which stands for ``\textbf{All Resolutions Inference}''. Both methods can be conveniently formulated in terms of the bound $\overline{V}$ defined in \eqref{eq:Vbar}, associated to the linear template family $(t_k)_{1 \leq k \leq m}$, where $t_k(\lambda)=\lambda k/m$, i.e. 
\begin{align}
\label{eq:Simes-bound}
\overline{V}_\lambda(S)&= \min_{1 \leq k \leq m}\left \{\sum_{ i \in S} 1\left[ p_i \geq \frac{\lambda k}{m}  \right]  + k-1\right\} \,.
\end{align}
As noted by \cite{Blanchard2020}, the Simes post hoc bound of \cite{Goeman2011} is simply $\mybar{V}_\alpha$. 
Moreover, letting $\overline{\alpha} = \alpha m/ h(\alpha)$, where 
\begin{equation*}
h(\alpha) = \max \left\{i \in \{1, \dots, m\}, \forall j \in \{1, \dots , i\}, p_{(m-i+j)} > \frac{\alpha j}{i}\right\}, 
\end{equation*}
the ARI bound of \cite{Goeman2019} is $\overline{V}_{\overline{\alpha}}$. 
The quantity $h(\alpha)$ is called the Hommel factor \citep{Hommel1988} and can be interpreted as a $(1-\alpha)$-level upper confidence bound on $|\mathcal{N}|$, the number of true null hypotheses.

If the null $ p $-values satisfy positive regression dependence then both of these methods result in simultaneous $ (1-\alpha) $-level FDP control. This is shown formally in \cite{Goeman2011} and \cite{Goeman2019} via closed testing and can also be shown to hold by combining the Simes inequality with the joint error rate framework of Section \ref{SS:JER}. To see this note that if the null $ p $-values are positive regression dependent \citep{Sarkar2008}, then the Simes inequality is satisfied, that is: 
\begin{align}\label{eq:simesinequality}
  \mathbb{P}\left(\exists 1 \leq k \leq \left| \mathcal{N} \right|: p^n_{(k:\mathcal{N})} \leq \frac{\alpha k}{|\mathcal{N}|} \right) \leq \alpha,
\end{align}
with equality if the null $ p $-values are independent.

In particular taking $ \lambda = \alpha $, and noting that $ \left| \mathcal{N} \right| \leq m $, the Simes inequality implies that \eqref{eq:probJER} holds (taking $K = m$ and $(t_k)_{1 \leq k \leq m} $  to be the linear reference family). Moreover, \cite{Goeman2019}'s Lemma 2 implies that if the null $ p $-values satisfy positive regression dependence, then
\begin{align}\label{eq:stepdownsimes}
\mathbb{P}\left(\exists 1 \leq k \leq \left| \mathcal{N} \right|: p^n_{(k:\mathcal{N})} \leq \frac{\alpha k}{h(\alpha)} \right) \leq \alpha.
\end{align}
Thus taking $ \lambda = \overline{\alpha} = \alpha m/ h(\alpha)$, it follows that \eqref{eq:probJER} holds with respect to the linear reference family. In particular the Simes procedure, which uses $ \overline{V}_{\alpha} $ as a bound, and ARI, which uses $ \overline{V}_{\overline{\alpha}} $, provide simultaneous $ (1-\alpha) $-level control of the FDP.

%

%

%
\if\biometrika0
In our results, presented in the following sections, we compare the performance of the non-parametric bootstrap approach to these parametric alternatives.
\fi

\fi

\section{Simulation Results}\label{SS:simulations}
\subsection{Simulation Setup}\label{SS:simsetup}
In order to assess empirically that our method correctly controls the joint error rate we run numerical simulations. We create the noise in these simulations by generating 2-dimensional stationary Gaussian random fields on domains which are $ 25 $ by $25$, $50$ by $50$ and $100$ by $100$ pixels. To do so we smooth Gaussian white noise with a Gaussian kernel with full width at half maximum (FWHM) in $ \left\lbrace 0, 4, 8 \right\rbrace $ (in pixel units), accounting for edge effects to ensure stationarity (see e.g. \cite{Davenport2022Ravi}), and scaled so that the variance is 1 everywhere.

We let the total number of subjects $ n $ range from $ 20 $ to $ 100 $. For each $ n $, smoothness level, image size and 
\if\biometrika0
$ \pi_0 \in \left\lbrace 0.5, 0.8, 0.9, 1 \right\rbrace$,
\else
$ \pi_0 \in \left\lbrace 0.5, 0.8, 1 \right\rbrace$,
\fi
we run 5000 simulations - each with 100 bootstraps - to test the joint error rate. For each simulation we do the following. First we generate $ n $ Gaussian random fields $ \epsilon_1, \dots, \epsilon_n $ as described above and add signal to them (as detailed in the next paragraph). We then randomly divide these images into 3 disjoint groups: $ G_1, G_2, G_3 \subset \left\lbrace 1, \dots, n \right\rbrace $ - performing assignment to each group with equal probability (we eliminate assignments where a given group has no entries). We test for the difference between the first and the second group and between the second and the third group - giving us two contrasts to differentiate between. We thus, in total, test $ 5000 $ hypotheses for the 50 by 50 scenario and $20000$ for the 100 by 100 case.

We vary the amount of signal in the datasets as follows. Given a proportion $ \pi_0 $ we randomly choose a subset $ \mathcal{N} $ of size $ \pi_0\left| \mathcal{H} \right| $ of $ \mathcal{H} = \left\lbrace (l, v): 1 \leq l \leq 2, v \in \mathcal{V}\right\rbrace $ to be null (which is thus different in each simulation) and add signal to ensure that the remainder are non-null. To do so, for $ 1 \leq i \leq n, $ and each $ v \in \mathcal{V} $, we set 
\begin{equation*}
Y_i(v) = 1[i \in G_2, (1,v) \not\in \mathcal{N}]  + 1[i \in G_3, (1,v) \not\in \mathcal{N}] + 1[i \in G_3, (2,v) \not\in \mathcal{N}] + \epsilon_i(v).
\end{equation*}
This ensures that the power of the test to detect a difference at $ h $ is equal for any $ h \in \mathcal{N}^C $. If $ \pi_0 = 1 $ then all hypotheses are null. An example realisation is shown in Figure \ref{simimages}.

In the next subsections we compare our bootstrap procedures, in terms of false positive control and power, to two parametric alternatives: the Simes procedure \citep{Goeman2011} and its step-down: all resolutions inference (ARI, \cite{Rosenblatt2018}) which are described in Section \ref{sec:param-approaches}. \if\biometrika1
Here we use the term parametric to indicate that dependency assumptions on the data are required in order for the methods to be valid.
\fi

\subsection{False Positive control}\label{JERFPR}
In each simulation setting, for $ 1 \leq j \leq 5000 $, we calculate a test statistic random field $ T_{n}^{(j)} $ and obtain $ \lambda $ thresholds for the single-step bootstrap, step-down bootstrap, Simes and ARI methods, where we have used 100 bootstraps for the non-parametric procedures. For each method we obtain $ \lambda $-thresholds $ \lambda_1, \dots, \lambda_{5000} $ allowing us to estimate the joint error rate via the statistic 
\if\biometrika1
$ \frac{1}{5000}\sum_{j = 1}^{5000} 1[f_{n, \mathcal{N}}(T_{n}^{(j)}) \leq \lambda_j] $
\else
\begin{equation*}
	\frac{1}{5000}\sum_{j = 1}^{5000} 1[f_{n, \mathcal{N}}(T_{n}^{(j)}) \leq \lambda_j]
\end{equation*}
\fi
which we refer to as the \textbf{empirical joint error rate}. Here $ 1[\cdot] $ denotes the indicator function. 

The results for the $50$ by $50$ simulations are displayed in Figure \ref{fig:fpr50} and those for the other domain sizes are shown in Figures  \ref{fig:fpr25} and \ref{fig:fpr100}.
The results for the bootstrap methods are shown in 
\if\biometrika1
black 
\else
blue
\fi
whilst those for the parametric methods are shown in 
\if\biometrika1
grey.
\else
red.
\fi
The solid lines indicate the step-down methods (i.e. ARI and the step-down bootstrap). 
These plots demonstrate that, given a reasonable number of subjects, the joint error rate of the bootstrap procedures converges to the nominal level, in this case 0.1. 
Empirically the parametric procedures are valid in all settings considered.
However, their control of the joint error rate is substantially below the nominal level i.e. when the applied smoothing is non-zero, while the bootstrap approaches demonstrate tighter control.
The step-down procedures provide an improvement on their single-step counterparts. 
This difference increases as $ \pi_0 $ decreases. See Section \ref{SS:power} for further details on the effect of $ \pi_0 $.
\begin{figure}[h!]
	\begin{subfigure}[t]{0.32\textwidth}
		\centering
		\includegraphics[width=\textwidth]{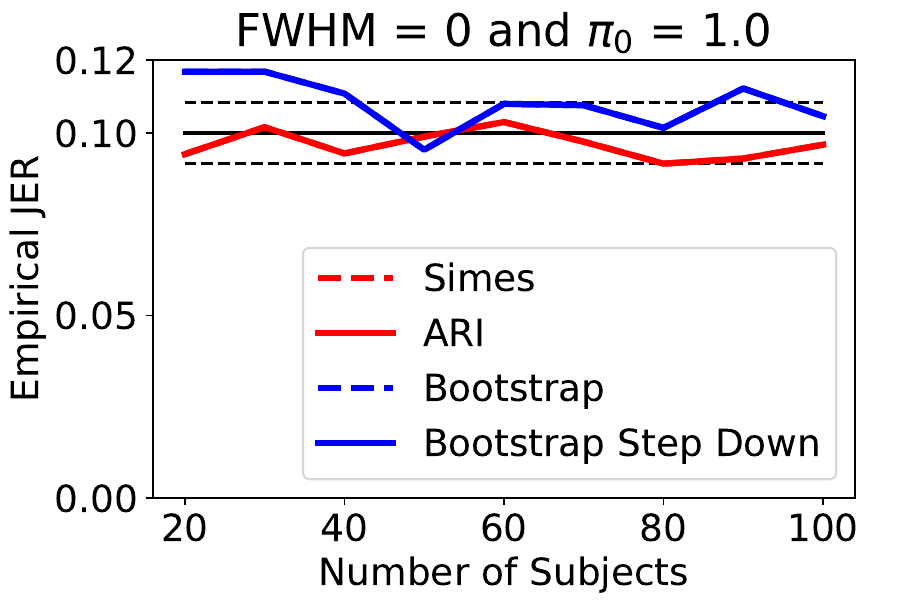}
	\end{subfigure}
	\hfill
	\begin{subfigure}[t]{0.32\textwidth}  
		\centering 
		\includegraphics[width=\textwidth]{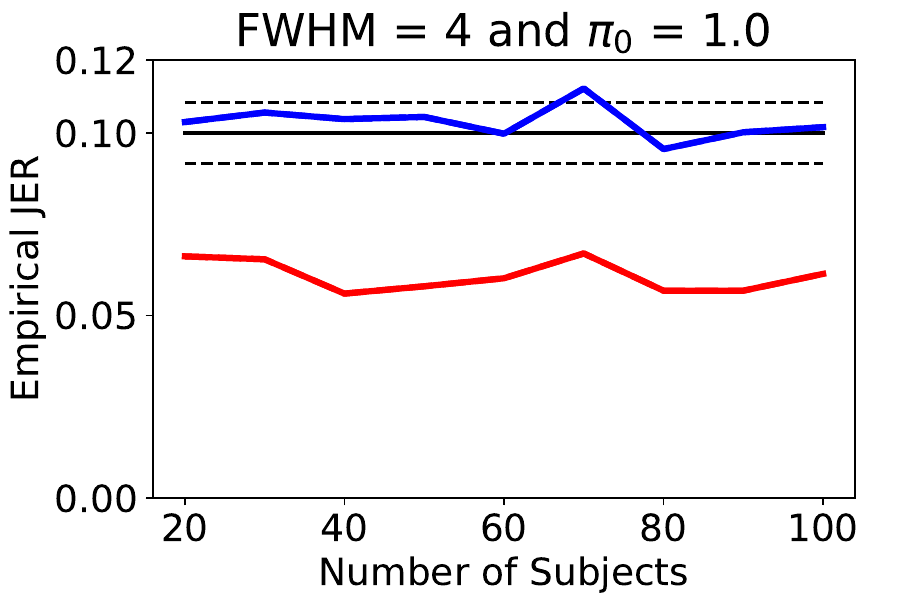}
	\end{subfigure}
	\hfill
	\begin{subfigure}[t]{0.32\textwidth}  
		\centering
		\includegraphics[width=\textwidth]{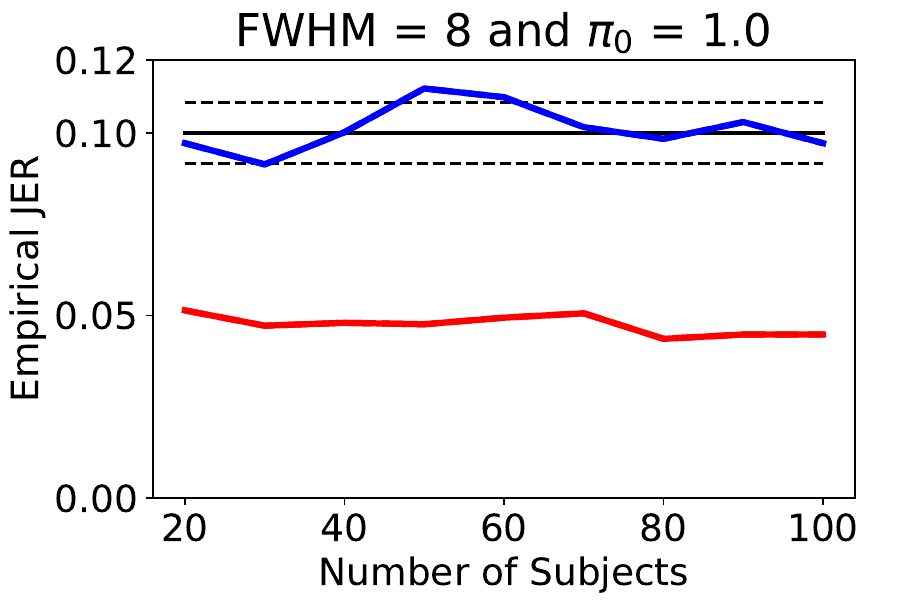}
	\end{subfigure}
\if\biometrika0
	\begin{subfigure}[t]{0.32\textwidth}
		\centering
		\includegraphics[width=\textwidth]{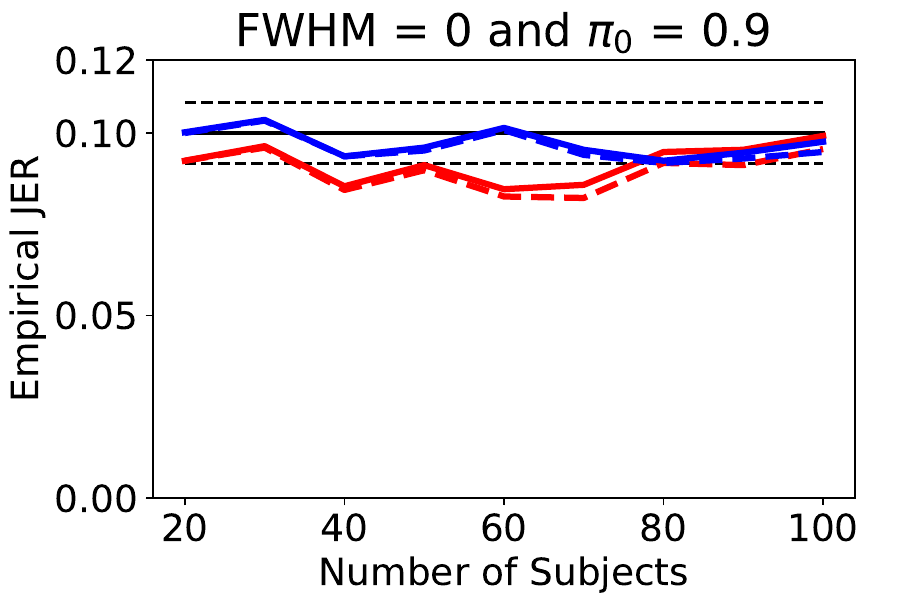}
	\end{subfigure}
	\hfill
	\begin{subfigure}[t]{0.32\textwidth}  
		\centering 
		\includegraphics[width=\textwidth]{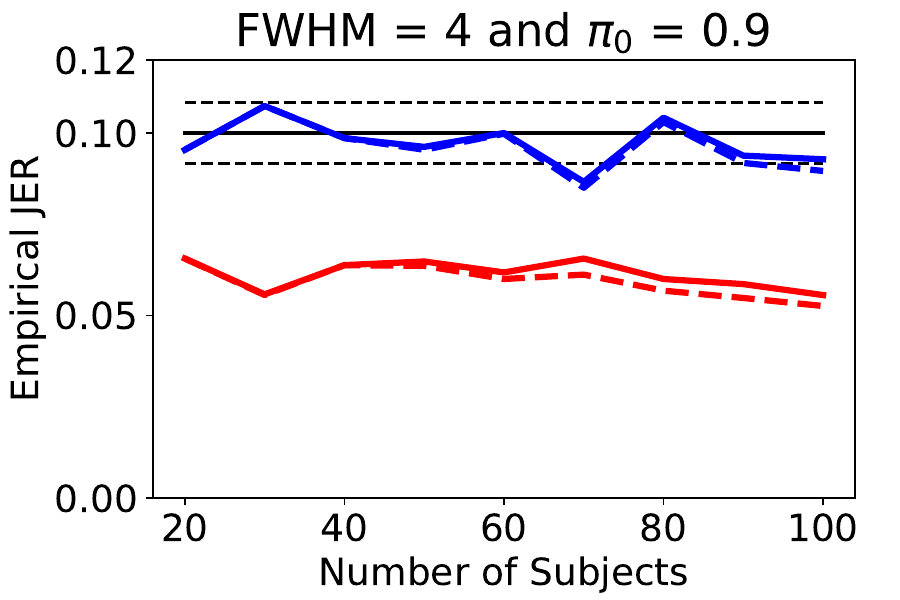}
	\end{subfigure}
	\hfill
	\begin{subfigure}[t]{0.32\textwidth}  
		\centering
		\includegraphics[width=\textwidth]{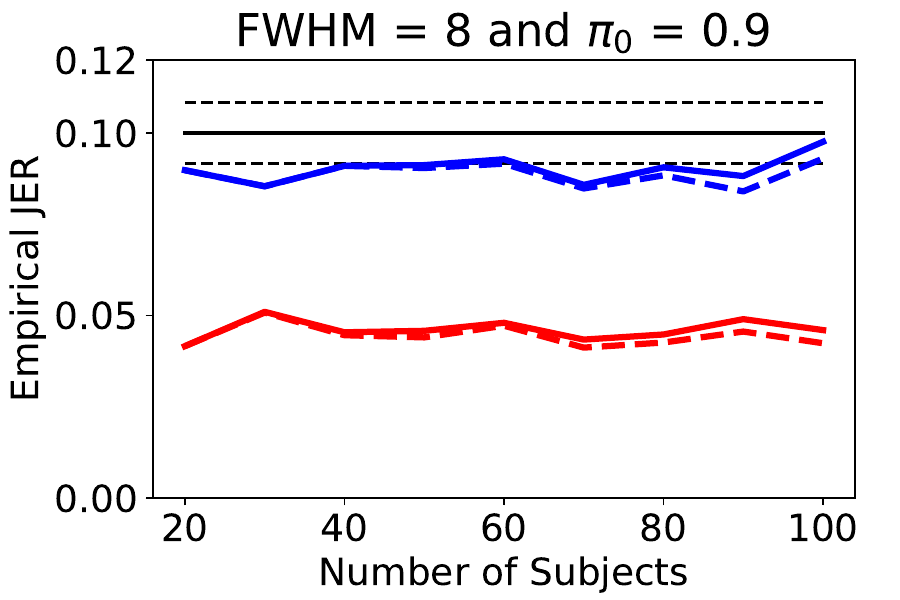}
	\end{subfigure}
\fi
	\begin{subfigure}[t]{0.32\textwidth}
		\centering
		\includegraphics[width=\textwidth]{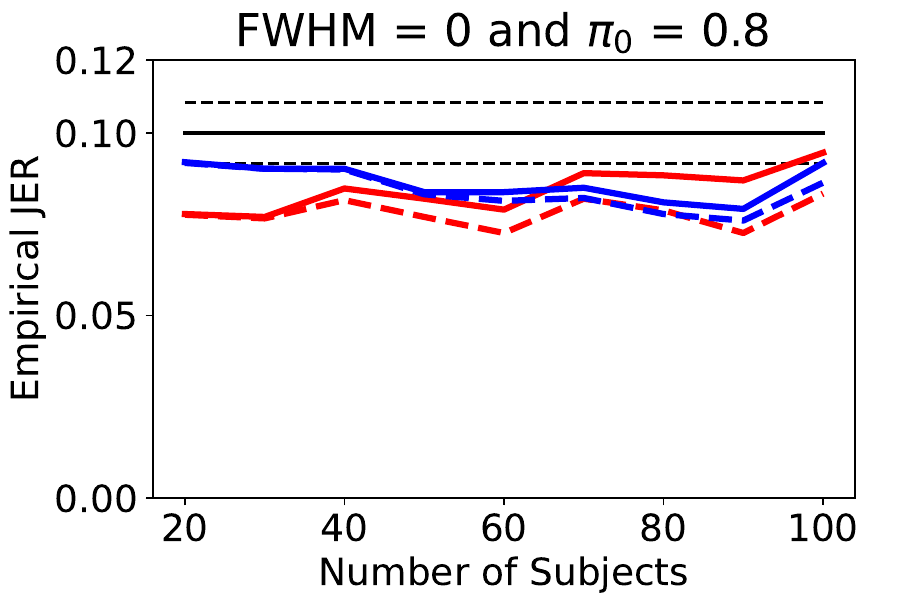}
	\end{subfigure}
	\hfill
	\begin{subfigure}[t]{0.32\textwidth}  
		\centering 
		\includegraphics[width=\textwidth]{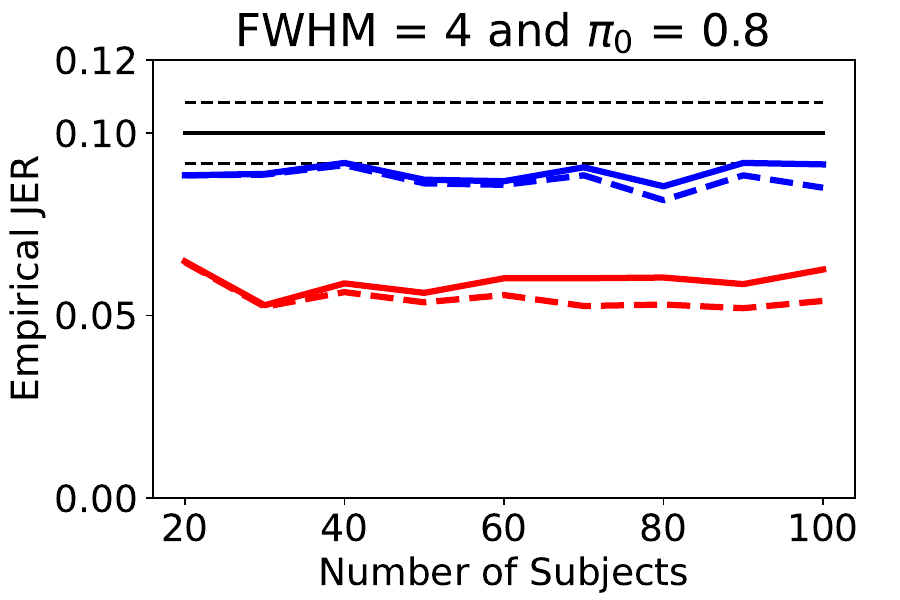}
	\end{subfigure}
	\hfill
	\begin{subfigure}[t]{0.32\textwidth}  
		\centering
		\includegraphics[width=\textwidth]{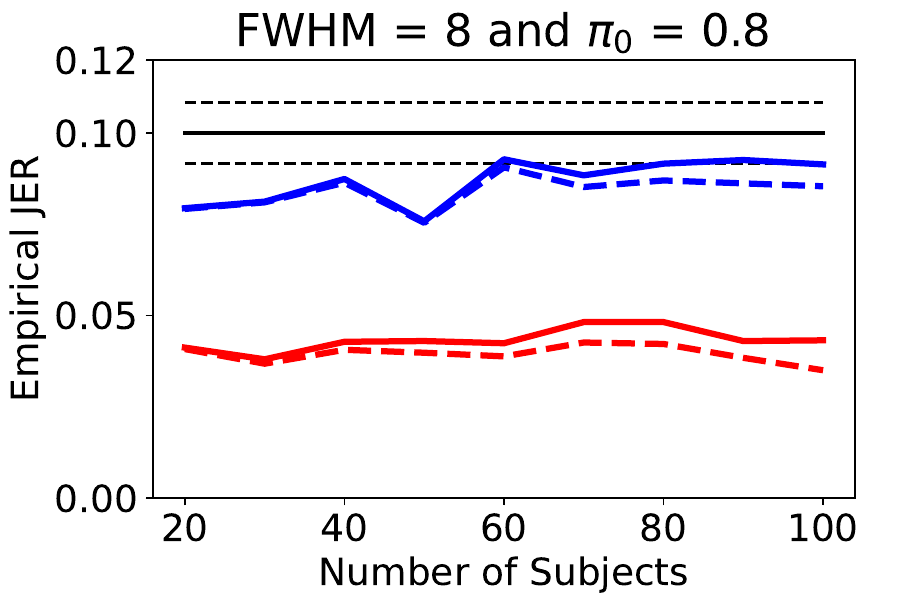}
	\end{subfigure}
	\begin{subfigure}[t]{0.32\textwidth}
		\centering
		\includegraphics[width=\textwidth]{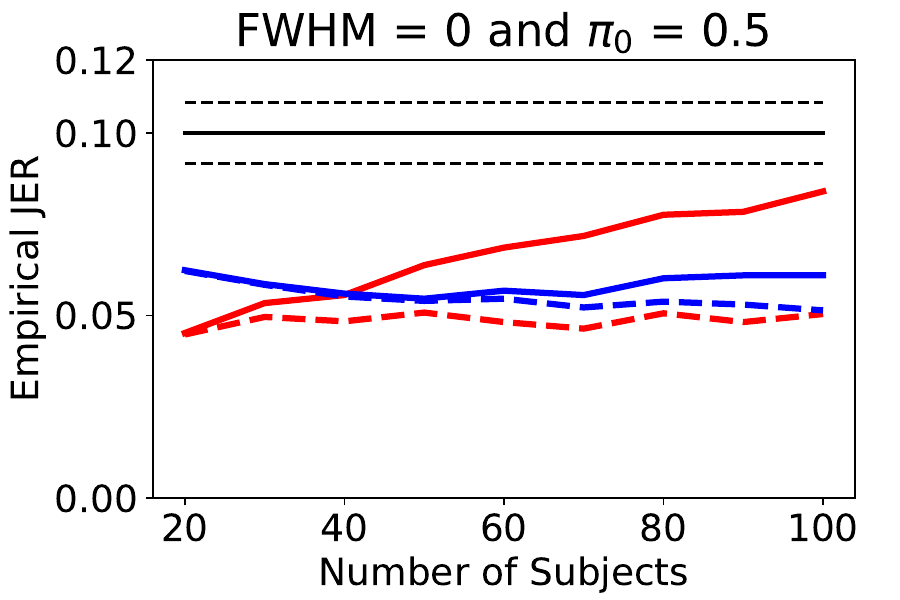}
	\end{subfigure}
	\hfill
	\begin{subfigure}[t]{0.32\textwidth}  
		\centering 
		\includegraphics[width=\textwidth]{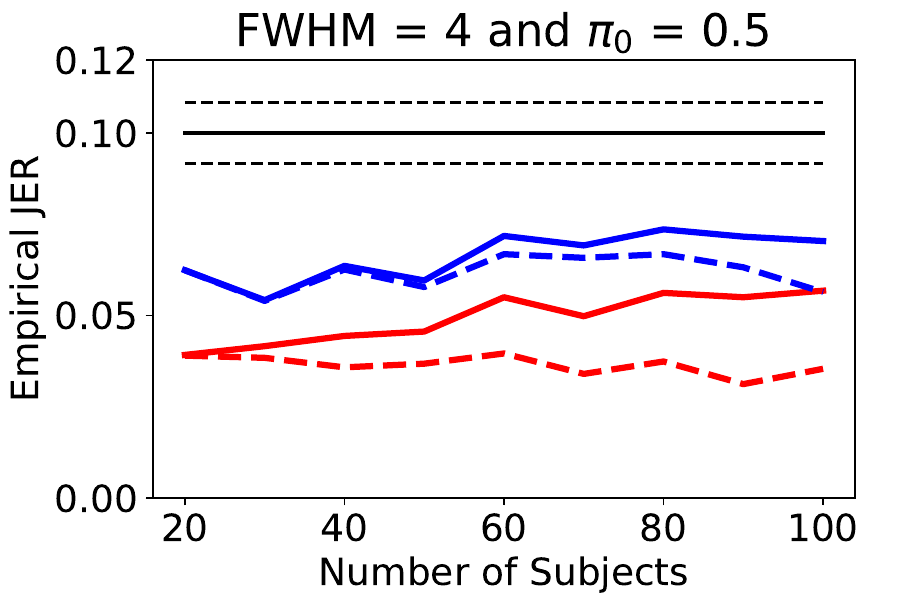}
	\end{subfigure}
	\hfill
	\begin{subfigure}[t]{0.32\textwidth}  
		\centering
		\includegraphics[width=\textwidth]{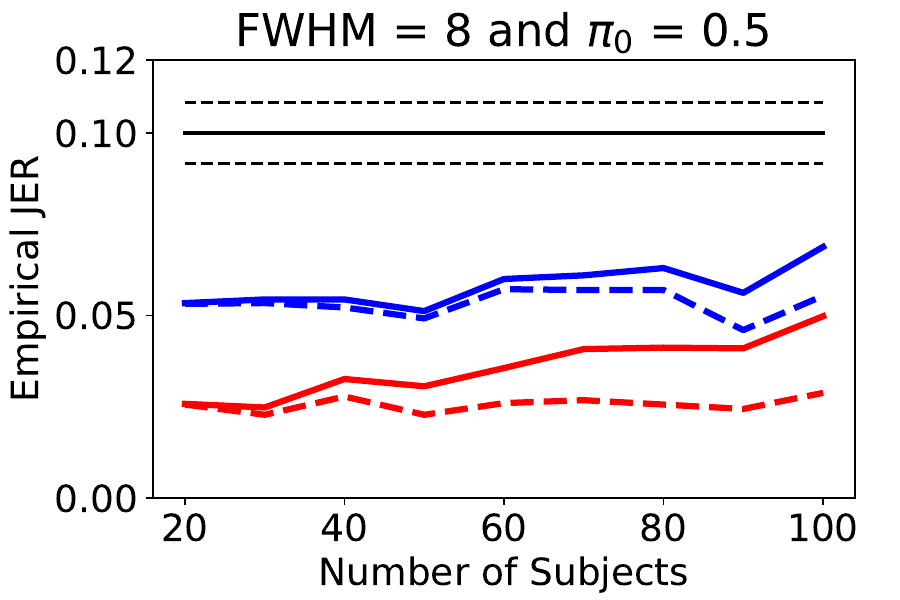}
	\end{subfigure}
	\caption{
		\if\biometrika1
		Comparing the empirical joint error rate across methods for the simulation setting described in Section \ref{SS:simsetup} for $ \alpha = 0.1 $ on the domain of size $50$ by $50$ pixels. The bootstrap procedures typically provide tighter control of the joint error rate than the parametric ones, except under independence. The bootstrap methods are shown in black whilst the parametric methods are shown in grey. The solid lines indicate the step-down methods. The thin flat black dashed lines provide 95\% marginal confidence bands.
		\else
		Comparing the empirical joint error rate across methods for the simulation setting described in Section \ref{SS:simsetup} for $ \alpha = 0.1 $ on the domain of size $50$ by $50$ pixels. The bootstrap procedures typically provide tighter control of the joint error rate than the parametric ones, except under independence. The bootstrap methods are shown in
		blue whilst the parametric methods are shown in red. The solid lines indicate the step-down methods. The thin flat black dashed lines provide 95\% marginal confidence bands based on the normal approximation to the binomial distribution.
		\fi
	}\label{fig:fpr50}
\end{figure}

\subsection{Power}\label{SS:power}
In this section we compare the power of the various methods in the simulation setting described in Section \ref{SS:simsetup} in the case where the applied FWHM is 4 pixels. We have chosen to focus on this level of smoothness because it represents a realistic level of applied smoothness and illustrates the benefits that can be achieved when using the bootstrap under dependence.

Here we shall use a notion of power originally proposed in \cite{Blanchard2020} to compare the ability of joint error rate controlling procedures to detect signal. Given a set $ R \subset \mathcal{H} $, define
\begin{equation}\label{eq:power}
\text{Pow}(R) := \mathbb{E}\left[ \frac{|R| - \mybar{V}(R)}{\left| R \cap (\mathcal{H} \setminus \mathcal{N}) \right|} \middle| \left| R \cap (\mathcal{H} \setminus \mathcal{N}) \right| > 0\right]
\end{equation}
where for each method $ \mybar{V} $ is the corresponding post-hoc bound. Here we consider the following choices of $ R $ with which we compare the power (as in \cite{Blanchard2020}). 1) $ R = \mathcal{H} $ and 2) taking $ R $ to be the hypotheses of $ \mathcal{H} $ which are rejected by the Benjamini Hochberg procedure, applied to the $ p $-values $ \left\lbrace p_{n,l}(v): (l, v) \in \mathcal{H}  \right\rbrace $, at a level $ 0.05 $. Note that, unlike in \cite{Blanchard2020}, no additional level of randomness in the choice of the sets in 2) is prescribed. We also consider taking $ R = \left\lbrace (l, v): p_{n, l}(v) \leq 0.05  \right\rbrace$, see Section \ref{S:app}, the results for which are similar in nature to scenario 1 from above. The results for  cases 1) and 2) are illustrated graphically in Figure \ref{fig:power}. These are for simulations on the $50$ by $50$ domain.

\begin{figure}[h!]
	\begin{subfigure}[t]{0.32\textwidth}
		\centering
		\includegraphics[width=\textwidth]{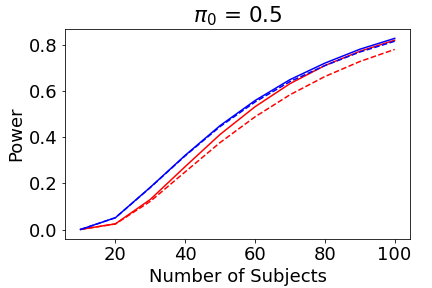}
	\end{subfigure}
	\hfill
	\begin{subfigure}[t]{0.32\textwidth}  
		\centering 
		\includegraphics[width=\textwidth]{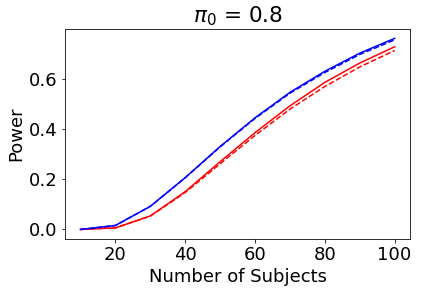}
	\end{subfigure}
	\hfill
	\begin{subfigure}[t]{0.32\textwidth}  
		\centering 
		\includegraphics[width=\textwidth]{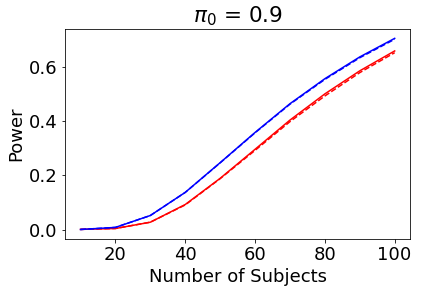}
	\end{subfigure}
	\begin{subfigure}[t]{0.32\textwidth}  
		\centering 
		\includegraphics[width=\textwidth]{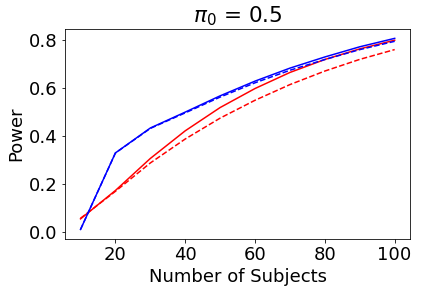}
	\end{subfigure}
	\begin{subfigure}[t]{0.32\textwidth}
		\centering
		\includegraphics[width=\textwidth]{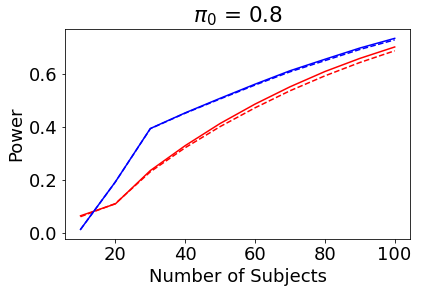}
	\end{subfigure}
	\hfill
	\begin{subfigure}[t]{0.32\textwidth}  
		\centering 
		\includegraphics[width=\textwidth]{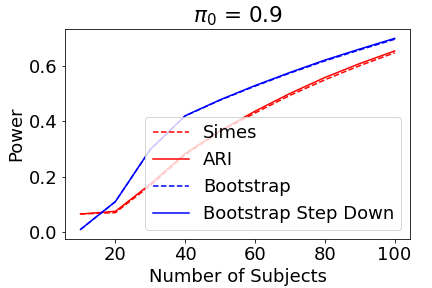}
	\end{subfigure}
	\hfill
	\caption{Plotting the power of the different methods against the number of subjects. The power for setting 1 (i.e. $ R $ = $ \mathcal{H} $) is shown in the top row and the power for setting 2 (i.e. taking $ R $ to be the Benjamini-Hochberg rejection set) is shown in the bottom row.}\label{fig:power}
\end{figure}

From these plots we can see that overall the bootstrap based approaches have a higher power than the parametric ones. The power of ARI only becomes comparable (or higher) to that of the bootstrap in the extreme scenario ($ \pi_0 = 0.5 $) given a large enough sample size. Additionally the bootstrap is not robust at the smallest sample size considered (i.e. $ n = 10 $) where it is slightly conservative. However it is important to note that in typical high-dimensional applications (neuroimaging, genetics)  $\pi_0 > 0.9$ and $ n $ is often substantially greater than $ 20. $

The lower the value of $ \pi_0 $, the greater the increase in power that is obtained by using the step-down algorithms. ARI is always more powerful than Simes by construction. In the relatively sparse scenarios (i.e. $ \pi_0 \geq 0.8 $) they have a very similar power however for $ \pi_0 = 0.5 $, ARI provides a marked improvement over Simes. The bootstrap step-down always improves on the standard bootstrap approach though the difference is not particularly large: even when $ \pi = 0.5 $ this increase is relatively small. The similarity of the standard and step-down procedures, for both the parametric and bootstrap methods, is consistent with the results obtained on real data which are described in the next subsections.


\section{Real Data results}
\subsection{Neuroimaging data application}\label{SS:neuroapp}
We have 3D functional Magnetic Resonance Imaging data from $ n = 386 $ unrelated subjects, who performed an $ m $-back working memory task, from the Human Connectome Project. After pre-processing (described in Section \ref{S:fmridataprocessing}) we obtain a 3-dimensional contrast image for each subject. We fit a linear model to these images including sex, height, weight, body mass index, two different measures of blood pressure, handedness and IQ (measured using the PMAT24\_A\_CR test score). We consider sex and IQ as a variables of interest. We obtain test-statistic contrasts for sex and IQ and a $ p $-value at each voxel for each contrast. We form clusters using a cluster defining threshold on the $ p $-values of $ p = 0.001 $, with each cluster being a contiguous set of voxels above the threshold (clusters are defined separately for each contrast of interest).

We use our bootstrap framework, performing the resampling using 1000 bootstraps, to provide a lower bound on the proportion of active voxels within each cluster, taking $ \alpha = 0.1$. This illustrates that multiple clusters, in different regions of the brain, have a relatively large proportion of active voxels  for the contrast of IQ. For the contrast of sex only a single cluster has a non-zero lower bound on the number of true positives. The bounds provided using the step-down bootstrap procedure are the same as the single-step version in this example.

We compare to the results that are obtained using Simes and ARI bounds (taking $\alpha = 0.1$) and see that our bootstrap approach results in higher lower bounds on the number of active voxels. In this setting the bounds obtained by the parametric procedures are very similar to each other, which is not surprising given the sparsity of the signal.
\if\biometrika0
For the IQ contrast the lower bounds provided by the bootstrap and ARI for the number of true positives and on the TDP within each cluster are shown graphically in the upper panel of Figure \ref{fig:realdata}. The corresponding plot for the sex contrast is shown in Section \ref{S:sexcontrast}. 
\else
The results are shown graphically (along with a comparison to the parametric results) in Figures \ref{fig:realdata} and \ref{sexcontrast}. 
\fi
\if\biometrika1
Direct comparison of the lower bounds, between the different methods, on the true discovery proportion and the number of true positives within each cluster is shown in Figure \ref{fig:tdphcpplots}.
\else
Direct comparison of the lower bounds is shown in Figure \ref{fig:tdphcpplots}.
\fi

\if\biometrika0
\begin{figure}[h]
	\begin{center}
		\includegraphics[width=\textwidth]{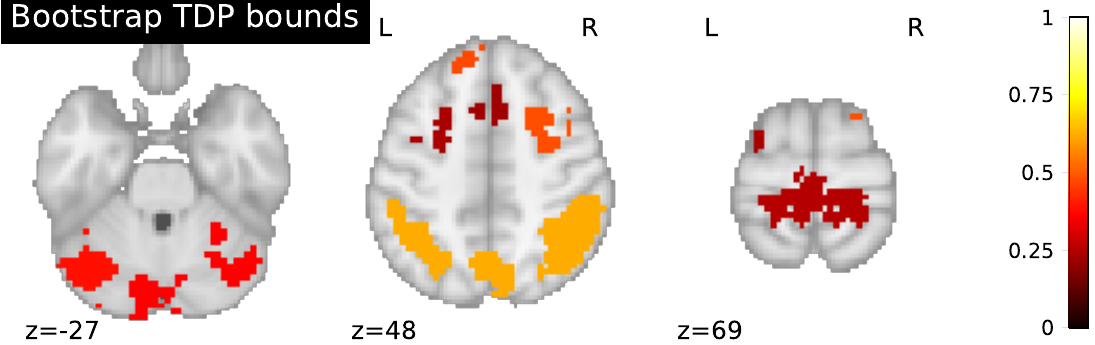}
		\includegraphics[width=\textwidth]{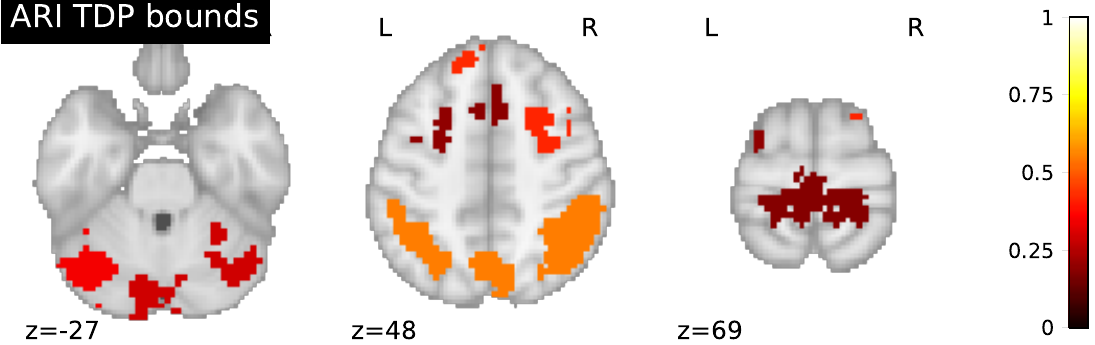}
	\end{center}
	\caption{TDP bounds within clusters for the contrast for IQ in the linear regression model fit to the HCP data. Each cluster is shaded a single colour which is the lower bound on the TDP. The upper panel gives the TDP bounds within each cluster provided by the bootstrap procedure. The lower panel gives the bounds provided by using ARI. The bounds given by the bootstrap are larger (as indicated by the lighter colours) indicating that the method is more powerful. Note that these images are 2D slices through the 3D brain and so voxels that are part of the same cluster are not necessarily connected.}\label{fig:realdata}
\end{figure}
\fi

\begin{figure}[h]
	\begin{center}
		\includegraphics[width=0.48\textwidth]{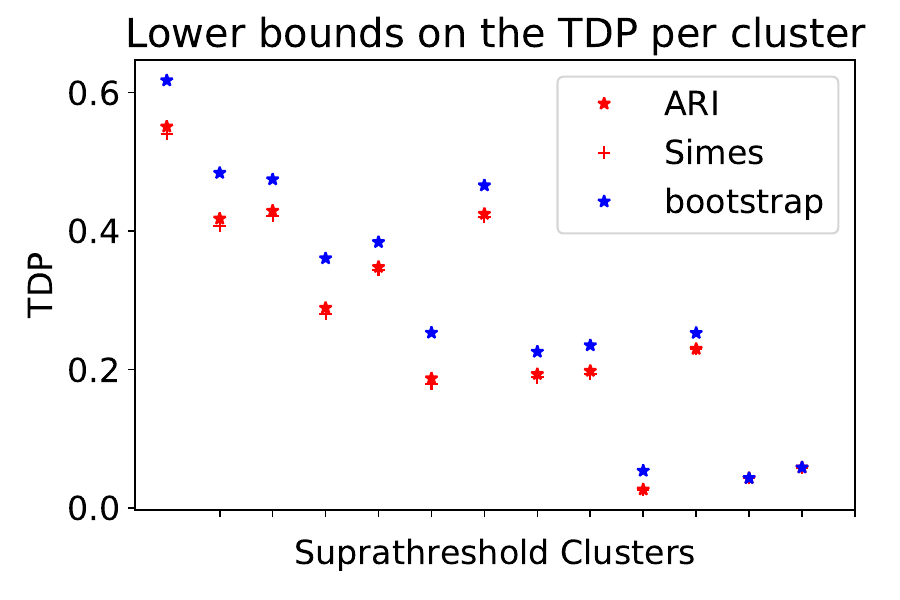}
		\includegraphics[width=0.48\textwidth]{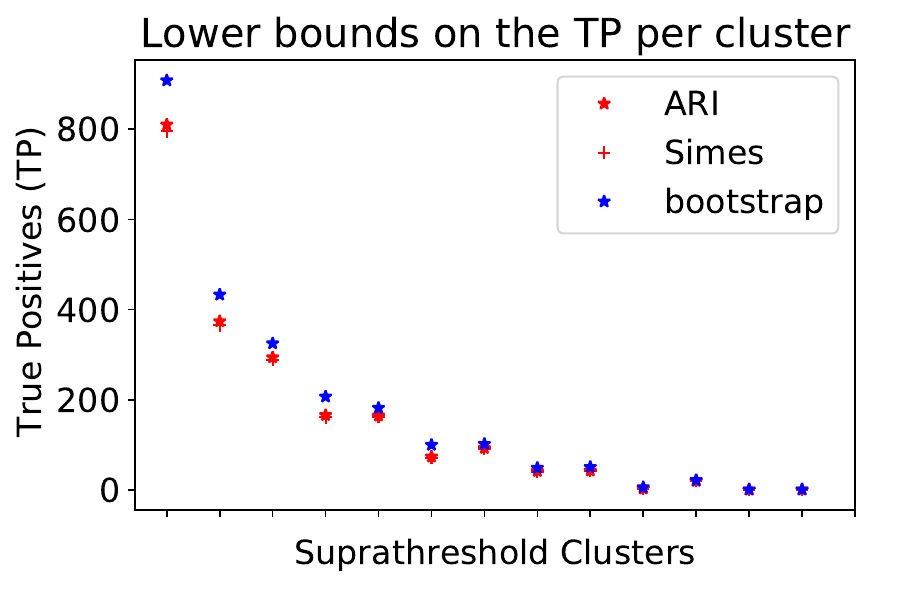}
	\end{center}
	\caption{Comparing the TDP and true positive lower bounds across clusters for the different methods. The bootstrap lower bounds are consistently higher than the parametric methods. Clusters are organized from left to right in terms of their size. Only one cluster for the sex contrast is found: this is the 2nd smallest cluster overall with a TP lower bound of 1 voxel. The sizes and bounds of the clusters in the IQ contrast are larger. For the largest cluster we are able to conclude that it contains 908 true positives using the bootstrap approach.}\label{fig:tdphcpplots}
\end{figure}

\subsection{Transcriptomic data application}\label{SS:trans}
In this section, we illustrate the application of our methods to a specific gene expression data set. 
Gene expression studies use microarray or sequencing biotechologies in order to measure the activity (or ``expression level'') of a large number of genes simultaneously. 
We focus on a study of chronic obstructive pulmonary disease (COPD), see \cite{bahr2013peripheral}, whose main goal was to identify genes whose expression level is significantly associated with lung function. 
In order to do this, the authors fit a linear model for this association for each gene, while controlling for the following covariates: age, sex, body mass index, parental history of COPD, and two smoking variables (smoking status and pack-years).
The number of subjects is $n = 135$ while the number of genes is $V = 12,531$, leading to a large-scale multiple testing problem. 
Using the Benjamini-Hochberg method to control the FDR at the $5\%$ level, 1,745 genes were found to be significantly associated.

We fit this linear model to the data, regressing the gene level data against the controlled covariates and lung function and considering a single contrast for lung function. We performed 1000 bootstraps and used these to obtain $ \lambda^*_{\alpha, 135, 1000} = 0.22 $, where we took $ \alpha = 0.1. $ This allows us to provide a $ (1-\alpha) $-level simultaneous lower bounds on the number of true positives within any specified set of genes. In particular it allows us to conclude (with $90\%$ confidence) that at least 1,354 of the 1,745 genes within the Benjamini-Hochberg significance set are active. The stepdown bootstrap provides the same bound as the single-step version in this case. Simes and ARI provide lower bounds on the number of true positives in this set of 917 and 966 respectively, which are substantially less informative than the bootstrap bounds.

In the absence of prior information on genes, a natural idea is to rank them by decreasing statistical significance. 
Our post hoc methods provide upper confidence curves on the proportion of true positives among the most significant genes.  
Such curves are displayed in Figure \ref{fig:fdpplot}, where the 
\if\biometrika1
black
\else
blue
\fi
 lines correspond to our proposed single-step and step-down bootstrap-based methods, and the\if\biometrika1
grey
\else
red
\fi lines correspond to the parametric approaches of \cite{Goeman2011} and \cite{Rosenblatt2018}.
These results are consistent with the numerical experiments of Section \ref{SS:simulations}.
First, the bootstrap method yields post hoc bounds that are substantially more informative that their parametric counterpart.
Second, the difference between single-step methods and their step-down counterpart is very small, which is consistent with the fact that the signal is expected to be sparse in such genomic data sets, corresponding to $\pi_0$ close to $1$. For the bootstrap there is in fact no difference between the single-step and step-down approach in this example.

A widely used approach in differential expression studies is to select genes based on the conjunction of a threshold on the $p$-values and a threshold on its effect size \citep{Cui2003}.
\cite{Ebrahimpoor2020} recently noted that this type of double selection can lead to inflated numbers of false discoveries when used in conjunction with FDR-based multiple testing corrections, whereas post hoc inference is by construction robust against this issue.
The use of our proposed post hoc bounds in this context is illustrated in the volcano plot in Figure \ref{fig:volcano} \citep{Cui2003}.
In this plot, each gene is represented in two dimensions by estimates of its effect size ($x$ axis, also kown as ``fold change'' in genomics) and p-value ($y$ axis), in a logarithmic scale.
Figure~\ref{fig:volcano} illustrates a particular selection, corresponding to the genes whose $p$-value is below $0.001$ and whose effect size is above $0.5$.
Our bootstrap-based bound  ensures that with probability $1-\alpha = 90\%$, among these $546$ genes, at least $490$ are true positives, corresponding to a FDP below $0.1$. 
Importantly, the $p$-value and effect size thresholds can be chosen post hoc, and multiple such choices can be made without compromising the statistical coverage of the associated bound.
For example, the bounds associated to the gene subsets with positive and negative effect size are also displayed in Figure~\ref{fig:volcano}. 

\begin{figure}[h]
	\centering\includegraphics[width=0.75\textwidth]{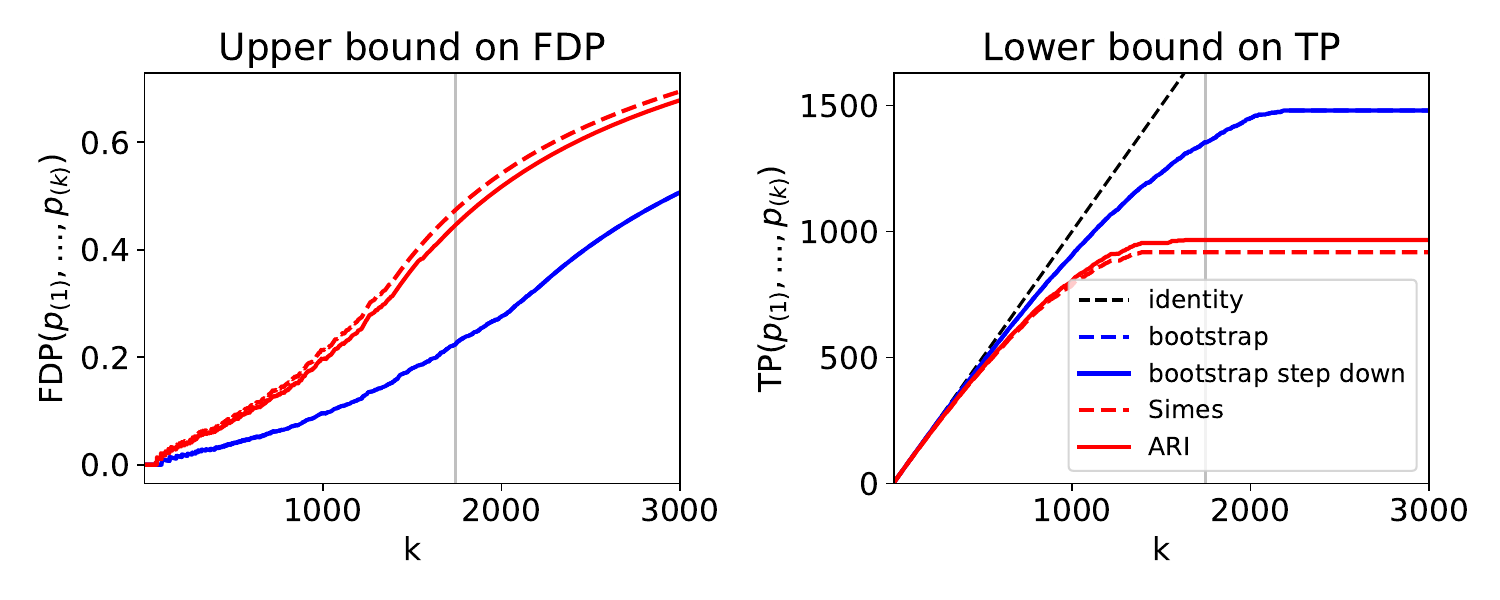}
	\if\biometrika0
	\centering\includegraphics[width=0.75\textwidth]{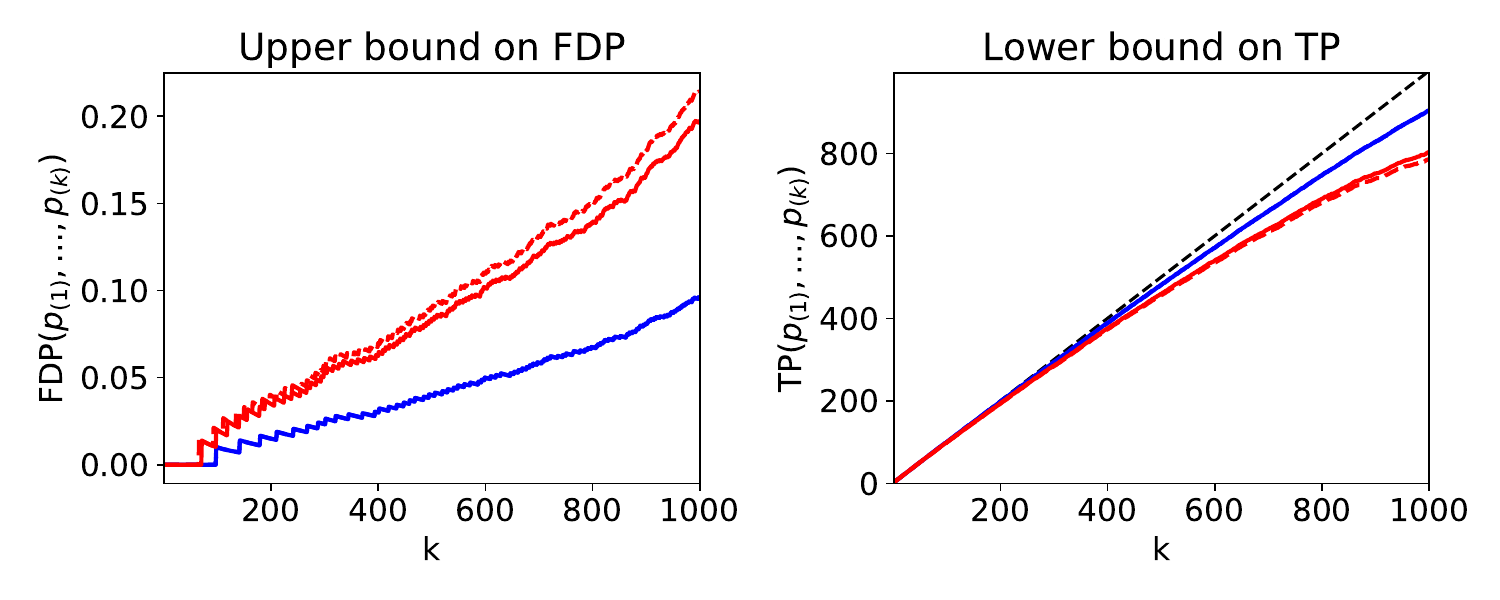}
	\fi
	\caption{False discovery proportion and true positive plots for the transcriptomic dataset. In the upper panels, for $ k = 1, \dots, 3000 $, upper bounds on the FDP and lower bounds on the number of true positives are provided by each of the methods for the sets comprised of the hypotheses with the $ k $ smallest $ p $-values. The silver vertical line corresponds to the location of the Benjamini-Hochberg rejection set. 
	\if\biometrika0
	The lower panels provide a zoomed in version of the same plot for for the $ 1000 $ smallest $ p $-values.
	\fi
	The bootstrap methods provide substantially better bounds than the parametric ones. ARI slightly improves on Simes while the step-down bootstrap is indistinguishable from the single-step bootstrap approach in this setting. 
	\if\biometrika1
	See Figure \ref{fig:lowerpanel} for a zoomed in version.
\fi}\label{fig:fdpplot}
\end{figure}

\begin{figure}[h!]
	\centering
	\if\biometrika1
	\includegraphics[width=0.8\textwidth]{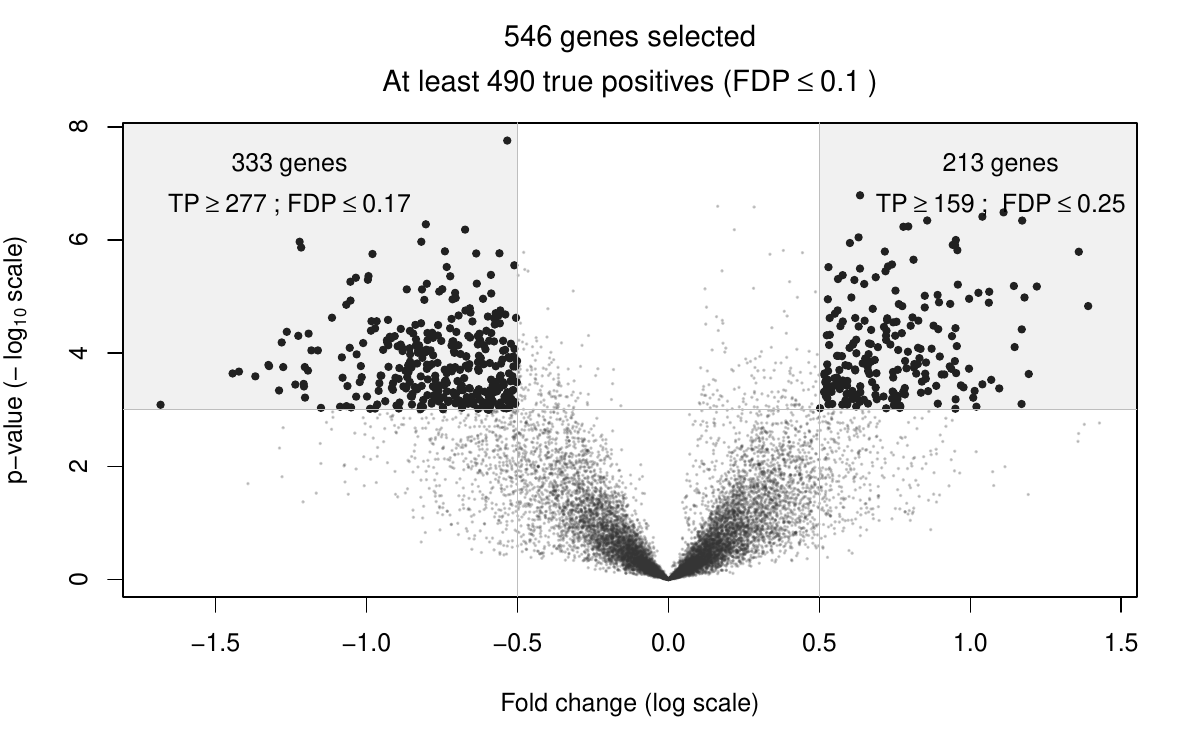}
	\else
	\includegraphics[width=0.8\textwidth]{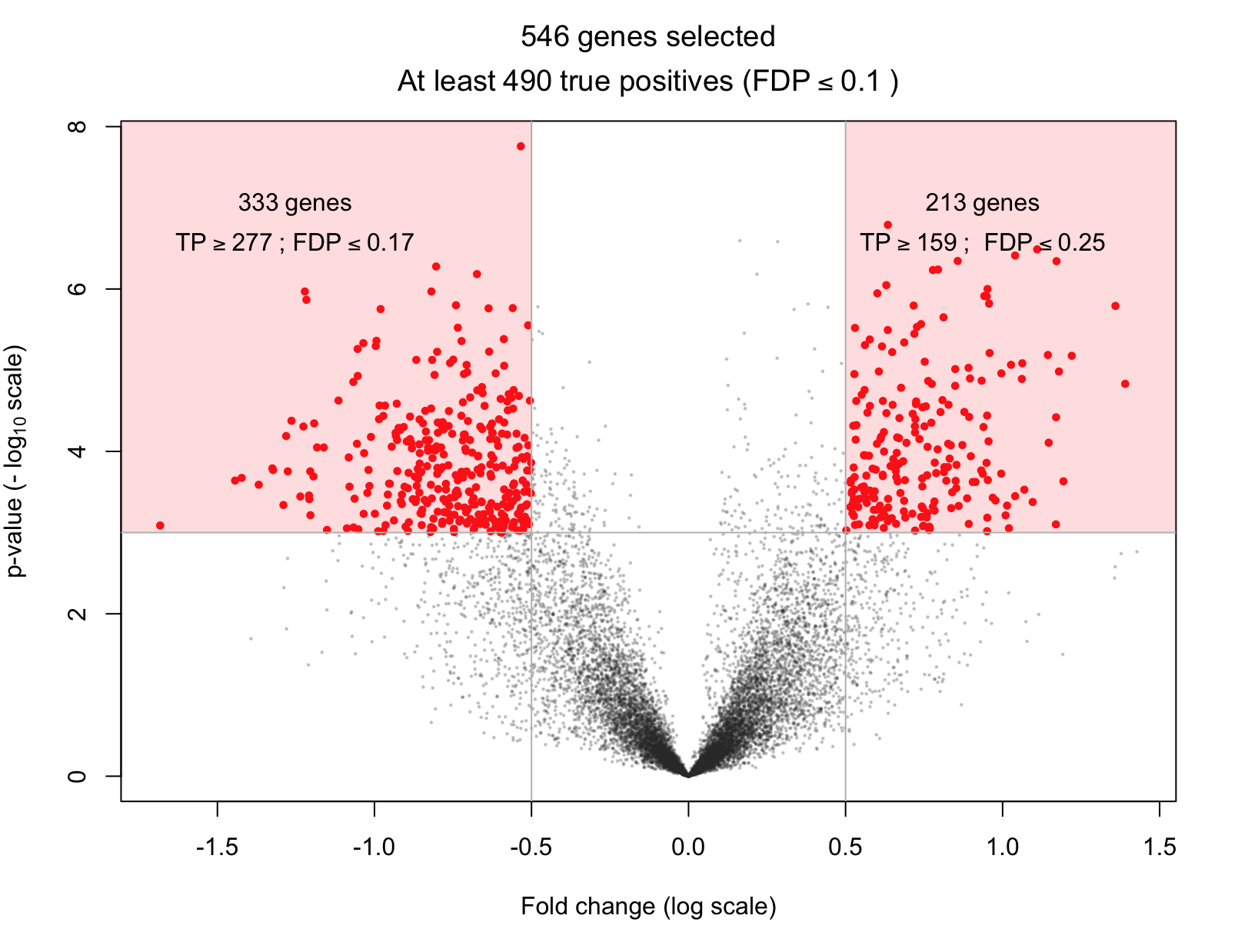}
	\fi
	\caption{A volcano plot for the $ p $-values for the transcriptomic data. For each gene this plots the estimated contrast effect size (labelled as fold change and corresponding to $ c^T\hat\beta_{135} $ where $ c$ is the contrast vector for COPD) against the $ p $-value, where both are measured in log scale. Two regions (shown via shading) are selected containing the genes whose $ p $-values are less than $ 10^{-3} $ and for which the absolute fold change is greater than $ 10^{0.5} $. Bounds on the true positives (TP) and FDP, overall and for the shaded regions are provided on the plot.}\label{fig:volcano}
\end{figure}

\if\biometrika1
\vspace{-0.5cm}
\fi
\section{Discussion}\label{S:discussion}
In this paper we have introduced a bootstrap method which provides simultaneous control of the FDP over subsets of hypotheses of multiple contrasts in the linear model. We have proved the asymptotic validity of this approach and shown, via simulation, that the error rate is controlled to the correct level given a reasonable number of subjects. 
\if\biometrika1
Our methods allow statements about the number of active voxels with a given set to be made. This inference is valid simultaneously over all sets and so guards against circular inference.
\fi

From our simulations and real data examples, we can see that the bootstrap approach typically provides better bounds than existing, state of the art parametric methods (i.e. Simes and ARI). This occurs because we are able to model the dependence within the data. The parametric methods, on the other hand, rely on the Simes inequality which is only exact under independence. Moreover the Simes inequality is only valid under positive regression dependence whereas the non-parametric bootstrap makes relatively few assumptions other than finite moments of the noise and the design. Moreover in real data situations there is typically relatively strong dependence within the data and so we would expect the bootstrap to give better bounds. This is illustrated in our brain imaging and transcriptomic examples where the bootstrap bounds provided substantial improvements over the ones derived using the parametric methods. 
\if\biometrika1

Further discussion is available in Section \ref{S:FD}. This includes a discussion of the nature of the improvements provided by the step-down approaches and motivations for the choice of template family. There we also show, via theory and simulations, that our results can be used to provide FWER control in the multiple contrast setting.
\fi
\if\biometrika0
Overall, our results are consistent with those obtained by~\cite{Enjalbert-Courrech} in the specific case of two-sample tests, where post hoc bounds based on non-parametric joint error rate control substantially outperformed their parametric counterpart.

The step-down bootstrap approach improves the power whilst maintaining control of the error rate. However in practice, as illustrated in the real datasets, the improvement is likely to be small as $ \pi_0 $ will be close to 1. Indeed in both of our real data examples there is no noticeable improvement. The improvement of ARI over Simes is typically non-zero but is rather small.  These results demonstrate that the step-down methods, whether parametric or otherwise, appear to require a relatively small value of $ \pi_0 $ before they substantial improvement on their single step versions. It is worth noting that the improvement in the bounds provided by ARI relative to Simes is greater than that of the step-down bootstrap relative to the single-step version. One possible reason for this discrepancy is that in the step-down bootstrap (Algorithm~\ref{Alg:stepdown}), only the first threshold $t_1$ is used at each step. This implies that only part of the information on the true positives is exploited.

It is important to note that for the bootstrap approach, control of the FDP is asymptotic. In Section \ref{JERFPR} we showed that given reasonable sample sizes and smoothness levels (e.g. $ n \geq 20 $ for FWHM = $ 4, 8 $) the joint error rate was controlled at the correct rate. At low smoothness levels and low sample sizes the error rate can be slightly inflated. In this scenario this inflation is counter balanced by a small amount of signal. Moreover at very small sample sizes e.g. $ n = 10 $ the bootstrap can be conservative (see Figure \ref{fig:power}). At the smoothness levels and sample sizes used in real data analyses, based on our theoretical and simulation results, we would expect the bootstrap to control the joint error rate to the desired level.

The template $(t_k)_{1 \leq k \leq K}$ is a free parameter of the proposed method, and optimising this choice for a particular application could lead to tighter TDP bounds. 
For the numerical experiments of this paper, we have only considered the linear template, for which $t_k(\lambda) =  \frac{\lambda k}{m}$, which is the most widely used in the post hoc inference literature~\citep{Goeman2019} and in particular for neuroimaging applications~\citep{Rosenblatt2018}. 
Other parametric templates are considered in \cite{Blanchard2020,Andreella2020}.
However, the experiments reported in \cite{Andreella2020} suggest that the linear template may be difficult to beat. 
A natural idea to go beyond parametric templates is to learn from the data the shape of the template itself.
This has been advocated by \cite{Meinshausen2006}, but the proof of the proposed method is invalid as it suffers from a circularity issue \cite[Remark 5.3]{Blanchard2020}. 
Recently, \cite{Blain2022} used an independent data set to learn the optimal template in the case of one- and two-sample testing.
Extending this idea to multivariate linear models as considered in the present paper is an interesting perspective for future research. 

Our proof of bootstrap consistency in the linear model is simpler than existing proofs of the same result. It takes advantage of the Lindeberg CLT and the fact that the bootstrapped $ \hat{\beta}s $ can be expanded as a sum which has nice asymptotic properties. The idea of using the Lindeberg CLT has been used to establish bootstrap consistency for the non-parametric bootstrap (\cite{Vander1998}, \cite{Kosorok2003}). It has also been used to prove results for high-dimensional linear models \citep{Mammen1993} and for univariate robust regression \citep{Shorack1982}. However, as far as we are aware, our work is the first to use it to establish consistency results for the multivariate residual bootstrap.

Our results can also be used to provide strong control of the familywise error rate over multiple contrasts (see Appendix \ref{A:FWER} for a proof, a formal discussion of this and the results of simulations). This comes about in two different ways. Firstly it arises as a direct consequence of joint error rate control when using the canonical reference family and taking $ \zeta_k = k-1 $. 
Secondly the familywise error rate can be targeted directly, along the lines of the approach of \cite{Westfall2011} (i.e. not simultaneously with joint error rate control), this follows from Theorem \ref{thm:JERcontrol} by taking $ K = 1 $, a result that is stated formally in Theorem \ref{thm:fwer}. There are a variety of methods to provide FWER control in the linear model but most of them are not suitable for the case where some of the variables of interest can take non-zero values. For instance Manly based permutation \citep{Manly1986} provides weak, rather than strong control when there are multiple covariates in the model which may or may not be non-zero. This occurs because Manly permutation acts by permuting the $ Y $s and thus does not generate resamples under the full null hypothesis - see Section \ref{A:permNwork} for details. Freedman-Lane \citep{Freedman1983} - another commonly used method - encounters similar issues. The bootstrap is able to avoid these issues, and thus provide strong control, because it centres the residuals before resampling. As discussed below, other forms of permutation testing could be used to provide the desired error rate control as an alternative to the bootstrap. 

An alternative approach to controlling the joint error rate (and the FWER) over multiple contrasts could be developed by considering permutations of the residuals rather than bootstrapping. Importantly, like the bootstrap, methods based on permuting the residuals via these methods are typically only valid asymptotically. This is because, when dealing with multiple contrasts in the linear model, exchangeability does not hold and so permutation is not exact: see Section \ref{A:permNwork} for further details. Permutation testing based methods in the linear model are very widespread \citep{Winkler2014}. As such establishing consistency results for these methods in the multivariate setting and using these to prove results on asymptotic joint error rate control is an interesting avenue for future research.

The choice of the error rate to control in a scenario where many hypotheses are being tested depends strongly on the goals of the researcher. The bounds that we have provided on the FDP provide more informative inference than simply controlling the FDR. As discussed in \cite{Neuvial2020}, under dependence controlling the FDR can lead to non-nonsensical results. Instead bounds on the FDP allow statements about the number of active voxels with a given set to be made. Moreover this inference is valid simultaneously over all sets and so guards against circular inference.
\fi

\newpage
\section*{Acknowledgments}
We are grateful to Alexandre Blain at the University of Toulouse for his help with demonstrating how to use the sanssouci code to get ARI to work and for checking that the Python implementation was consistent with the implementation in the ARIbrain R package. 
We are grateful to François Bachoc at the University of Toulouse for his help with the proof of Lemma \ref{lem:F0cont}.
SD is grateful to Fabian Telschow at Humboldt University for useful discussions on bootstrapping.

SD was supported by NIH grant R01EB026859. SD and PN were supported by the SansSouci ANR project (ANR-16-CE40-0019).
BT was supported by the KARAIB AI chair (ANR-20-CHIA-0025-01) and the FastBig ANR project (ANR-17-CE23-0011).

Data were provided in part by the Human Connectome Project, WU-Minn Consortium (Principal Investigators: David Van Essen and Kamil Ugurbil; 1U54MH091657) funded by the 16 NIH Institutes and Centers that support the NIH Blueprint for Neuroscience Research; and by the McDonnell Center for Systems Neuroscience at Washington University.

\newpage
\appendix
\if\biometrika1
\tableofcontents
\fi

\if\biometrika1
\section{ Further discussion}\label{S:FD}

\fi

\section{ Further theory for random fields}\label{SS:furtherf}
\subsection{Operations and convergence}
Given two random fields $ g$ and $ g': \mathcal{V}\rightarrow \mathbb{R}^L $, operations of addition and subtraction can be performed pointwise and so $ g+g'$ and $ g-g'$ are well defined. Moreover if instead $ g: \mathcal{V}\rightarrow \mathbb{R} $ then multiplication and division can also be performed pointwise and so, in that case, $gg'$ and $g'/g $ are well-defined.

Given $ D \in \mathbb{N} $, suppose that $ \mathcal{V} = \left\lbrace u_1, \dots, u_V \right\rbrace $ for some $ V \in \mathbb{N} $ and $ u_1, \dots, u_V \in \mathbb{R}^D $. For $ L \in \mathbb{N} $, let $ g: \mathcal{V} \rightarrow \mathbb{R}^L $ be a random field. Then we define vec$(g) \in \mathbb{R}^{L V}$ to be the vector whose $ ((i-1)L + j) $th element is $g_j(u_i)$ for $ 1\leq i \leq V $ and $ 1 \leq j \leq L. $ We refer this operation as \textbf{vectorization}. This allows us to easily define notions of convergence. Given a sequence: $ ((g_n)_{n \in \mathbb{N}}, g) $ of random fields from $ \mathcal{V} $ to $ \mathbb{R}^L $ we say that $ g_n $ converges to $ g $ in distribution (resp. probability/almost surely) if vec$(g_n)$ converges in distribution (resp. probability/almost surely) to vec$(g)$. We will write this as $ g_n \convd g $ (resp. $ g_n \convp g/g_n \convas g $) - a notation that we will also use for random variables in what follows. Given such a sequence we will write $ g_{n,j} \,\, (1 \leq j \leq L)$ to denote its components. 

\begin{definition}
	Given $ L, L' \in \mathbb{N}$, a random field $ g:\mathcal{V}\rightarrow \mathbb{R}^L$ and $ M \in \mathbb{R}^{L' \times L} $ then we define the random field $ Mg $ which sends $ v \in \mathcal{V} $ to $ Mg(v)  \in \mathbb{R}^{L'}$. Moreover, if $ L = 1 $ and $ a \in \mathbb{R}^{L'} $ is a vector then we define the random field $ ag $ which sends $ v \in \mathcal{V} $ to $ ag(v) \in \mathbb{R}^{L'} $
\end{definition}
\begin{lemma}\label{lem:matemult}
	For $ L, L' \in \mathbb{N} $ let $ g:\mathcal{V}\rightarrow \mathbb{R}^L$ be a random field with covariance $ \textswab{c} $ and let $ M \in \mathbb{R}^{L' \times L} $, then $ Mg $ has covariance 
	\begin{equation*}
	M \,\textswab{c}\, M^T.
	\end{equation*}
	Moreover if $ g $ is Gaussian then so is $ Mg. $
\end{lemma}

\subsection{Subsetting random fields}\label{SS:subsetting}
In what follows we will want to restrict random fields to subsets. This is defined formally as follows.
\begin{definition}\label{def:restrict}
	Suppose we have a random field $ g:\mathcal{V} \rightarrow \mathbb{R}^L $, some $ L \in \mathbb{N}, $ and a set valued function: $ \mathcal{N} $ on $ \mathcal{V} $, such that for each $ v \in \mathcal{V}$,  $\mathcal{N}_v \subset \left\lbrace 1, \dots, L \right\rbrace $. Then we define the \textbf{restriction} of $ g $ to $ \mathcal{N} $ to be the map $ g|_{\mathcal{N}}: \Omega \rightarrow \left\lbrace h: \mathcal{V} \rightarrow \bigcup_{1 \leq j \leq L} \mathbb{R}^j \right\rbrace $ such that
	$ g|_{\mathcal{N}}(\omega)(v) $ is the vector 
	$ (g_k(\omega)(v): k \in \mathcal{N}_v )^T \in \mathbb{R}^{\left| \mathcal{N}_v \right|}. $
\end{definition}
Given a set function $ \mathcal{N} $, defined as in Definition \ref{def:restrict}, we can stack the entries of $ g|_{\mathcal{N}} $ to create vec$ (g|_{\mathcal{N}}) $ and thus define $ g_n|_{\mathcal{N}} \convd g|_{\mathcal{N}} $,$ g_n|_{\mathcal{N}} \convp g|_{\mathcal{N}}$ and $g_n|_{\mathcal{N}} \convas g|_{\mathcal{N}} $.

\begin{definition}
	Given an $ L $-dimensional Gaussian field, $ g \sim \mathcal{G}(\mu, \textswab{c}) $ for some mean $ \mu $ and covariance $ \textswab{c} $ and a set function $ \mathcal{N} $ as defined above, we shall write $ \mathcal{G}(\mu, \textswab{c})|_\mathcal{N} $ to denote the distribution of the restricted random field. I.e. $ g|_{\mathcal{N}} \sim  \mathcal{G}(\mu, \textswab{c})|_\mathcal{N}. $ Given 
	\begin{equation*}
	f:\left\lbrace h:\mathcal{V} \rightarrow \mathbb{R}^L\right\rbrace \rightarrow \mathbb{R}
	\end{equation*}
	we shall write $ X \sim f(\mathcal{G}(\mu, \textswab{c})) $ to indicate that $ X $ is a real valued random variable which has the same distribution as $ f(g) $. Given 
	\begin{equation}
	f:\left\lbrace h:\mathcal{V} \rightarrow \bigcup_{1 \leq j \leq L} \mathbb{R}^j\right\rbrace \rightarrow \mathbb{R} 
	\end{equation}
	we similarly define the notation $ f(\mathcal{G}(\mu, \textswab{c})|_{\mathcal{N}}) $.
\end{definition}
\section{ Consistency of the bootstrap in the linear model}\label{S:Lbootproof}

\subsection{Lindeberg Central Limit Theorem}
In order to prove our main results we require Proposition \ref{prop:lindeberg} (stated below) which we prove using the Lindeberg CLT (see e.g. \cite{Vander2000} Chapter 2.8). We will also require the following lemma.
\begin{lemma}\label{lem:bound}
	Let $ X$ and $ Y $ be random variables such that $ \mathbb{E}\left[ \left| X \right|^{2 + \eta} \right] < \infty $ and $ \mathbb{E}\left[ \left| Y \right|^{K} \right] < \infty$ for some $ K, \eta > 0 $, then for all $ a \in \mathbb{R}, $
	\begin{equation*}
	\mathbb{E}\left[ X^2 1[ a|Y| > \gamma] \right] \leq 
	\gamma^{-K/q}a^{K/q}\mathbb{E}\left[ \left| X \right|^{2 + \eta} \right]^{1/(1+\eta/2)}\mathbb{E}\left[ \left| Y \right|^{K} \right]^{1/q}
	\end{equation*}
	where $ q = 1 - (1+\eta/2)^{-1}. $
\end{lemma}
\begin{proof}
	By Holder's inequality for $ p, q > 0 $ such that $ \frac{1}{p} + \frac{1}{q} = 1, $
	\begin{align*}
	\mathbb{E}\left[ X^2 1[ a|Y| > \gamma] \right] &\leq \mathbb{E}\left[ X^{2p} \right]^{1/p}\mathbb{E}\left[1[ a|Y| > \gamma] \right]^{1/q} = \mathbb{E}\left[ X^{2p} \right]^{1/p}\mathbb{P}\left( a|Y| > \gamma \right)^{1/q}\\
	& \leq \mathbb{E}\left[ X^{2p} \right]^{1/p}\left( \frac{\mathbb{E}\left[ a^K\left| Y \right|^K \right]}{\gamma^K} \right)^{1/q} = 
	\gamma^{-K/q}a^{K/q}\mathbb{E}\left[ X^{2p} \right]^{1/p}\mathbb{E}\left[ \left| Y \right|^{K} \right]^{1/q}
	\end{align*}
	where the middle inequality holds by Markov's inequality. Taking $ p = 1 + \eta/2 $ and $ q = 1 - \frac{1}{p} $, the result follows. 
\end{proof}
\begin{proposition}\label{prop:lindeberg}
	Given a sequence $ (k_n)_{n \in \mathbb{N}} $, let  $\left\lbrace \xi_{n,i}: n,i \in \mathbb{N}, 1 \leq i \leq k_n \right\rbrace$ be a triangular array of mean-zero random fields on $ \mathcal{V} $ which are i.i.d within rows and have finite covariance.	Let $\left\lbrace a_{ni}: n,i \in \mathbb{N}, 1 \leq i \leq n \right\rbrace$ be a triangular array of $ D- $dimensional vectors such that $ \sum_{i = 1}^n \left\lVert a_{ni} \right\rVert^{2+ K/q} \rightarrow 0$ as $ n \rightarrow \infty $ and $ \sup_{i,n} \mathbb{E}\left[ \left| \xi_{n,i} \right|^{\max(K, 2+\eta)} \right] < \infty $ for some $ K > 0 $, any $ \eta > 0 $ and $ q = 1 -(1+\eta/2)^{-1} $. Let $ A_n = (a_{n1}, \dots, a_{nk_n}) \in \mathbb{R}^{D \times k_n}$ and suppose that $ A_n^TA_n \rightarrow \Sigma \in \mathbb{R}^{D\times D}.$ For $ n \in \mathbb{N} $, let $ \textswab{c}_n $ be the covariance function of $ \xi_{n,1} $ and suppose that as $ n \rightarrow \infty, 	\textswab{c}_n \rightarrow \textswab{c}$ (pointwise) for some covariance function $ \textswab{c} $ on $ \mathcal{V} $. Then as $ n \rightarrow \infty, $
	\begin{equation*}
	\sum_{i = 1}^{k_n} a_{ni} \xi_{n,i} \convd \mathcal{G}(0, \textswab{c} \Sigma).
	\end{equation*}
\end{proposition}
\begin{proof}
	The proof is an application of the Lindeberg CLT (see e.g. \cite{Vander1998} Proposition 2.27) to the vectors $ \text{vec}(a_{ni}\xi_{n,i}) $. There are two conditions to verify. The first is to show that the covariance converges. We can show this blockwise, i.e., for each $ u, v \in\mathcal{V} $,
	\begin{align*}
	\sum_{i = 1}^{k_n} \cov(a_{ni} \xi_{n,i}(u), a_{ni} \xi_{n,i}(v)) 
	&= \sum_{i = 1}^{k_n} \mathbb{E}\left[a_{ni}\xi_{n,i}(u)\xi_{n,i}(v) a_{ni}^T\right] \\
	&= \textswab{c}_n(u,v) \sum_{i = 1}^{k_n} a_{ni}a_{ni}^T = \textswab{c}_n(u,v) A_n^TA_n.
	\end{align*}
	which converges to $ \textswab{c}(u,v) \Sigma $ as $ n \rightarrow \infty.$	For the second condition we need to show that for all $ \gamma > 0, $
	\begin{equation*}
	\sum_{i = 1}^{k_n}\mathbb{E} \left[ \left\lVert \text{vec}(a_{ni}\xi_{n,i}) \right\rVert^2 1[\left\lVert \text{vec}(a_{ni}\xi_{n,i}) \right\rVert > \gamma] \right] \underset{n \rightarrow \infty}{\longrightarrow} 0.
	\end{equation*}
	We can expand the left hand side as
	\begin{align}\label{eq:star}
	\sum_{i = 1}^{k_n} &\mathbb{E}\left[\sum_{v \in \mathcal{V}}\left\lVert a_{ni} \xi_{n,i}(v)\right\rVert^2 1\left[\sum_{u \in \mathcal{V}}\left\lVert a_{ni} \xi_{n,i}(u)\right\rVert^2 > \gamma^2\right]\right] \\
	&\hspace{2cm} \leq \sum_{i = 1}^{k_n} \sum_{v \in \mathcal{V}}\left\lVert a_{ni} \right\rVert^2 \mathbb{E}\left[ \xi_{n,i}(v)^2 \sum_{u \in \mathcal{V}} 1\left[\left\lVert a_{ni} \right\rVert\left| \xi_{n,i}(u) \right| > \gamma\left| \mathcal{V} \right|^{-1/2}\right]\right] \\
	&\hspace{2cm} = \sum_{i = 1}^{k_n} \left\lVert a_{ni} \right\rVert^2\sum_{u, v \in \mathcal{V}} \mathbb{E}\left[ \xi_{n,i}(v)^2 1\left[\left\lVert a_{ni} \right\rVert\left| \xi_{n,i}(u) \right| > \gamma\left| \mathcal{V} \right|^{-1/2}\right]\right] \\
	&\hspace{2cm} \leq C\sum_{i = 1}^{k_n} \left\lVert a_{ni} \right\rVert^{2+K/q}
	\end{align}
	for some fixed constant $ C > 0, $ chosen in accordance with Lemma \ref{lem:bound}. This bound converges to zero as $ n \rightarrow \infty $.
\end{proof}

\subsection{Proof of Theorem \ref{thm:bootactual}}\label{SS:LBP}
Here we prove Theorem \ref{thm:bootactual} from the main text. To so do we first establish some results. The first lemma we prove shows that the contribution of $ X_n(X_n^TX_n)^{-1}X_n^TE_n $ is zero asymptotically.
\begin{lemma}\label{lem:pnen}
	Under Assumption \ref{ass:X}, letting $ P_n = X_n(X_n^TX_n)^{-1}X_n^T $, as $ n \rightarrow \infty $,
	\begin{equation*}
		\lVert P_nE_n \rVert \convp 0.
	\end{equation*}
\end{lemma}
\begin{proof}
	Letting $ \beta = (2+\delta/2)^{-1} $, we have
	\begin{equation*}
	P_nE_n = X_n(X_n^TX_n)^{-1}X_n^TE_n = \frac{X_n}{n^{\beta}}\left(\frac{X_n^TX_n}{n}\right)^{-1}\left(\frac{X_n^TE_n}{n^{1-\beta}}\right).
	\end{equation*}
	Thus,
	\begin{equation*}
	\lVert P_nE_n \rVert =  \left\lVert\frac{X_n}{n^{\beta}}\left(\frac{X_n^TX_n}{n}\right)^{-1}\left(\frac{X_n^TE_n}{n^{1-\beta}}\right)\right\rVert \leq \left\lVert\frac{X_n}{n^{\beta}} \right\rVert\left\lVert\left(\frac{X_n^TX_n}{n}\right)^{-1} \right\rVert\left\lVert\left(\frac{X_n^TE_n}{n^{1-\beta}}\right)\right\rVert.
	\end{equation*}
	$\frac{X_n^TE_n}{\sqrt{n}} $ converges in distribution so $\left\lVert\left(\frac{X_n^TE_n}{n^{1-\beta}}\right)\right\rVert \convp 0 $ since $ 1 - \beta > \frac{1}{2} $ and $ \left(\frac{X_n^TX_n}{n}\right)^{-1} $ converges almost surely to $ \Sigma_X^{-1} $ by Lemma \ref{lem:davenconv}. Applying the Gershgorin circle theorem and the AM-RM inequality, we have 
	\begin{equation*}
	\left\lVert X_n \right\rVert \leq \max_{1 \leq i \leq n} \sum_{j = 1}^p |(X_n)_{ij}| = \max_{1 \leq i \leq n} \sum_{j = 1}^p |(x_i)_j| \leq \frac{p}{\sqrt{p}} \max_{1 \leq i \leq n} \lVert x_i \rVert.
	\end{equation*} 
	$ n^{-\beta}\max_{1 \leq i \leq n} \lVert x_i \rVert \convas 0$ since $ \mathbb{E}(\lVert x_1 \rVert^{2+\delta}) < \infty$, so in particular $ \lVert n^{-\beta} X_n \rVert \convas 0. $ Combining these results and using Slutsky, it follows that $ \lVert P_nE_n \rVert \convp 0.$
\end{proof}
In order to apply the Lindeberg Central Limit theorem and subsequently the triangular law of large numbers we will need the following bound on the moments of the bootstrapped residuals.
\begin{lemma}\label{lem:momentbound}
	Under Assumption \ref{ass:X}, conditional on $ (X_m, Y_m)_{m \in \mathbb{N}} $, for almost all sequences $ (X_m, Y_m)_{m \in \mathbb{N}} $,
	\begin{equation*}
		\sup_{n \in \mathbb{N}, 1\leq i \leq n}	\mathbb{E}(E^b_{n,i})^4 < \infty.
	\end{equation*}	
\end{lemma}
\begin{proof}
	Let $ P_n = X_n(X_n^TX_n)^{-1}X_n^T. $ For each $ n \in \mathbb{N} $, conditional on $ (X_m, Y_m)_{m \in \mathbb{N}} $, and $ 1 \leq i \leq n $ and $ 1 \leq b \leq B, $
	\begin{align*}
	\mathbb{E}(E^b_{n,i})^4 &= \sum_{j = 1}^n \frac{1}{n}\left(  \hat{E}_{n,j} - \frac{1}{n}\sum_{l = 1}^n \hat{E}_{n,l} \right)^4
	=\frac{1}{n}\sum_{j = 1}^n \left(\epsilon_j - (P_nE_n)_j - \frac{1}{n} \sum_{l = 1}^n \left(\epsilon_l - (P_nE_n)_l\right)\right)^4
	\end{align*}
	Now $ \lVert P_nE_n \rVert $ converges in probability to 0 unconditionally by Lemma \ref{lem:pnen} and so
	\begin{equation*}
	\max_{1\leq l \leq n} |(P_nE_n)_l| \convp 0,
	\end{equation*} 
	since $ \max_{1\leq l \leq n} (P_nE_n)_l^2 \leq  \lVert P_nE_n \rVert^2. $ In particular it follows that for $ M > 0 $,
	\begin{equation}\label{eq:makehay}
	\lim_{k \rightarrow \infty} \mathbb{P}\left( \max_{n \geq k } \max_{1 \leq l \leq n} |(P_nE_n)_l| > M\right) \rightarrow 0
	\end{equation}
	as $ k \rightarrow \infty. $ For $ k \in \mathbb{N} $, Let $ A_k = \left\lbrace \max_{n \geq k } \max_{1 \leq l \leq n} |(P_nE_n)_l| \leq M \right\rbrace $, then equation \eqref{eq:makehay} implies that $ \mathbb{P}\left( \cup_k A_k\right) = 1$ since the sets are nested. 
	As such for $ \omega \in \cup_k A_k $, $ \omega$ is contained in $ A_K $ some $ K = K(\omega) \in \mathbb{N}. $ It follows that 
	\begin{equation*}
	\max_{n > K} \max_{1 \leq l \leq n} |(P_nE_n)_l| \leq M 
	\end{equation*} 
	almost everywhere which implies that, almost surely, 
	\begin{equation*}
	\max_{n \in \mathbb{N}} \max_{1 \leq l \leq n} |(P_nE_n)_l| \leq M' = M + \max_{1 \leq n \leq K} \max_{1 \leq l \leq n} |(P_nE_n)_l|.
	\end{equation*}
	We can thus almost surely bound $ \mathbb{E}(E^b_{n,i})^4 $ by 
	\begin{equation*}
	\frac{1}{n} \sum_{ j = 1}^n \sum_{k = 0}^4 \left( \epsilon_j  - \frac{1}{n}\sum_{l = 1}^n \epsilon_l\right)^k (2M')^{4-k}
	\leq (2M')^4 \frac{1}{n} \sum_{ j = 1}^n \sum_{k = 0}^4 \left( \epsilon_j  - \frac{1}{n}\sum_{l = 1}^n \epsilon_l\right)^k.
	\end{equation*}
	The right hand side converges almost surely by the strong law of large numbers to a quantity that is the same for each $ i $. It follows that the supremum over $ i,n $ of $ \mathbb{E}(E^b_{n,i})^4  $ is bounded, a fact that is true almost everywhere since $ \mathbb{P}( \cup_k A_k ) = 1. $
\end{proof}

Given these results we're now ready to prove Theorem \ref{thm:bootactual}.
\begin{proof}
	Expanding, we have that 
	\begin{equation*}
	\sqrt{n}(\hat{\beta}_n^b -\hat{\beta}_n) = \sqrt{n}(X_n^TX_n)^{-1}X_n^T E^b_n = \left(\frac{X_n^TX_n}{n} \right)^{-1} \frac{1}{\sqrt{n}} \sum_{i = 1}^{n} x_i E^b_{n,i}.
	\end{equation*}
	Applying Lemma \ref{lem:davenconv}, $ \left(\frac{X_n^TX_n}{n} \right)^{-1}  $ converges a.s. to $ \Sigma_X^{-1} $. Moreover, $ (E^b_{n,i})_{n \in \mathbb{N}, 1\leq i\leq n}$ is a triangular array which is mean-zero and i.i.d within rows so result will follow by applying Proposition \ref{prop:lindeberg} with $ \eta = 2 $ and $ K < \frac{\delta}{2}$ and taking $ a_{ni} = x_i/\sqrt{n} $ for $ 1 \leq i \leq k_n = n $ and $ n \in \mathbb{N}$. Letting $ A_n = (a_{n1}, \dots, a_{nn}) $ we have $ A_n^TA_n = \frac{1}{n}\sum_{i = 1}^n \lVert x_i \rVert^{2} \convas \Sigma_X $ as $ n \rightarrow \infty $. Moreover,
	\begin{equation*}
		\sum_{i = 1}^n \lVert a_{ni} \rVert^{2+K/q} = \frac{1}{n^{2+K/q}}\sum_{i = 1}^n \lVert x_{i} \rVert^{2+K/q} < \frac{1}{n^{2+K/q}}\sum_{i = 1}^n \lVert x_{i} \rVert^{2+\delta} \convas 0.
	\end{equation*}
	The requisite bounds on the moments are provided by Lemma \ref{lem:momentbound}.
	
	As such in order to apply Proposition \ref{prop:lindeberg} it suffices to show that the covariance converges. In order to do so, for each $ u, v \in \mathcal{V} $, conditional on $ (X_m,Y_m)_{m \in \mathbb{N}} $,
	\begin{align*}
	\cov(E^b_{n,1}(u), E^b_{n,1}(v)) &= 
	\sum_{j = 1}^n \frac{1}{n}\left(  \hat{E}_{n,j}(u) - \frac{1}{n}\sum_{l = 1}^n \hat{E}_{n,l}(u) \right)\hat{E}_{n,j}(v)\\
	&= \frac{1}{n}\hat{E}_n(u)^T\hat{E}_n(v) - \left(\frac{1}{n} \sum_{j = 1}^n \hat{E}_{n,j}(u) \right)\left(\frac{1}{n} \sum_{j = 1}^n \hat{E}_{n,j}(v) \right)
	\end{align*}
	Now, letting $ P_n = X_n(X_n^TX_n)^{-1}X_n $ and letting $ I_n $ be the $ n \times n $ identity matrix,
	\begin{equation*}
	\frac{1}{n}\hat{E}_n(u)^T\hat{E}_n(v) = \frac{1}{n}E_n(u)^T(I_n - P_n)E_n(v) = \frac{1}{n}E_n(u)^TE_n(v) - 
	\frac{1}{n}E_n(u)^TP_nE_n(v).
	\end{equation*}
	We can write $ \frac{1}{n}E_n(u)^TE_n(v) = \frac{1}{n}\sum_{i = 1}^{n} \epsilon_i(u) \epsilon_i(v)$, which converges almost surely to $ \textswab{c}(u,v) $ by the strong law of large numbers. Moreover,
	\begin{align*}
	\frac{1}{n}E_n(u)^TP_nE_n(v) &= \frac{1}{n}E_n(u)^TX_n(X_n^TX_n)^{-1}X_n^TE_n(v)\\
	&=\left(\frac{ X_n^TE_n(u)}{n} \right)^T\left( \frac{X_n^TX_n}{n} \right)^{-1}\left(\frac{ X_n^TE_n(u)}{n}\right)
	\end{align*}
	which converges almost surely to zero as $ n \rightarrow \infty. $ Finally, 
	\begin{align*}
	\frac{1}{n} \sum_{j = 1}^n \hat{E}_{n,j}(u) &= \frac{1}{n} 1_n^T (I_n - P_n)E_n(u) = \frac{1}{n}1_n^TE_n(u) - \frac{1}{n}1_n^T X_n(X_n^TX_n)^{-1}X_n^T E_n(u)\\
	&= \frac{1}{n}\sum_{i = 1}^n \epsilon_i(u) - \left( \frac{1}{n}\sum_{i = 1}^n x_i \right)\left( \frac{X_n^TX_n}{n} \right)^{-1}\left( \frac{1}{n}\sum_{i = 1}^n x_i\epsilon_i(u) \right)
	\end{align*}
	which converges almost surely to 0 as $ n \rightarrow \infty.$ This establishes the CLT. \\
	
	\noindent It remains to show that the variance converges. To do so recall that 
	\begin{equation*}
	(\hat{\sigma}^b_n)^2 = \frac{1}{n}\sum_{i = 1}^n (E^b_{n,i})^2 - \left( \frac{1}{n} \sum_{i = 1}^n E^b_{n,i} \right)^2.
	\end{equation*}
	The $ E^b_{n,i} $ are i.i.d and mean-zero and the covariance of $ E^b_{n,i} $ converges as shown above. As such by the Lindeberg CLT, conditional on $ (X_m, Y_m)_{m \in \mathbb{N}}, $ $ \frac{1}{\sqrt{n}} \sum_{i = 1}^n E^b_{n,i} $ converges in distribution. Thus dividing by $  \sqrt{n}$, it follows that $\frac{1}{n} \sum_{i = 1}^n E^b_{n,i} $ converges in probability to zero as $ n \rightarrow \infty, $ for almost all sequences $ (X_m,Y_m)_{m \in \mathbb{N}}  $. For the first term, note that 
	\begin{equation}
	\mathbb{E}(E^b_{n,i})^2 = 
	\sum_{j = 1}^n \frac{1}{n}\left(  \hat{E}_{n,j} - \frac{1}{n}\sum_{l = 1}^n \hat{E}_{n,l} \right)^2
	\end{equation}
	which converges to $ \sigma^2 $ almost surely as $ n \rightarrow \infty. $ As such the result follows by the triangular weak law of large numbers which we can apply because $ \sup_{n \in \mathbb{N}, 1\leq i \leq n}	\mathbb{E}(E^b_{n,i})^4$ is bounded by Lemma \ref{lem:momentbound}.
\end{proof}

\begin{remark}
	Formally the randomness in the data can be separated from the randomness in the bootstrap resamples by considering a cross product of probability spaces (one for the bootstrap randomness and one for the randomness in the data). This allows the expectations conditional on $ (X_m, Y_m)_{m \in \mathbb{N}} $ to be well defined. For further details of this decomposition, see e.g. \cite{Kosorok2003} and \cite{Telschow2020Delta}.
\end{remark}

\if\biometrika1

\fi

\section{ Proofs for the main text}

\subsection{Proofs for Section 2}

\subsubsection{Proof of Lemma \ref{lem:JERSIM}}
\begin{proof}
	For $ 1 \leq k \leq K $, we have
	\begin{align*}
	\left\lbrace \left| R_k(\lambda) \cap \mathcal{N} \right| >  k-1 \right\rbrace &= \left\lbrace \left| \left\lbrace (l,v) \in \mathcal{N}: p_{n,l}(v) \leq t_k(\lambda)\right\rbrace\right| > k - 1 \right\rbrace\\
	&= \left\lbrace  p^n_{(k:\mathcal{N})} \leq t_k(\lambda)  \right\rbrace = \left\lbrace  t_k^{-1}(p^n_{(k:\mathcal{N})} ) \leq \lambda \right\rbrace
	\end{align*}
	As such, 
	\begin{align*}
	\bigcup_{1 \leq k \leq K} \left\lbrace \left| R_k(\lambda) \cap \mathcal{N} \right| > k-1 \right\rbrace = \left\lbrace \min_{1 \leq k \leq K} t_k^{-1}(p^n_{(k:\mathcal{N})})\leq \lambda \right\rbrace =\left\lbrace \min_{1 \leq k \leq K\wedge |\mathcal{H} |} t_k^{-1}(p^n_{(k:\mathcal{N})})\leq \lambda \right\rbrace.
	\end{align*}
\end{proof}

\begin{remark}
	Lemma \ref{lem:JERSIM} can be generalized to arbitrary $ \zeta_k $. The statement of the result in that case is that
	\begin{align*}
	\text{JER}\left( (R_k(\lambda),\zeta_k)_{1 \leq k \leq K} \right) = \mathbb{P}\left( \min_{1 \leq k \leq K\wedge |\mathcal{H} |} t_k^{-1}(p^n_{(\zeta_k + 1:\mathcal{N})})\leq \lambda \right).
	\end{align*}
	Throughout the main text we take $ \zeta_k = k - 1 $, this can be motivated by the fact that it implies that each individual rejection region $ R_k(\lambda) $ controls the $ k $-familywise error rate. However other choices provide valid inference, see \cite{Blanchard2020} for a discussion of the different choices of $ \zeta_k $. In particular this implies that the results in Section \ref{S:LMJER} can trivially be generalized to arbitrary $ \zeta_k $.
\end{remark}

\subsection{Proofs for Section 3}
We will need the following useful Lemma which is \cite{Davenport2021peakCRs}'s Lemma 8.2.

\begin{lemma}\label{lem:davenconv}
	Suppose that $ (X_n)_{n \in \mathbb{N}} $ satisfies Assumption \ref{ass:X}a and let $ \Sigma_X = \mathbb{E}\left[ x_1 x_1^T \right] $, then $ \Sigma_X $ is invertible and 
	\begin{equation*}
	\left( \frac{X_n^TX_n}{n} \right)^{-1} \convas \Sigma_X^{-1}.
	\end{equation*}
\end{lemma}

\subsubsection{Convergence in the Linear Model}
In this section we establish results for asymptotics of coefficients and test-statistics in the linear model, written in terms of the framework of random fields.
\begin{lemma}\label{lem:lm}
	Suppose that $ (X_n)_{n \in \mathbb{N}} $ and $ (\epsilon_n)_{n \in \mathbb{N}} $ satisfy Assumption \ref{ass:X}. Then 
	\begin{equation*}
	\sqrt{n}(\hat{\beta}_n - \beta) \convd \mathcal{G}(0, \textswab{c}_{\epsilon}\Sigma_X^{-1}). 
	\end{equation*}
\end{lemma}
\begin{proof}
	For each $ n \in \mathbb{N}, $
	\begin{equation*}
	\sqrt{n}(\hat{\beta}_n - \beta) = \sqrt{n}(X_n^TX_n)^{-1}X_n^T \epsilon_n = \left( \frac{X_n^TX_n}{n} \right)^{-1} \frac{1}{\sqrt{n}} \sum_{i = 1}^n x_i \epsilon_i.
	\end{equation*}
	By the Central Limit Theorem, $ \frac{1}{\sqrt{n}} \sum_{i = 1}^n x_i \epsilon_i $ converges to a $ p $-dimensional Gaussian random field with covariance
	\begin{align*}
	\cov\left( x_1\epsilon_{1}(u),  x_1\epsilon_{1}(v) \right) = \mathbb{E}\left[ x_1\epsilon_1(u)\epsilon_1(v)x_1^T \right]
	=\mathbb{E}\left[ \epsilon_1(u)\epsilon_1(v) \right] \mathbb{E}\left[ x_1x_1^T \right]  = \textswab{c}_{\epsilon}(u,v) \Sigma_X
	\end{align*}
	for $ u, v \in \mathcal{V} $.
	$ \left( \frac{X_n^TX_n}{n} \right)^{-1} $ converges almost surely to $ \Sigma_X^{-1} $ by Lemma \ref{lem:davenconv} and so the result follows by applying Lemma \ref{lem:matemult} and Slutsky as the limiting distribution has covariance (for each $ u,v, \in \mathcal{V} $)
	\begin{equation*}
	\Sigma_X^{-1} (\textswab{c}_{\epsilon}(u,v) \Sigma_X)\Sigma_X^{-1} = \textswab{c}_{\epsilon}(u,v) \Sigma_X^{-1}.
	\end{equation*}
\end{proof}

\noindent Let $ \textswab{c}': \mathcal{V} \times \mathcal{V} \rightarrow \mathbb{R} $ be the covariance function such that for all $ u,v \in \mathcal{V} $
\begin{equation}\label{eq:covdefs}
\textswab{c}'(u,v) = \rho_{\epsilon}(u,v) AC\Sigma_X^{-1}C^TA^T
\end{equation}
where $ A \in \mathbb{R}^{L \times L} $ is a diagonal matrix with $ A_{ll} = (c_l^T\Sigma_X^{-1}c_l)^{-1/2} $ for $ 1 \leq l \leq L. $
Then we having the following results (using the subsetting notation defined in Section \ref{SS:subsetting}).

\begin{theorem}\label{thm:tstatconv}
	For $ n \in \mathbb{N} $, let $ S_n $ be the $ L $-dimensional random field on $ \mathcal{V} $ defined by 
	\begin{equation*}
	S_{n,l} = \frac{c_l^T(\hat{\beta}_n - \beta)}{\hat{\sigma}_n\sqrt{ c_l^T(X_n^TX_n)^{-1}c_l}}.
	\end{equation*}
	for $ 1 \leq l \leq L $. Then, under the conditions of Lemma \ref{lem:lm}, as $ n \rightarrow \infty, $
	\begin{equation*}
	S_n \convd \mathcal{G}(0, \textswab{c}')
	\end{equation*}
	and it follows that
	\begin{equation*}
	T_n|_{\mathcal{N}} \convd \mathcal{G}(0, \textswab{c}')|_{\mathcal{N}}.
	\end{equation*}
\end{theorem}
\begin{proof}
	We can write 
	\begin{equation*}
	S_n =\sqrt{n}A_n C (\hat{\beta}_n - \beta_n)/\hat{\sigma}_n.
	\end{equation*}
	where $ A_n $ is a diagonal matrix with $ (A_n)_{ll} = \left( c_l^T\left( \frac{X_n^TX_n}{n} \right)^{-1}c_l \right)^{-1/2} $. $ A_n \convas A $ by Lemma \ref{lem:davenconv} and $ \hat{\sigma}_n \convas \sigma $ as $ n \rightarrow \infty $. So applying Lemmas \ref{lem:lm} and \ref{lem:matemult} and Slutsky, the first result follows. For $ (v, l) \in \mathcal{N} $, $ c^T_l\beta(v) = 0$. As such $ S_n|_{\mathcal{N}} = T_n|_{\mathcal{N}}$ and it follows that  
	\begin{equation*}
	T_n|_{\mathcal{N}} \convd \mathcal{G}|_{\mathcal{N}}.
	\end{equation*}
\end{proof}

\subsubsection{Proof of Theorem \ref{thm:bootactual}}
	Our proof of this result is written in Section \ref{SS:LBP}. \cite{Eck2018} recently proved a similar result, the advantage of the proof that we provide in Section \ref{SS:LBP} is that it is substantially simpler than the proof of the corresponding result (their Theorem 1) contained in \cite{Eck2018}. In what follows here we demonstrate that their Theorem 1 implies our Theorem \ref{thm:bootactual} by translating their result into the more compact notation of random fields.
\begin{proof}
	Applying \cite{Eck2018}'s Theorem 1a (conditioning on $ (X_m, Y_m)_{m \in \mathbb{N}} $ and restricting to the probability 1 event that $ (\frac{1}{n}X_n^TX_n)^{-1} \conv \Sigma_X^{-1} $ which exists by Lemma \ref{lem:davenconv}), we see that 
	\begin{equation*}
	\sqrt{n}(\text{vec}(\hat{\beta}_n^b) - \text{vec}(\hat{\beta}_n)) \conv N(0, \Sigma \otimes \Sigma_X^{-1}),
	\end{equation*}
	where $ \Sigma = \cov(\text{vec}(\epsilon_{1})) $. It follows that $ \sqrt{n}(\hat{\beta}_n^b -\hat{\beta}_n) $ converges in distribution to a Gaussian random field. The form of the covariance in the statement of the theorem follows as writing $ \mathcal{V} = \left\lbrace u_1, \dots, u_V \right\rbrace $, for $ 1\leq l,m \leq L $ and $ 1\leq j,k \leq V$,
	\begin{equation*}
	(\Sigma \otimes \Sigma_X^{-1}))_{L(j-1) + l, L(k-1) + m} = \textswab{c}_{\epsilon}(u_j, u_k)(\Sigma_X^{-1})_{lm}.
	\end{equation*}
	The convergence in probability of the bootstrapped variance follows from \cite{Eck2018}'s Theorem 1b.
\end{proof}

\begin{remark}
	\cite{Eck2018}'s theorem needs to be applied with care as they write the model as $ Y = \beta X + \epsilon $ rather than via the more standard formulation of $ Y = X\beta + \epsilon $, i.e. they take $ \beta $ to be a row vector rather than a column vector. Their vec operation is thus the result of stacking a transposed matrix. As such the resulting distribution in the statement of their Theorem 1 is $ N(0, \Sigma_X^{-1} \otimes \Sigma) $ rather than $ N(0, \Sigma \otimes \Sigma_X^{-1}) $ - we have transposed it in order to match our notation.
\end{remark}

\begin{remark}
	\cite{Eck2018}'s Theorem 1 is stated in terms of fixed design matrices which converge. Here we assume that the design is random but condition on it which allows us to apply their Theorem 1 because $ (x_n)_{n \in \mathbb{N}} $ and $ (y_n)_{n \in \mathbb{N}} $ are independent. \cite{Eck2018} has an alternative result (their Theorem 2) which applies when $(x_n, y_n)_{n \in \mathbb{N}}$ has a joint distribution, however this requires an alternative form of the bootstrap first introduced in \cite{Freedman1981}.
\end{remark}

\subsubsection{Proof of Theorem \ref{thm:boottstat}}
\begin{proof}
	We can write 
	\begin{equation*}
	T_n^b =\sqrt{n}A_n C (\hat{\beta}_n^b - \hat{\beta}_n)/\hat{\sigma}^b_n.
	\end{equation*}
	where $ A_n $ is defined as in the proof of \ref{thm:tstatconv}. Theorem 1 implies that as $ n \rightarrow \infty $, $\hat{\sigma}_n^b \convp \sigma $. Moreover $  A_n \convas A $ so applying Theorem \ref{thm:bootactual}, Lemma \ref{lem:matemult} and Slutsky, the  first result holds. The second result immediately follows from the first.
\end{proof}

\if\biometrika0
\subsection{Proof of Theorem \ref{thm:consist}}

\fi

\subsection{Proofs for Section \ref{S:LMJER}}
\subsubsection{Setup}
In what follows we will require the following lemma.
\begin{lemma}\label{cor:convcor}
	Let $ (F_n)_{n \in \mathbb{N}}, F $ be CDFs such that $ F_n $ converges to $ F $ pointwise and $ F $ is continuous. Let $ (\lambda_n)_{n\in \mathbb{N}} \in \mathbb{R}$ be a sequence such that $ \lambda_n \rightarrow \lambda \in \mathbb{R}$ as $ n \rightarrow \infty $, then 
	\begin{equation*}
	F_n(\lambda_n) \rightarrow F(\lambda).
	\end{equation*}
\end{lemma}
\begin{proof}
	We can write
	\begin{equation*}
	F_n(\lambda_n) - F(\lambda) = F_n(\lambda_n) - F(\lambda_n) + F(\lambda_n) - F(\lambda).
	\end{equation*}
	$ F_n $ converges uniformly to $ F $ (as CDFs which converge pointwise to a continuous limit do so uniformly) so $F_n(\lambda_n) - F(\lambda_n) \rightarrow 0 $ as $ n \rightarrow \infty $ and $  F(\lambda_n) - F(\lambda) \rightarrow 0 $ because $ F $ is continuous.
\end{proof}

\subsubsection{Proof of Theorem \ref{thm:JERcontrol}}
In order to facilitate the proof we will first make some further definitions. Firstly, given $ H \subset \mathcal{H} $ and $ T:\mathcal{V} \rightarrow \mathbb{R}^L $ define $ p_{(k:H)}(T) $ to be the minimum value in the set
\begin{equation}
\left\lbrace 2 - 2 \Phi(\left| T_l(v) \right|): (l,v) \in H \right\rbrace
\end{equation} 
where $ \Phi $ is the CDF of a standard normal distribution. Secondly, given $ H \subset \mathcal{H}  $, let $ f_H: \left\lbrace h: \mathcal{V}  \rightarrow \mathbb{R}^L \right\rbrace \rightarrow \mathbb{R} $ send $ T \in \left\lbrace h: \mathcal{V}  \rightarrow \mathbb{R}^L \right\rbrace  $ to $ \min_{1 \leq k \leq K \wedge \left| \mathcal{H}  \right|} t_k^{-1}(p_{(k:H)}(T)) $. Thirdly given a function $ S $ such that 
\begin{equation*}
	 S: \mathcal{V} \rightarrow \bigcup_{0 \leq j \leq L} \mathbb{R}^j, 
\end{equation*}
$ n \in \mathbb{N} $ and $ 1 \leq k \leq |\mathcal{H}|$ we shall define $ q^n_{k}(S) $ to be the $ k $th minimum value in the set 
	\begin{equation*}
	\left\lbrace 2 - 2\Phi_{n-r_n}(\left| S_l(v) \right|): v \in \mathcal{V}, l \leq \text{dim}(S(v))  \right\rbrace
	\end{equation*} 
	when this is well defined and take $ q^n_{k}(S) $ to be 1 when it is not (i.e. when $ k $ is larger than the size of the set). Here for $ z \in \bigcup_{1 \leq j \leq L} \mathbb{R}^j $, $ \text{dim}(z) $ denotes the dimension of $ z $. Similarly define $ q_k(S) $ to be the $ k $th minimum value in the set 
	\begin{equation*}
	\left\lbrace 2 - 2\Phi(\left| S_l(v) \right|): v \in \mathcal{V}, l \leq \text{dim}(S(v))  \right\rbrace.
	\end{equation*} 
	Finally we define functions $ \phi_{n}: \left\lbrace h: \mathcal{V} \rightarrow \bigcup_{0 \leq j \leq L} \mathbb{R}^j \right\rbrace \rightarrow \mathbb{R} $ which send
	\begin{equation*}
	S \mapsto \min_{1 \leq k \leq K \wedge \left| \mathcal{H} \right|} t_k^{-1}(q^n_{k}(S))
	\end{equation*}
 	and $ \phi: \left\lbrace h: \mathcal{V} \rightarrow \bigcup_{0 \leq j \leq L} \mathbb{R}^j \right\rbrace \rightarrow \mathbb{R} $ which sends $ S \mapsto \min_{1 \leq k \leq K \wedge \left| \mathcal{H} \right|} t_k^{-1}(q_{k}(S)). $\\
 	
	\noindent With these definitions in mind we are ready to prove Theorem \ref{thm:JERcontrol}.
\begin{proof}
	Defining $  \textswab{c}' $ as in Section \ref{SS:convergence}, $T_n|_{\mathcal{N}} $ converges to $ \mathcal{G}(0, \textswab{c}')|_{\mathcal{N}} $ in distribution by Theorem \ref{thm:tstatconv}. As such, using the fact that $ \phi_n $ is the composition of functions which are either continuous or converge uniformly with range $ [0,1] $,
    by Lemma \ref{lem:unifconv} and the Continuous Mapping Theorem,
	\begin{equation}\label{eq:tconv}
	f_{n, \mathcal{N}}(T_n) = \phi_n(T_n|_{\mathcal{N}}) \convd \phi(\mathcal{G}(0, \textswab{c}')|_{\mathcal{N}}) = f_{\mathcal{N}}(\mathcal{G}(0, \textswab{c}')).
	\end{equation}
	By the same logic, and applying Theorem \ref{thm:boottstat}, for sets $ H $ such that $ \mathcal{N} \subset H \subset \mathcal{H} $,
	\begin{equation}\label{eq:bc}
	f_{n, H}(T_n^b) \convd f_{H}(\mathcal{G}(0, \textswab{c}')).
	\end{equation}
	This convergence occurs conditional on the data, a fact that we take as implicit in \eqref{eq:bc} and in the rest of the proof. As such, applying Theorem \ref{thm:consist}, it follows almost surely that $ \lambda^*_{\alpha,n,B}(H) \rightarrow \lambda_{\alpha} = F^{-1}(\alpha) $, where $ F $ is the CDF of $ f_{H}(\mathcal{G}(0, \textswab{c}')) $ using the fact that $ F $ is strictly increasing (which follows from the form of $ f_H $ and the fact that the density of the multivariate normal distribution is positive everywhere) and continuous, by Lemma \ref{lem:F0cont}. Letting $ F_n $ be the CDF of $ f_{n, \mathcal{N}}(T_n) $ and $ F_0 $ be the CDF of $ f_{\mathcal{N}}(\mathcal{G}(0, \textswab{c}')) $, we have $ F_n \conv F_0 $ pointwise using \eqref{eq:tconv} and the fact that $ F_0 $ is continuous (which follows from Lemma \ref{lem:F0cont}). As such, applying Lemma \ref{cor:convcor} (since $F_n$ and $F_0$ are CDFs and $ F_0 $ is continuous), it follows that for all $ \epsilon > 0, $
	\begin{align*}
	\lim_{n \rightarrow \infty} &\lim_{B \rightarrow \infty} \mathbb{P}\left( f_{n, \mathcal{N}}(T_n) \leq \lambda^*_{\alpha,n,B}(H) \right) \\
	&\quad\leq \lim_{n \rightarrow \infty} \lim_{B \rightarrow \infty} \mathbb{P}\left( f_{n, \mathcal{N}}(T_n) \leq \lambda_{\alpha} + \epsilon \right)+ \mathbb{P}(|\lambda^*_{\alpha,n,B}(H) - \lambda_{\alpha}| > \epsilon)\\
	&\quad= \lim_{n \rightarrow \infty}  \lim_{B \rightarrow \infty} F_n(\lambda_{\alpha}+\epsilon) + \lim_{n \rightarrow \infty}  \lim_{B \rightarrow \infty} \mathbb{P}(|\lambda^*_{\alpha,n,B}(H) - \lambda_{\alpha}| > \epsilon)\\
	&\quad= F_0(\lambda_{\alpha}+\epsilon) \leq F(\lambda_{\alpha}+\epsilon) = \alpha + \epsilon.
	\end{align*} 
	Taking $ \epsilon $ to zero proves the upper bound. When $ H = \mathcal{N} $ arguing similarly but with $ \lambda_\alpha - \epsilon $ yields a lower bound and the desired equality. Note that the final inequality holds because 
	\begin{equation*}
	f_{H}(\mathcal{G}(0, \textswab{c}')) = \min_{1 \leq k \leq K \wedge \left| \mathcal{H}  \right|} t_k^{-1}(p_{(k:H)}(\mathcal{G}(0, \textswab{c}'))) \leq \min_{1 \leq k \leq K \wedge \left| \mathcal{H}  \right|} t_k^{-1}(p_{(k:\mathcal{N})}(\mathcal{G}(0, \textswab{c}'))) = f_{\mathcal{N}}(\mathcal{G}(0, \textswab{c}'))
	\end{equation*}
	and so 
	\begin{equation*}
	F_0(\lambda_{\alpha}+\epsilon) = \mathbb{P}\left( f_{\mathcal{N}}(\mathcal{G}(0, \textswab{c}')) \leq \lambda_{\alpha}+\epsilon \right) \leq \mathbb{P}\left( f_{H}(\mathcal{G}(0, \textswab{c}')) \leq \lambda_{\alpha}+\epsilon \right) = F(\lambda_{\alpha}+\epsilon).
	\end{equation*}
\end{proof} 
\subsubsection{Proof of Corollary \ref{cor:posthoc}}
\begin{proof}
	For any $ \epsilon > 0 $, and all large enough $ n $ and $ B $, we have 
	\begin{equation*}
	\mathbb{P}\left( \min_{1 \leq k \leq K\wedge \left| \mathcal{H}  \right|} t_k^{-1}(p_{(k:\mathcal{N})}(T_n)) \leq \lambda^*_{\alpha,n,B}(\mathcal{H} ) \right) \leq \alpha + \epsilon
	\end{equation*}
	and so, arguing as in \cite{Blanchard2020},
	\begin{equation*}
	\mathbb{P}\left(|H \cap \mathcal{N}| \leq \mybar{V}_{\alpha, n,B}(H),\,\, \forall H \subset \mathcal{H}  \right) \leq 1-\alpha - \epsilon.
	\end{equation*}
	The result follows by sending $ \epsilon $ to zero.
\end{proof}

\subsubsection{Proof of Theorem \ref{thm:stepdown}}
The proof is similar to that of Proposition 4.5 of \cite{Blanchard2020}.
\begin{proof}
	Let 
	\begin{equation*}
	\Omega_n = \left\lbrace p^{n}_{(k:\mathcal{N})}(T_n) \geq t_k(\lambda^*_{\alpha, n, B}(\mathcal{N})) \text{ for all } 1 \leq k \leq K\right\rbrace.
	\end{equation*}
	Then by Theorem \ref{thm:JERcontrol}, 
	\begin{equation*}
	\lim_{n\rightarrow\infty} \lim_{B \rightarrow \infty} \mathbb{P}(\Omega_n) = 1 - \alpha.
	\end{equation*}
	We claim that on the event $ \Omega_n $, $ \mathcal{N} \subseteq \hat{H}_n $. We prove this inductively, using the notation from Algorithm \ref{Alg:stepdown}. $ \mathcal{N} \subseteq H^{(0)} $ trivially. Assuming that $ \mathcal{N} \subseteq H^{(j-1)} $ for some $ j \in \mathbb{N}, $ it follows that $ p^{n}_{(k:H^{(j-1)})} \leq p^{n}_{(k:\mathcal{N})}$ i.e. that $ f_{n, H^{(j-1)}} \leq f_{n,\mathcal{N}}$. In particular, 
	\begin{equation*}
	\lambda^*_{\alpha, n, B}(H^{(j-1)}) \leq \lambda^*_{\alpha, n, B}(\mathcal{N})
	\end{equation*}
	and thus (since we are on $ \Omega_n $ and $ t_1 $ is increasing),
	\begin{equation*}
	p^{n}_{(1:\mathcal{N})}(T_n) \geq t_1(\lambda^*_{\alpha, n, B}(H^{(j-1)}))
	\end{equation*}
	which implies that $ \mathcal{N} \subseteq H^{(j)} $. Thus $ \mathcal{N} \subseteq \hat{H}_n $ and so for all $ 1 \leq k \leq K, $ 
	\begin{equation*}
	p^{n}_{(k:\mathcal{N})}(T_n) \geq t_k(\lambda^*_{\alpha, n, B}(\mathcal{N})) \geq 
	t_k(\lambda^*_{\alpha, n, B}(\hat{H}_n))
	\end{equation*}
	and so 
	\begin{equation*}
	f_{n, \mathcal{N}}(T_n) \geq \lambda^*_{\alpha,n,B}(\hat{H}_n).
	\end{equation*}
	The post hoc bound follows as in the proof of Corollary \ref{cor:posthoc}.
\end{proof}

\newpage
\section{ Further Theory}
\if\biometrika1

\fi
\subsection{FWER inference}\label{A:FWER}
In brain imaging it is also desirable to control the familywise error rate (FWER) over the hypotheses. This corresponds to performing multiple testing inference on the data and returning a set of active hypotheses $ R \subset \mathcal{H}  $ such that
\begin{equation*}
\text{FWER} = \mathbb{P}\left(R \cap \mathcal{N} \right) \leq \alpha.
\end{equation*}

When a single test is being used (for a single contrast or an $ F $-test at each voxel), brain imaging studies typically use a permutation based procedure \citep{Winkler2014} in order to control these error rates. In the case of multiple contrasts this approach is not always applicable because exchangeability can break down - see Section \ref{A:permNwork}. However the residual bootstrap can be used instead. In particular we have the following theorem which follows as a corollary of Theorem \ref{thm:JERcontrol} by taking $ K = 1 $ and taking $ t_1(x) = x $ for $ x \in [0,1]. $
\begin{theorem}\label{thm:fwer}
	For $ \alpha \in (0,1) $ and $ n, B \in \mathbb{N} $, let $ \lambda'_{\alpha,n,B} $ be the $ \alpha $-quantile of the bootstrap distribution (based on $ B $ bootstraps) of \begin{equation*}
	p_{1:\mathcal{H} }(T_n) = \min_{(l, v) \in \mathcal{H} } p_{n,l}(v).
	\end{equation*}
	Let $R_{n,B} = \left\lbrace (l,v) \in \mathcal{H}:  p_{n,l}(v) \leq \lambda'_{\alpha,n,B} \right\rbrace $. Then
	\begin{equation*}
	\lim_{n \rightarrow \infty} \lim_{B \rightarrow \infty}\mathbb{P}\left(R_{n,B} \cap \mathcal{N} \right) = 
	\lim_{n \rightarrow \infty} \lim_{B \rightarrow \infty}\mathbb{P}\left(\min_{(l,v) \in \mathcal{N}}p_{n,l}(v) \leq \lambda'_{\alpha,n,B}  \right) \leq \alpha.
	\end{equation*} 
	So choosing $ R_{n,B} $ as the rejection set provides asymptotic strong control of the FWER.
\end{theorem}
Moreover, a step-down version of this result follows as a corollary to Theorem \ref{thm:stepdown}. For these results the control of the FWER does not occur simultaneously with the control of the joint error rate. However let us revert to the general setting of the paper and take $ \lambda^*_{\alpha,n,B} $ to be the $ \alpha-$quantile of the bootstrap distribution of $ f_{n,\mathcal{N}} $ (as defined in the statement of Theorem \ref{thm:JERcontrol}). Then FWER control is automatically entailed with control of the joint error rate by using the rejection set $ R = \left\lbrace (l,v) \in \mathcal{H}: p_{n,l} \leq t_1^{-1}(\lambda^*_{\alpha,n,B}) \right\rbrace $. When $ K > 1 $, typically $ t_1^{-1}(\lambda^*_{\alpha,n,B}) $ will be less than the value of $ \lambda'_{\alpha,n,B} $ from Theorem \ref{thm:fwer} so this will result in less power but comes with the advantage of holding jointly with control of the joint error rate. Which version is to be preferred depends on which error rate one desires to control.

\begin{figure}[h!]
	\begin{subfigure}[t]{0.32\textwidth}
		\centering
		\includegraphics[width=\textwidth]{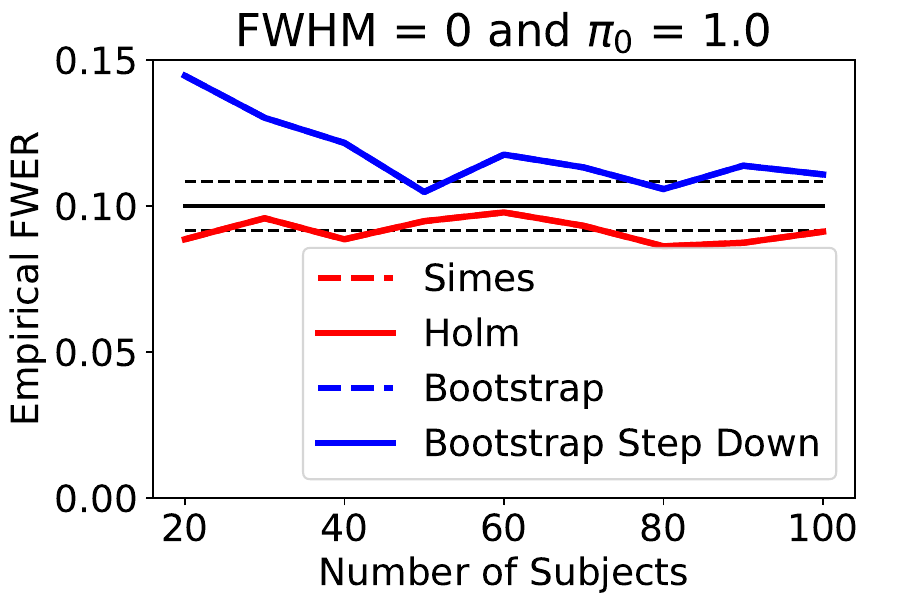}
	\end{subfigure}
	\hfill
	\begin{subfigure}[t]{0.32\textwidth}  
		\centering 
		\includegraphics[width=\textwidth]{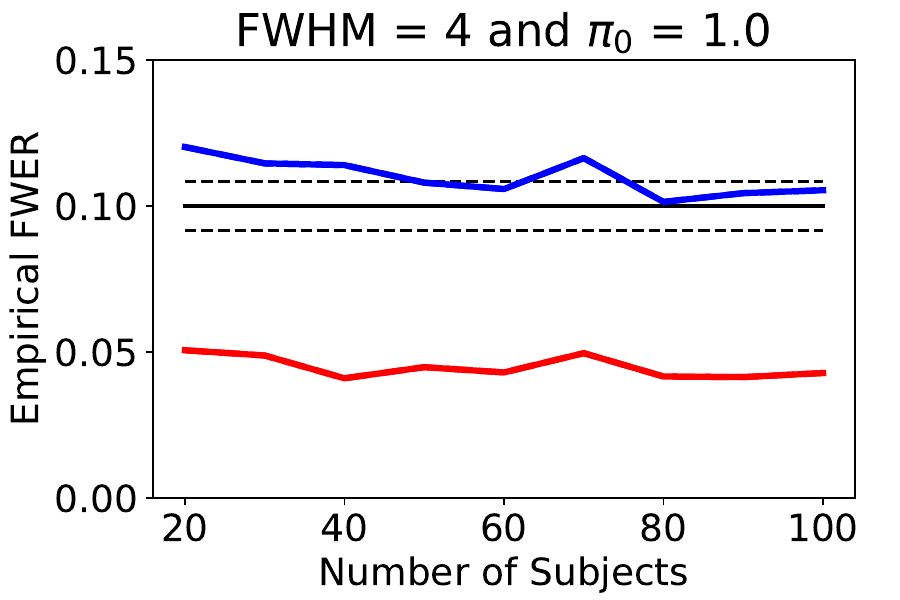}
	\end{subfigure}
	\hfill
	\begin{subfigure}[t]{0.32\textwidth}  
		\centering
		\includegraphics[width=\textwidth]{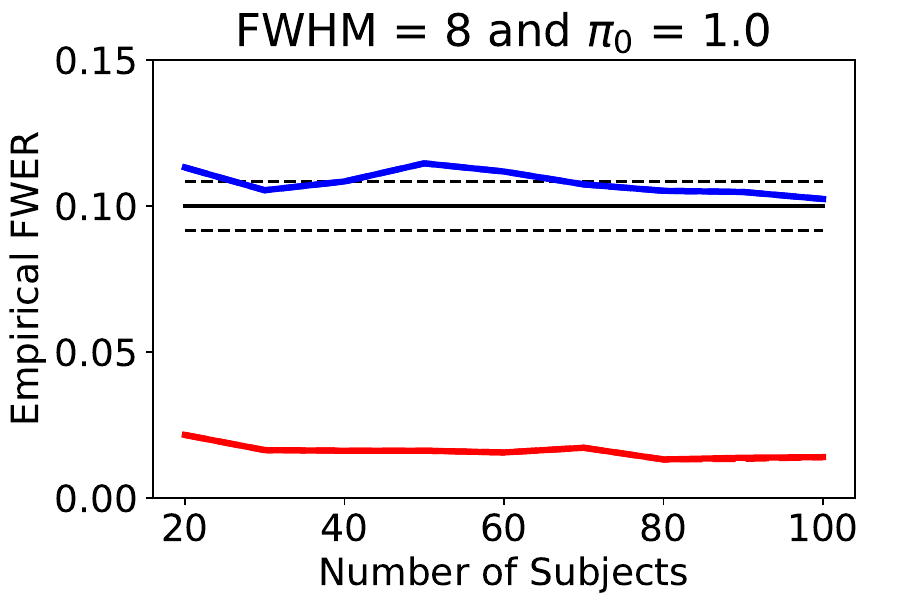}
	\end{subfigure}
	\begin{subfigure}[t]{0.32\textwidth}
		\centering
		\includegraphics[width=\textwidth]{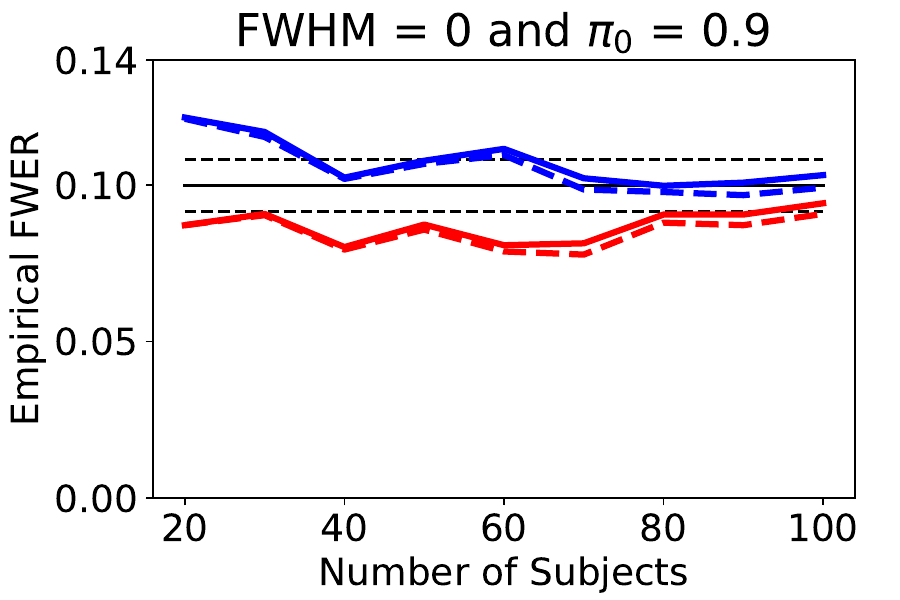}
	\end{subfigure}
	\hfill
	\begin{subfigure}[t]{0.32\textwidth}  
		\centering 
		\includegraphics[width=\textwidth]{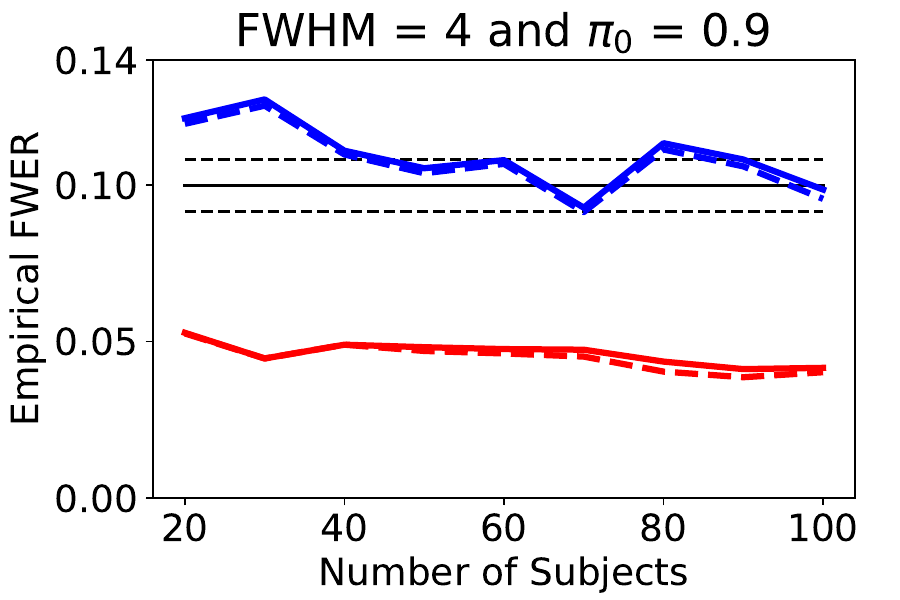}
	\end{subfigure}
	\hfill
	\begin{subfigure}[t]{0.32\textwidth}  
		\centering
		\includegraphics[width=\textwidth]{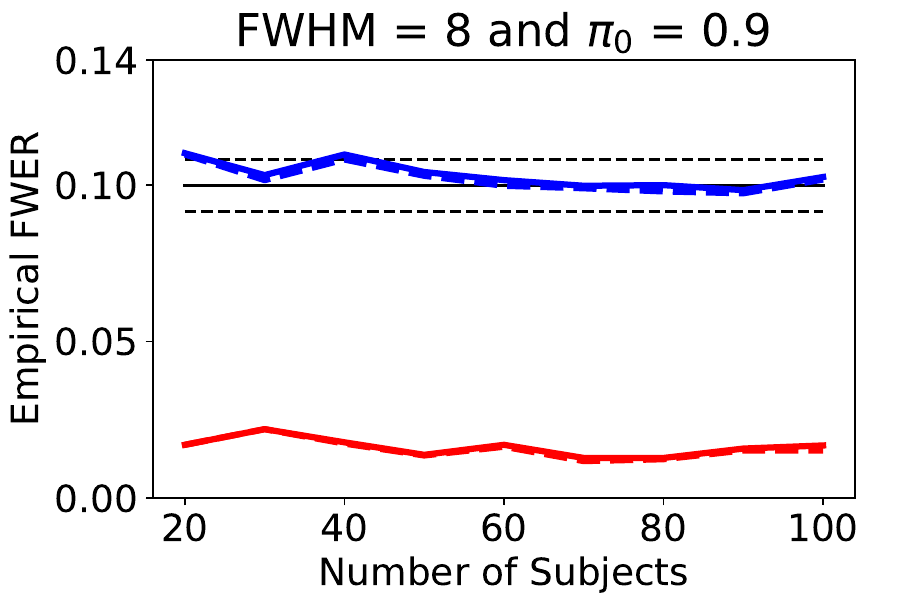}
	\end{subfigure}
	\begin{subfigure}[t]{0.32\textwidth}
		\centering
		\includegraphics[width=\textwidth]{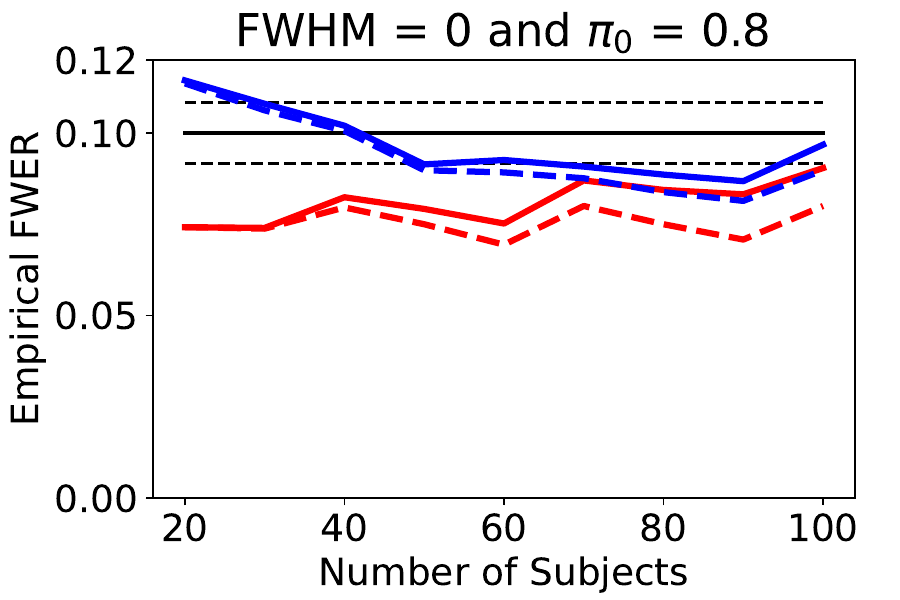}
	\end{subfigure}
	\hfill
	\begin{subfigure}[t]{0.32\textwidth}  
		\centering 
		\includegraphics[width=\textwidth]{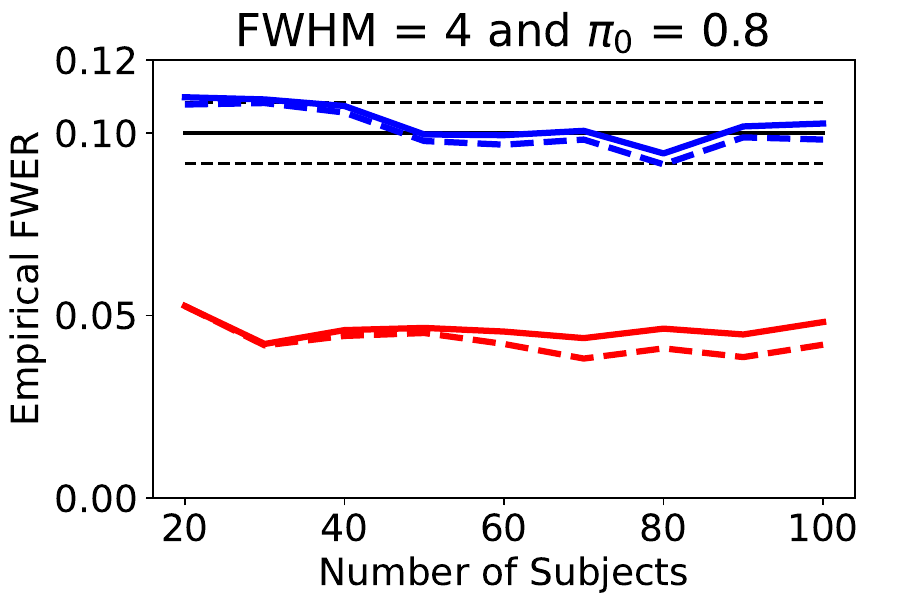}
	\end{subfigure}
	\hfill
	\begin{subfigure}[t]{0.32\textwidth}  
		\centering
		\includegraphics[width=\textwidth]{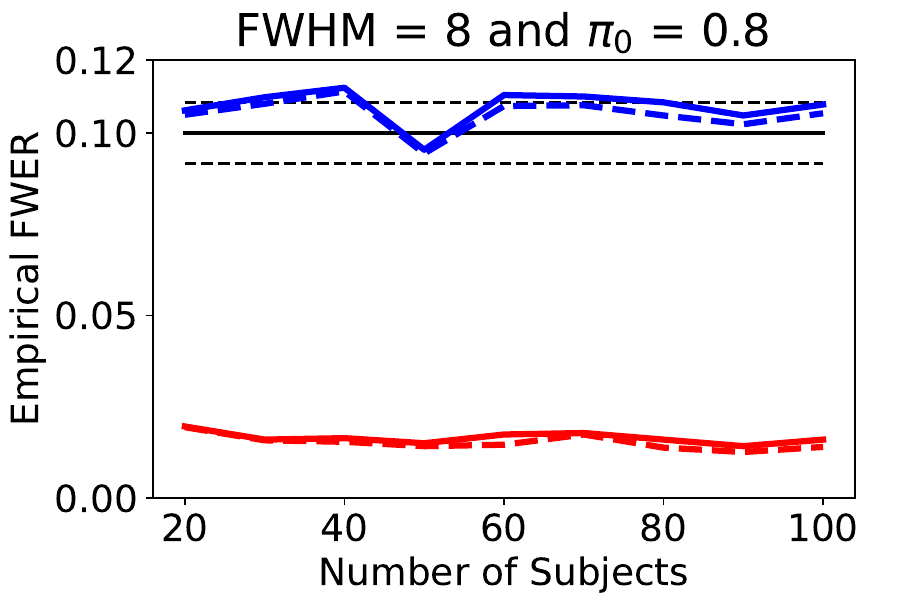}
	\end{subfigure}
	\begin{subfigure}[t]{0.32\textwidth}
		\centering
		\includegraphics[width=\textwidth]{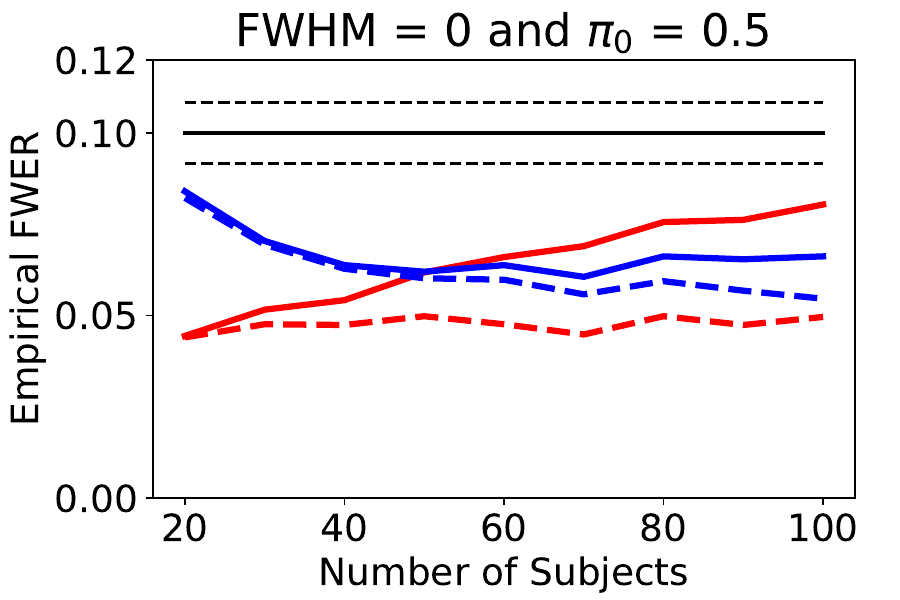}
	\end{subfigure}
	\hfill
	\begin{subfigure}[t]{0.32\textwidth}  
		\centering 
		\includegraphics[width=\textwidth]{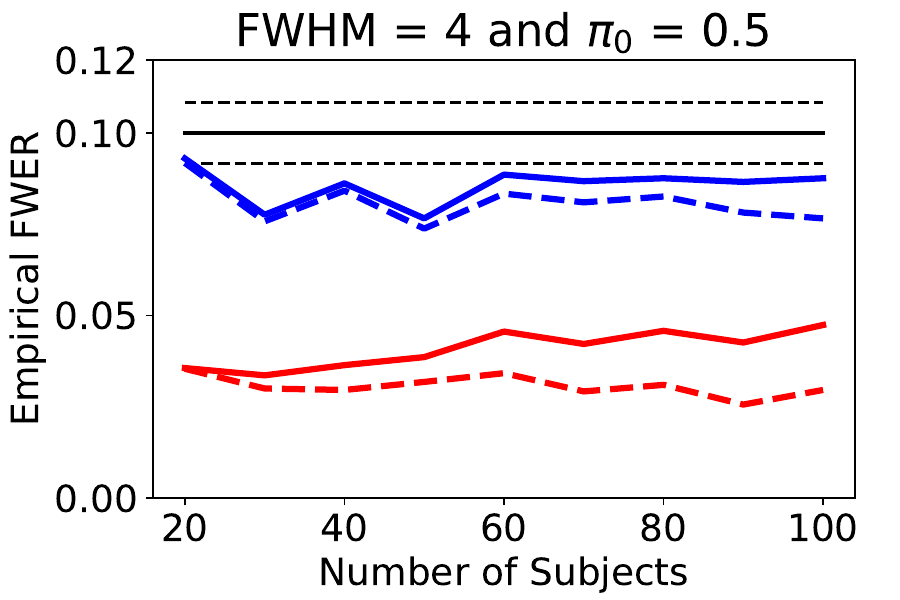}
	\end{subfigure}
	\hfill
	\begin{subfigure}[t]{0.32\textwidth}  
		\centering
		\includegraphics[width=\textwidth]{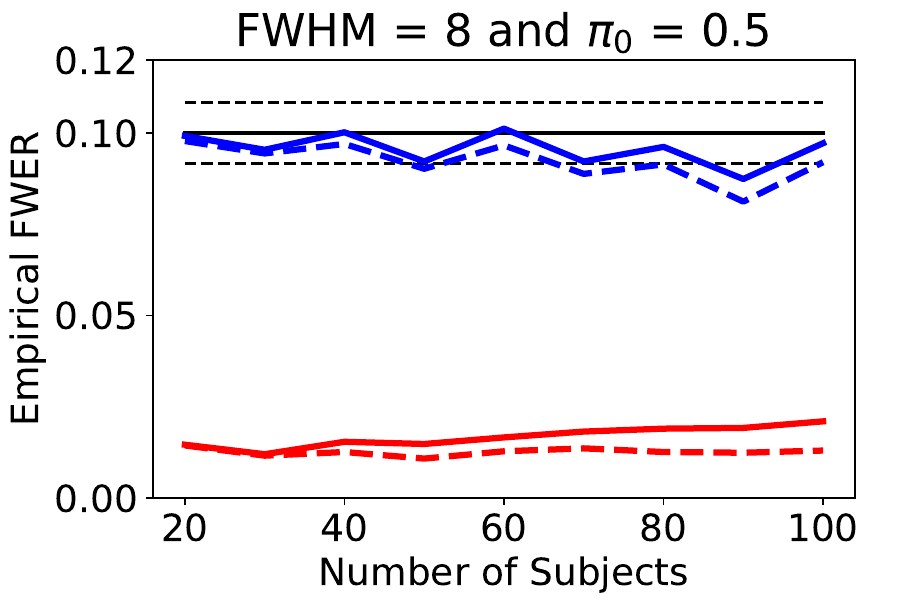}
	\end{subfigure}
	\caption{Direct FWER control using the different methods. Here the parametric procedures are Bonferroni and its step-down: \cite{Holm1979}.}\label{fig:fwer50}
\end{figure}

\subsection{Permutation in the Linear Model}\label{A:permNwork}
Here we show that under the alternative that $ \beta \not= 0 $ at a given point (e.g. voxel or gene), naively permuting the data (\cite{Manly1986}) does not generate data under the global null even when the noise is exchangeable.

\begin{claim}
	Suppose that the global null is not true, i.e. $ \beta \not= 0 $, then permuting $ Y $ is not equivalent to generating data under the global null (and so cannot be used to generate under the null and provide strong control over contrasts).
\end{claim}
\begin{proof}
	Let $ P $ be a permutation matrix, then
	\begin{equation*}
	PY = P(X\beta + \epsilon) = PX\beta + P\epsilon.
	\end{equation*}
	Now 
	\begin{equation*}
	P\epsilon \sim \epsilon
	\end{equation*}
	by exchangability. However $ PX\beta \not= 0 $ so it is not true that $ PY \sim \epsilon $ which is what we want (because we need to simulate under the null model, in order to provide strong control over contrasts).  $ PX\beta $ is a random variable (due to randomness in $ P $) with a non-zero mean and variance. So regressing $ PY $ against $ X $ gives linear model coefficients of 
	\begin{align*}
	\hat\beta^* &= (X^TX)^{-1}X^TPY = (X^TX)^{-1}X^TP(X\beta + \epsilon)\\
	&= (X^TX)^{-1}X^TPX\beta +(X^TX)^{-1}X^TP\epsilon.
	\end{align*}
	Now, under exchangeability, 
	\begin{equation*}
	(X^TX)^{-1}X^TP\epsilon  \sim (X^TX)^{-1}X^T\epsilon 
	\end{equation*}
	which indeed is the distribution of the linear model estimates under the null, however 
	\begin{equation*}
	(X^TX)^{-1}X^TPX\beta \not= 0
	\end{equation*}
	which causes a problem.
\end{proof}
\subsection{Additional Lemmas for the proofs}
\begin{lemma}\label{lem:unifconv}
	Suppose that $ (Z_n)_{n \in \mathbb{N}}, Z $ are $ \mathbb{R}^M $ valued random variables, for some $ M \in \mathbb{N} $.  Let $ (f_n)_{n \in \mathbb{N}}, f $ be functions from $\mathbb{R}^M \rightarrow I$ for some compact set $ I \subset \mathbb{R} $. Suppose that $ f_n $ converges uniformly to $ f $, that $ f $ is continuous and that $ Z_n \convd Z $, then
	\begin{equation*}
	f_n(Z_n) \convd f(Z).
	\end{equation*}
\end{lemma}
\begin{proof}
	Given any continuous and bounded function $ h: \mathbb{R} \rightarrow \mathbb{R} $,
	\begin{equation*}
	\left| \mathbb{E}\left[ h(f_n(Z_n)) \right] - \mathbb{E}\left[ h(f(Z)) \right] \right| \leq  \left| \mathbb{E}\left[ h(f_n(Z_n)) \right] -\mathbb{E}\left[ h(f(Z_n)) \right] \right| + \left| \mathbb{E}\left[ h(f(Z_n)) \right] - \mathbb{E}\left[ h(f(Z)) \right] \right|.
	\end{equation*}
	The set $ I $ is compact and so the restriction of $ h $ to $ I $ is uniformly continuous. So for any $ \epsilon > 0 $ there is some $ \delta $ such that $ \left| h(x) - h(y) \right| < \epsilon$ for all $ x, y \in I $ such that $ \left| x - y \right| < \delta $. By uniform convergence, there is some $ N \in \mathbb{N}$ such that for all $ n > N $, $ \left| f_n(z) - f(z) \right| < \delta $ for all $z \in \mathbb{R}^L $. The functions $ (f_n)_{n \in \mathbb{N}}, f $ take values within $I$ so it follows that
	\begin{equation*}
	\left| \mathbb{E}\left[ h(f_n(Z_n)) \right] -\mathbb{E}\left[ h(f(Z_n)) \right] \right| \leq  \mathbb{E}\left[ \left| h(f_n(Z_n)) - h(f(Z_n)) \right| \right] < \mathbb{E}\left[ \epsilon \right] = \epsilon.
	\end{equation*}
	So this term converges to zero as $ n \rightarrow \infty $. The second term: $ \left| \mathbb{E}\left[ h(f(Z_n)) \right] - \mathbb{E}\left[ h(f(Z)) \right] \right| $ also converges to zero as $ h \circ f $ is a continuous bounded function and $ Z_n \convd Z$ as $ n \rightarrow \infty $ (by applying the Portmanteau Lemma). Thus, as $ n \rightarrow \infty, $
	\begin{equation*}
	\mathbb{E}\left[ h(f_n(Z_n)) \right] \rightarrow \mathbb{E}\left[ h(f(Z)) \right].
	\end{equation*}
	Since this holds for any continuous bounded $ h $ the result follows by Portmanteau.
\end{proof}

\begin{lemma}\label{lem:F0cont}
	Let $ F_H $ be the CDF of $ \min_{1 \leq k \leq K \wedge \left| \mathcal{H}  \right|} t_k^{-1}(p_{(k:H)}(T)) $ where $ T \sim \mathcal{G}(0, \textswab{c}') $ and $ \textswab{c}' $ is defined as in Section \ref{SS:convergence}, then $ F_H $ is continuous.
\end{lemma}
\begin{proof}
	It is sufficient to show that for all $ \lambda \in \mathbb{R} $, $ \mathbb{P}\left( \min_{1 \leq k \leq K \wedge \left| \mathcal{H}  \right|} t_k^{-1}(p_{(k: \left| \mathcal{H}  \right|)}(T)) = \lambda \right) = 0.$ To show this, choose $ \lambda \in \mathbb{R} $, then
	\begin{align*}
	\mathbb{P}\left( \min_{1 \leq k \leq K \wedge \left| \mathcal{H}  \right|} t_k^{-1}(p_{(k:H)}(T)) = \lambda \right) &\leq \mathbb{P}\left( \exists 1 \leq k \leq \left| \mathcal{H} \right| \text{ s.t. } t_k^{-1}(p_{(k:H)}(T)) = \lambda \right)\\
	&= \mathbb{P}\left( \exists 1 \leq k \leq m \text{ s.t. } p_{(k:H)}(T) = t_k(\lambda) \right)\\
	&\leq \sum_{k = 1}^{m} \mathbb{P}\left( p_{(k:H)}(T) = t_k(\lambda)\right)\\
	&\leq  \sum_{k = 1}^m \mathbb{P}\left( \exists (l,v) \in \mathcal{H}: 2(1 - \Phi(\left| T_l(v) \right|)) = t_k(\lambda)  \right)\\
	&\leq  \sum_{k = 1}^m \sum_{(l,v) \in \mathcal{H}} \mathbb{P}\left( 2(1 - \Phi(\left| T_l(v) \right|)) = t_k(\lambda)  \right).
	\end{align*}
	Now given $ (l, v) \in \mathcal{H} $ and $ 1 \leq k \leq m, $
	\begin{equation*}
	\mathbb{P}\left( 2(1 - \Phi(\left| T_l(v) \right|)) = t_k(\lambda)  \right) = \mathbb{P}\left( \left| T_l(v) \right| = \Phi^{-1}\left( 1 - t_k(\lambda)/2 \right) \right) = 0
	\end{equation*}
	since $ T_l(v) $ is a Gaussian random variable. The result follows.
\end{proof}

\section{ fMRI data pre-processing}\label{S:fmridataprocessing}
Participants underwent a working memory task in which they were shown images asked to remember them. They were reshown them at a subsequent point. This is known as an $ m $-back task when $ m \in \mathbb{N}$ is number of intervals between when each image is shown and then repeated - see \cite{Barch2013} for further details. The data we have consists of images that give the difference between the brain scans of participants under the 2-back and 0-back conditions. The data was pre-processed at the first level using nilearn.  the images were then smoothed using an isotropic Gaussian kernel with an FWHM of $ 4/3 $ voxels (4 mm).
 
\section{ Further figures}
\subsection{Simes vs ARI for the IQ contrast}
\begin{figure}[H]
	\begin{center}
		\includegraphics[width=\textwidth]{./Figures/HCP/TDP_boot_iq_sex_1.pdf}
		\includegraphics[width=\textwidth]{./Figures/HCP/TDP_ari_iq_sex_1.pdf}
		\includegraphics[width=\textwidth]{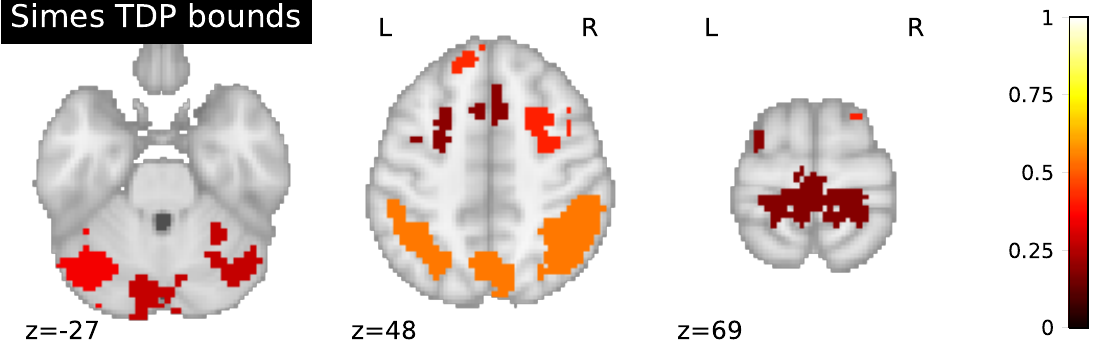}
	\end{center}
\caption{TDP bounds within clusters for the contrast for IQ in the linear regression model fit to the HCP data. Each cluster is shaded a single colour which is the lower bound on the TDP. The upper panel gives the TDP bounds within each cluster provided by the bootstrap procedure. The lower panels gives the bounds provided by using ARI and the Simes procedure. The bounds given by the bootstrap are larger (as indicated by the light colours) indicating that the method is more powerful. (Note that the step-down bootstrap gave the same bounds as the bootstrap and so is not shown.) Note that these images are 2D slices through the 3D brain and so voxels that are part of the same cluster are not necessarily connected.}\label{fig:realdata}
\end{figure}

\subsection{The contrast for sex}\label{S:sexcontrast}
Less activation is found for the contrast of sex in the linear model fit to the HCP data. In this case only a single cluster above the cluster defining threshold has non-zero lower bound. The bound provided is the same for all the parametric and bootstrap methods that we consider, in particular they all conclude that at least one of the 17 voxels within this cluster has non-zero activation. The cluster (and its TDP) is illustrated in Figure \ref{sexcontrast}.

\begin{figure}[H]
	\begin{center}
		\includegraphics[width=0.5\textwidth]{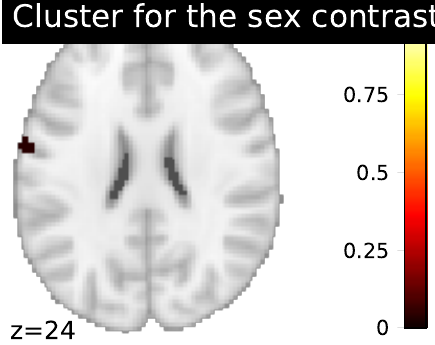}
	\end{center}
\caption{Illustrating the cluster in the sex contrast with non-zero activation.}\label{sexcontrast}
\end{figure}

\subsection{Illustrating the simulation setup}
\begin{figure}[H]
	\begin{center}
		\includegraphics[width=0.49\textwidth]{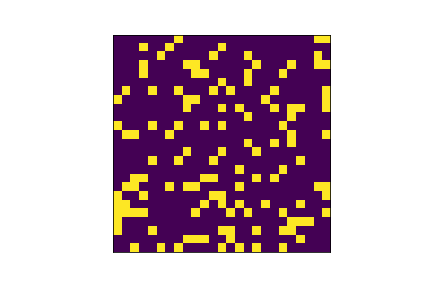}
		\includegraphics[width=0.49\textwidth]{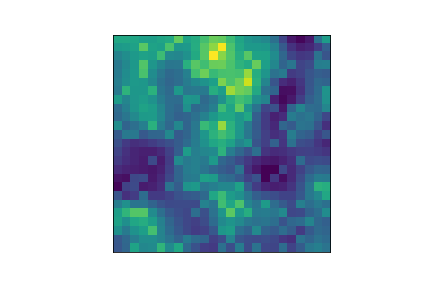}
	\end{center}
	\caption{Illustrating the simulation setup for a domain of size [25,25] and a smoothness of 4 pixels. Left: the signal for the first contrast. Right: a realisation of one of the subjects in $ G_2. $}\label{simimages}
\end{figure}

\newpage
\subsection{Additional JER control plots}
In this section we present the results of the simulations to consider JER control where the domain of the data in the simulations is $25$ by $25$ or $100$ by $100$ rather than $50$ by $50$. The results for the $25$ by $25$ simulations are shown in Figure \ref{fig:fpr25} and those for the  $100$ by $100$ simulations are shown in Figure \ref{fig:fpr100}. The results are similar to those in the main text.

\begin{figure}[h!]
	\begin{subfigure}[t]{0.32\textwidth}
		\centering
		\includegraphics[width=\textwidth]{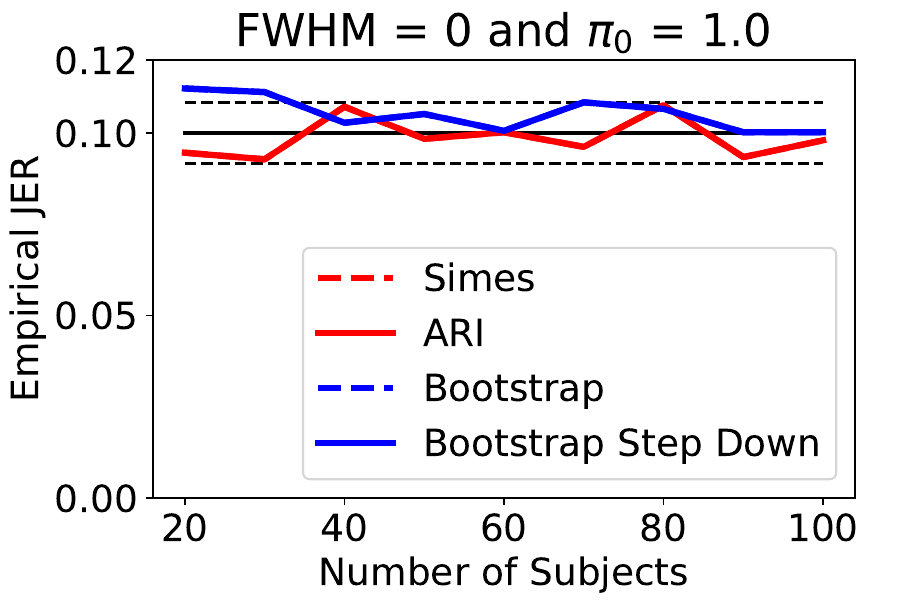}
	\end{subfigure}
	\hfill
	\begin{subfigure}[t]{0.32\textwidth}  
		\centering 
		\includegraphics[width=\textwidth]{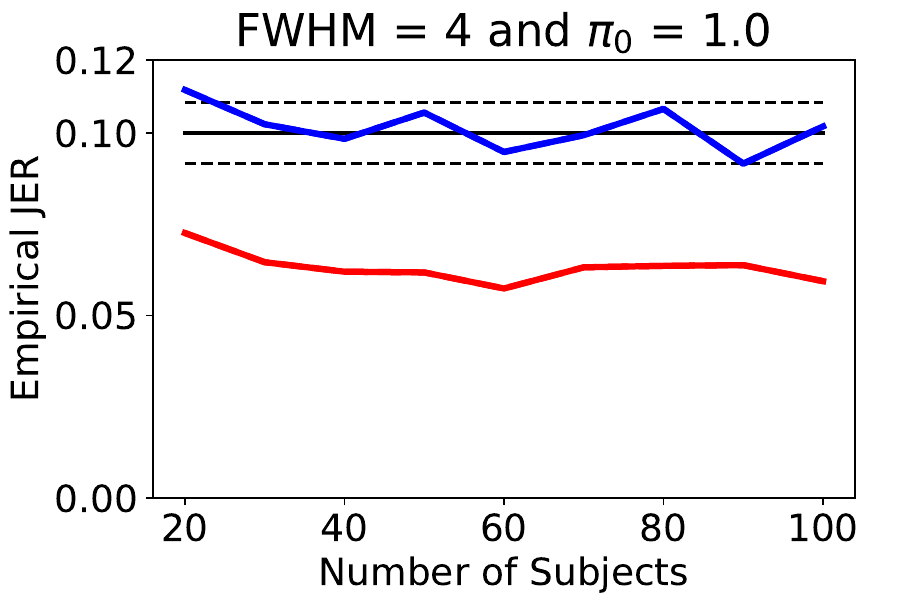}
	\end{subfigure}
	\hfill
	\begin{subfigure}[t]{0.32\textwidth}  
		\centering
		\includegraphics[width=\textwidth]{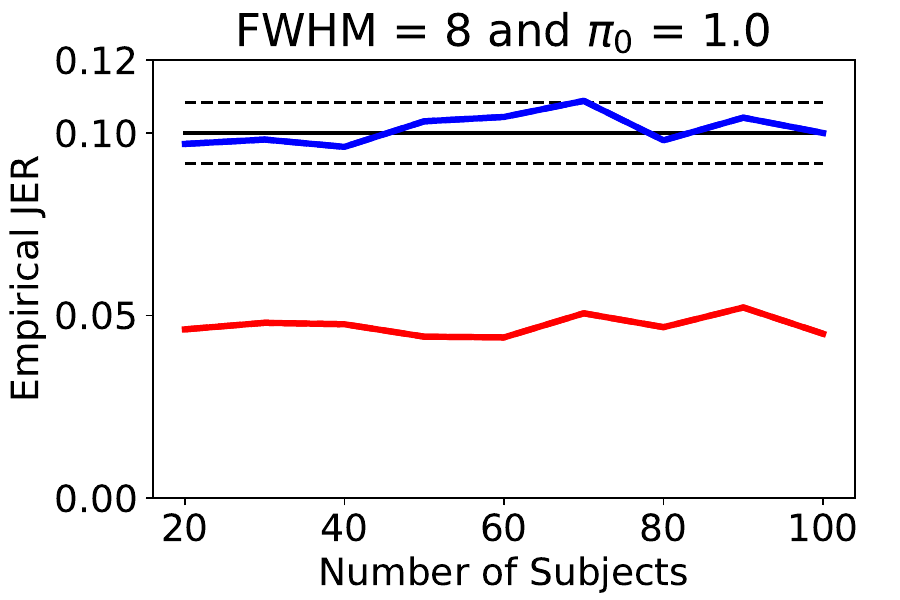}
	\end{subfigure}
	\if\biometrika0
	\begin{subfigure}[t]{0.32\textwidth}
		\centering
		\includegraphics[width=\textwidth]{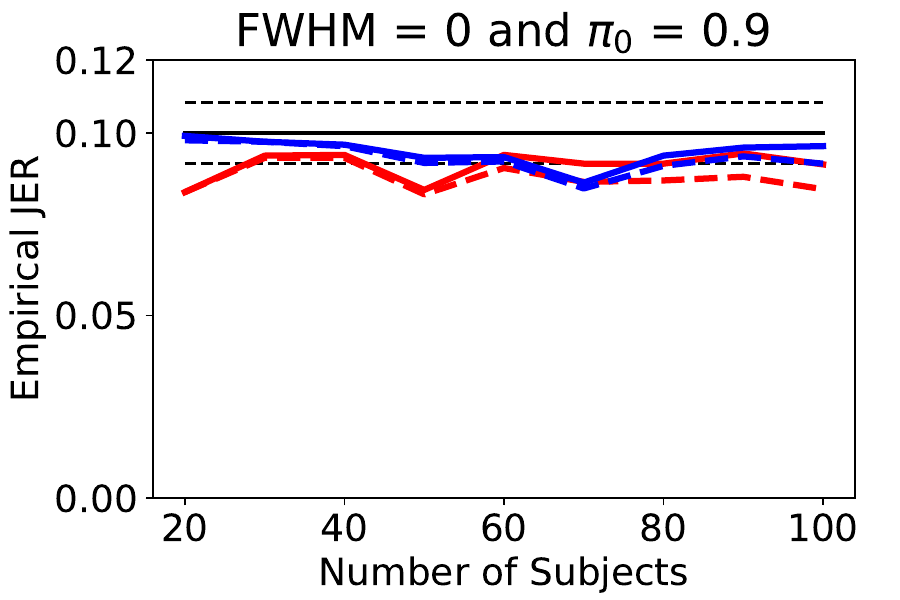}
	\end{subfigure}
	\hfill
	\begin{subfigure}[t]{0.32\textwidth}  
		\centering 
		\includegraphics[width=\textwidth]{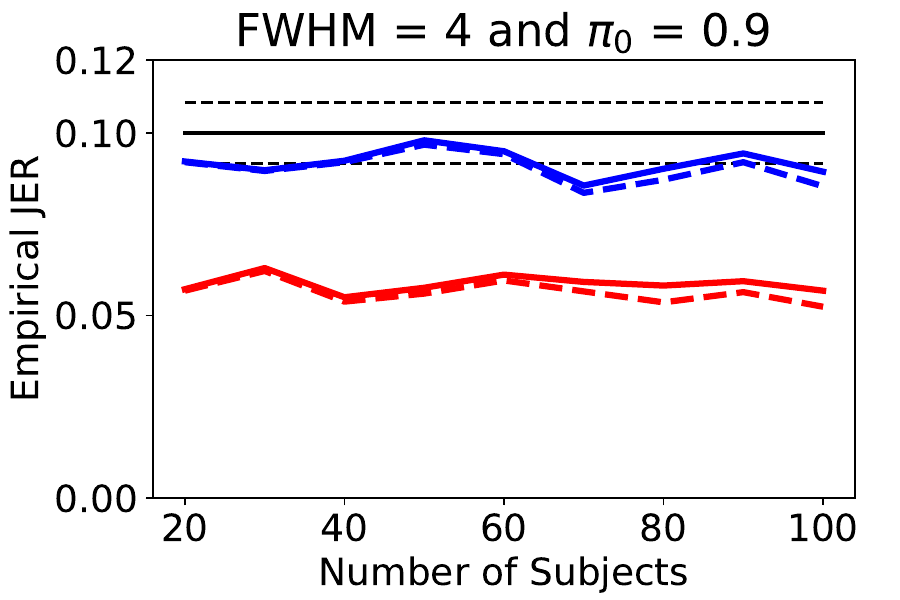}
	\end{subfigure}
	\hfill
	\begin{subfigure}[t]{0.32\textwidth}  
		\centering
		\includegraphics[width=\textwidth]{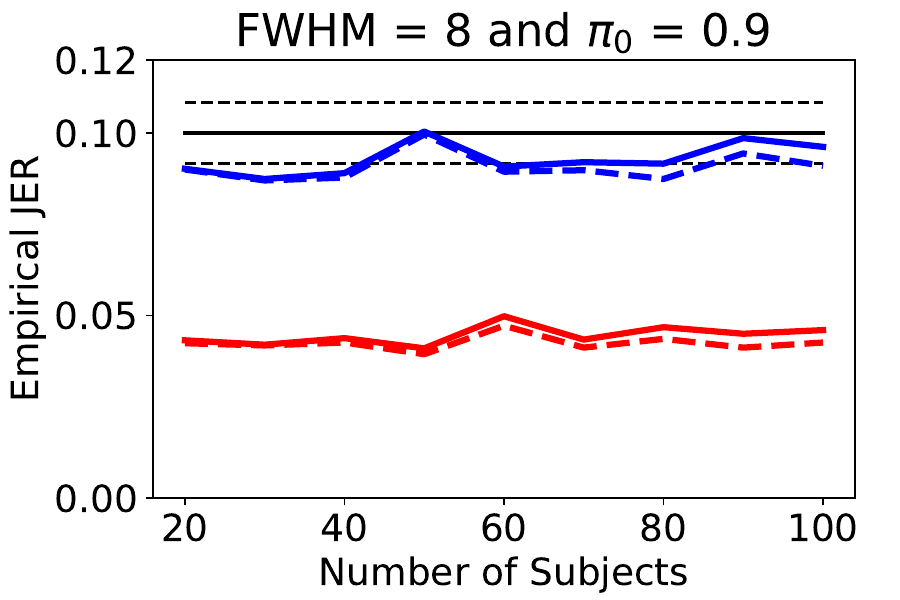}
	\end{subfigure}
	\fi
	\begin{subfigure}[t]{0.32\textwidth}
		\centering
		\includegraphics[width=\textwidth]{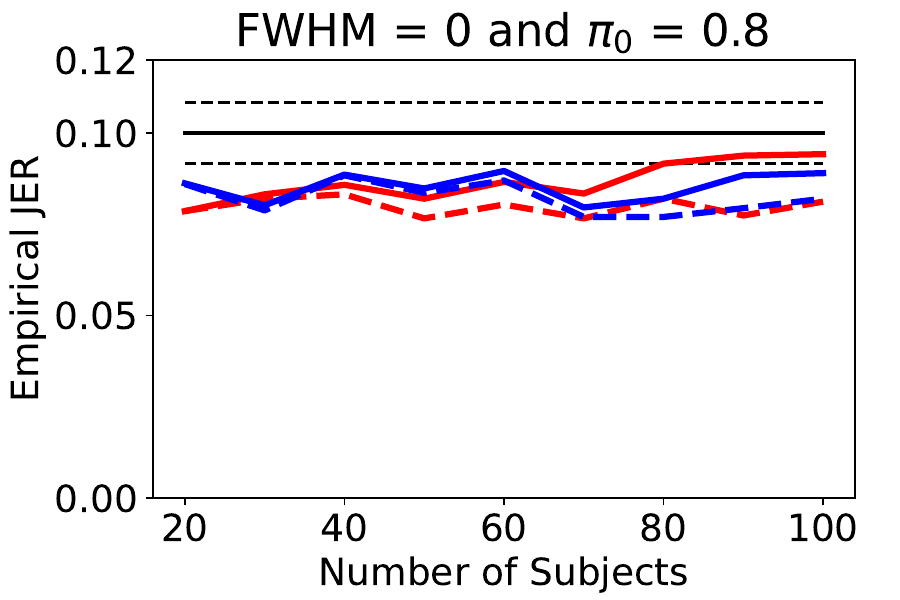}
	\end{subfigure}
	\hfill
	\begin{subfigure}[t]{0.32\textwidth}  
		\centering 
		\includegraphics[width=\textwidth]{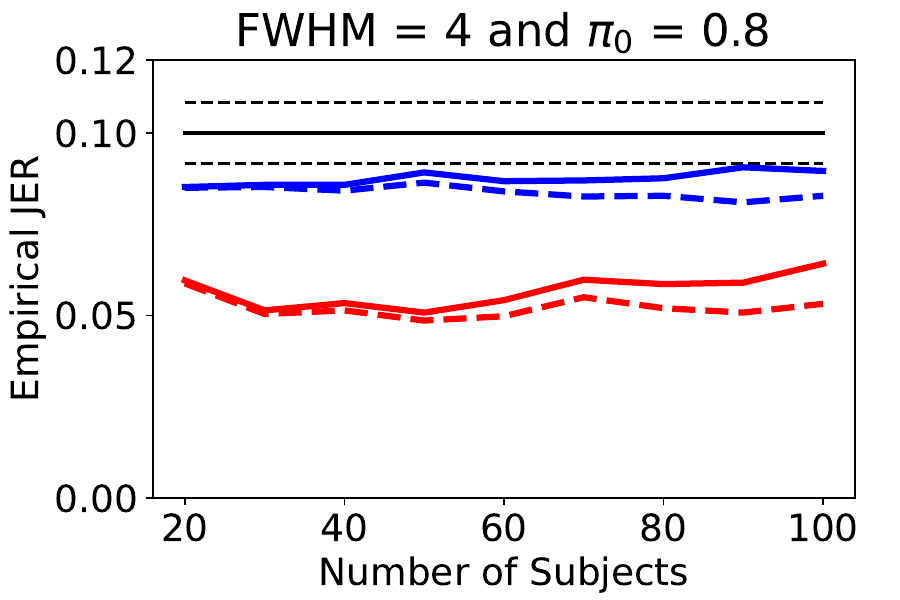}
	\end{subfigure}
	\hfill
	\begin{subfigure}[t]{0.32\textwidth}  
		\centering
		\includegraphics[width=\textwidth]{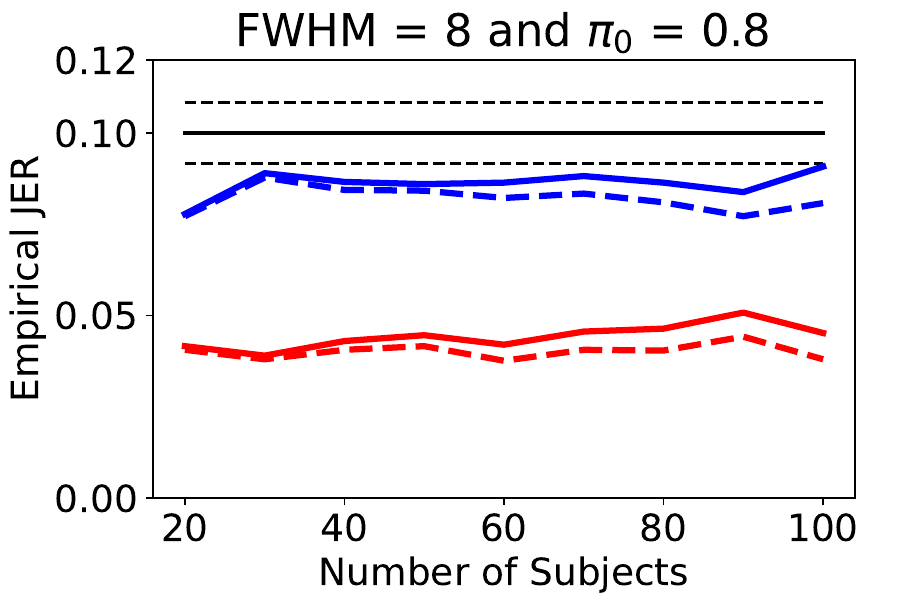}
	\end{subfigure}
	\begin{subfigure}[t]{0.32\textwidth}
		\centering
		\includegraphics[width=\textwidth]{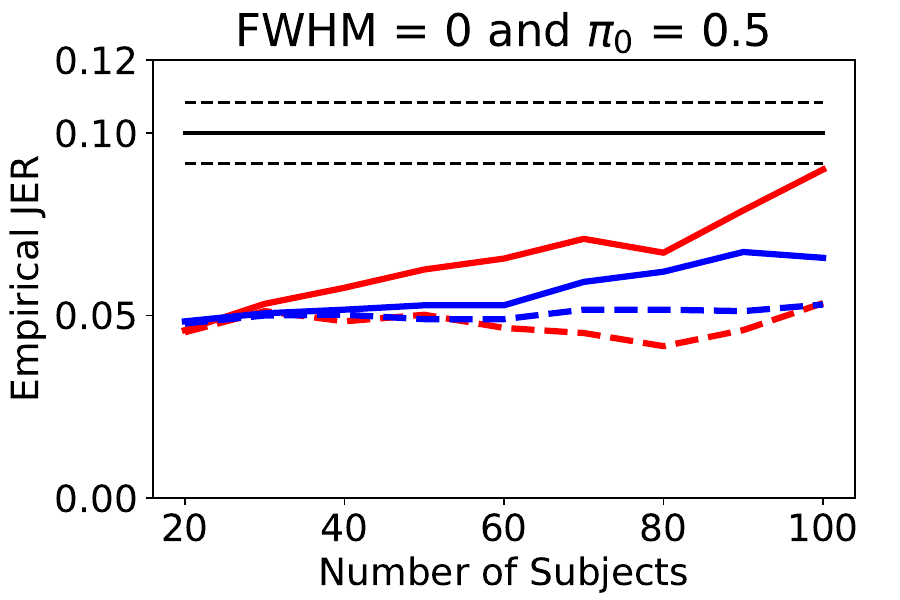}
	\end{subfigure}
	\hfill
	\begin{subfigure}[t]{0.32\textwidth}  
		\centering 
		\includegraphics[width=\textwidth]{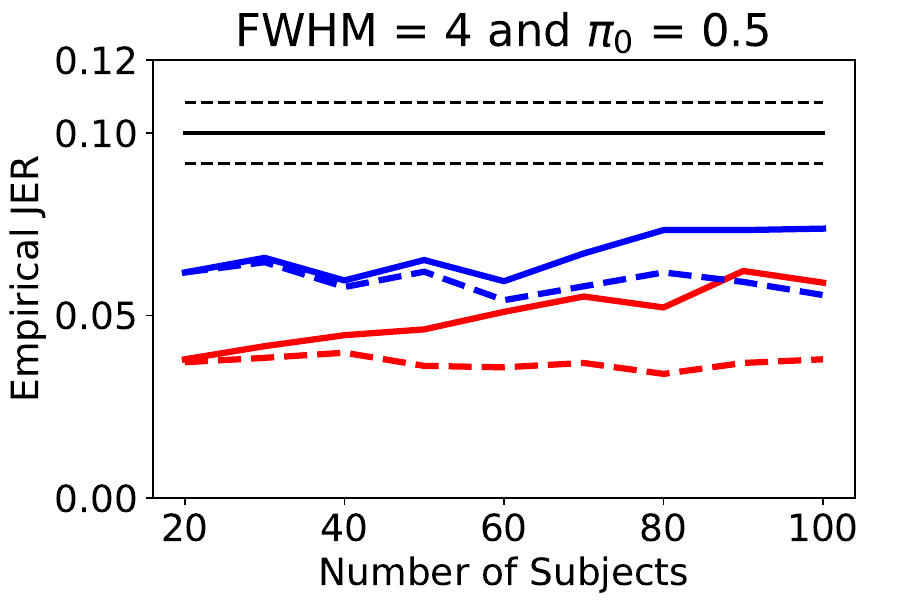}
	\end{subfigure}
	\hfill
	\begin{subfigure}[t]{0.32\textwidth}  
		\centering
		\includegraphics[width=\textwidth]{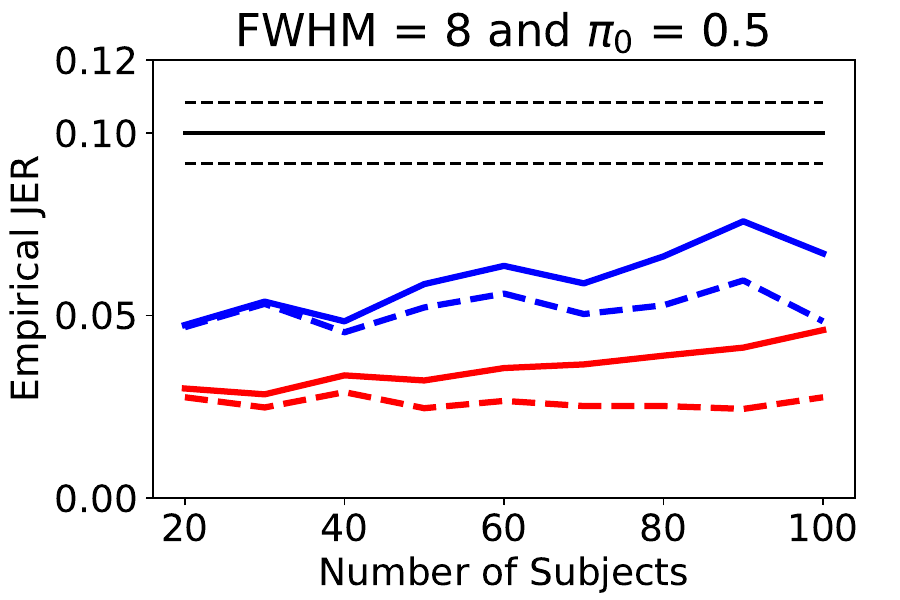}
	\end{subfigure}
	\caption{Empirical joint error rate for the $25$ by $25$ simulations.
	}\label{fig:fpr25}
\end{figure}

\begin{figure}[h!]
	\begin{subfigure}[t]{0.32\textwidth}
		\centering
		\includegraphics[width=\textwidth]{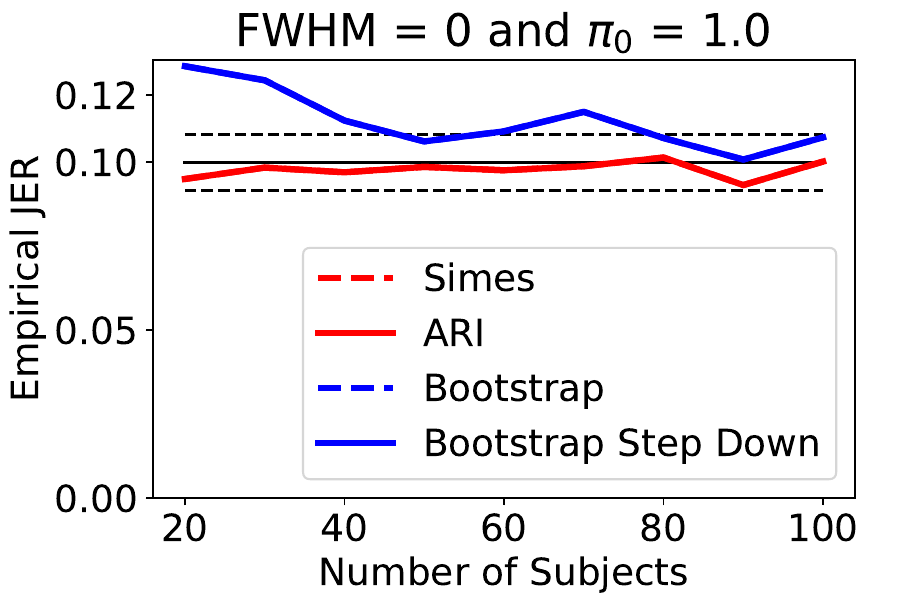}
	\end{subfigure}
	\hfill
	\begin{subfigure}[t]{0.32\textwidth}  
		\centering 
		\includegraphics[width=\textwidth]{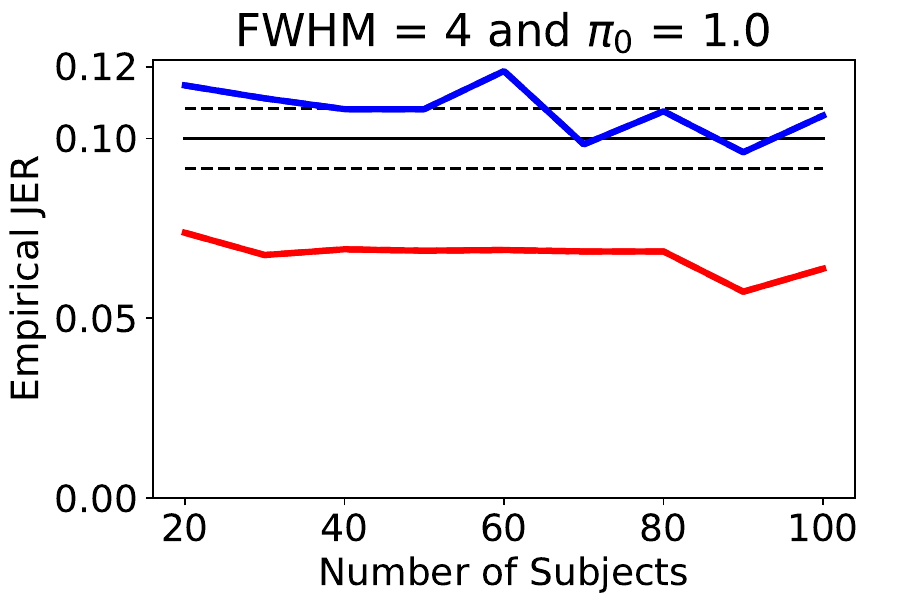}
	\end{subfigure}
	\hfill
	\begin{subfigure}[t]{0.32\textwidth}  
		\centering
		\includegraphics[width=\textwidth]{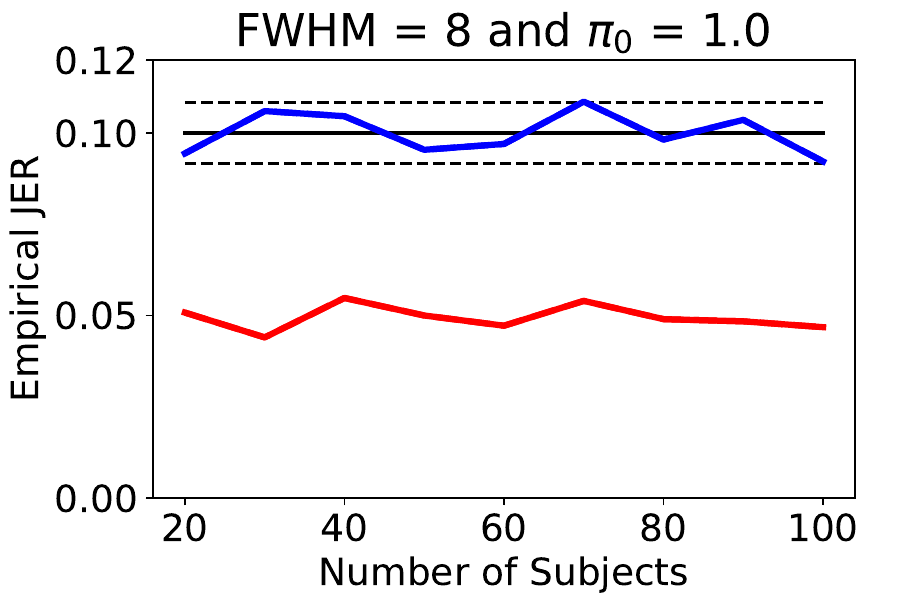}
	\end{subfigure}
	\if\biometrika0
	\begin{subfigure}[t]{0.32\textwidth}
		\centering
		\includegraphics[width=\textwidth]{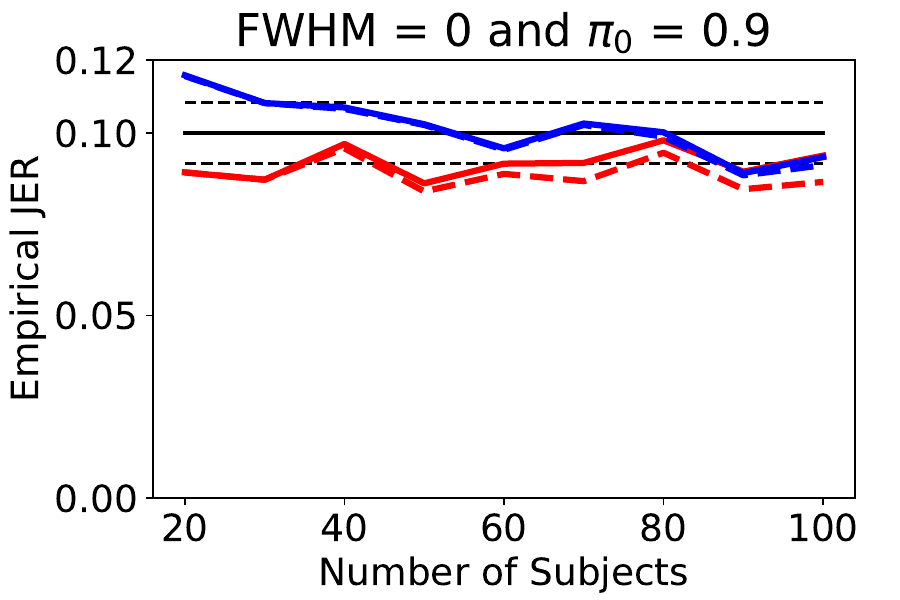}
	\end{subfigure}
	\hfill
	\begin{subfigure}[t]{0.32\textwidth}  
		\centering 
		\includegraphics[width=\textwidth]{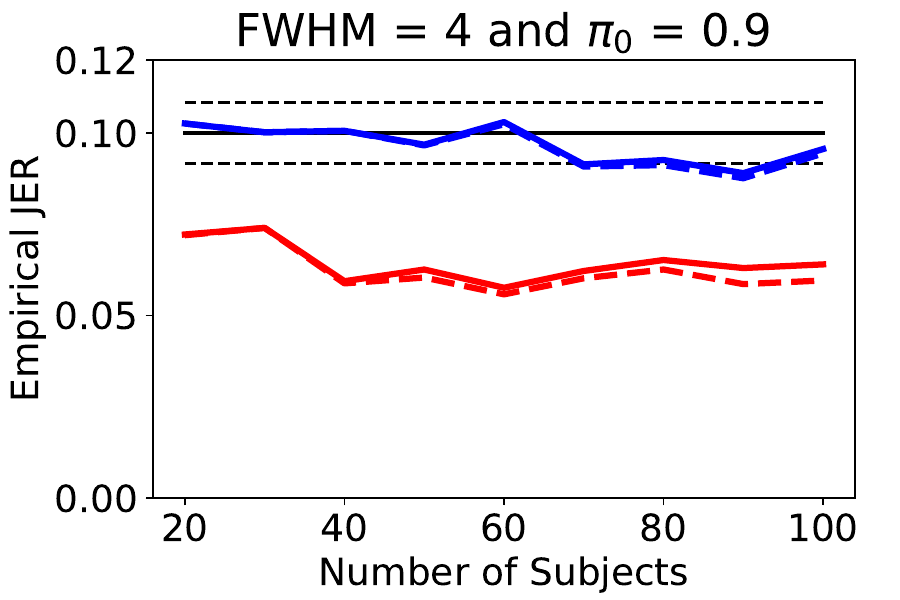}
	\end{subfigure}
	\hfill
	\begin{subfigure}[t]{0.32\textwidth}  
		\centering
		\includegraphics[width=\textwidth]{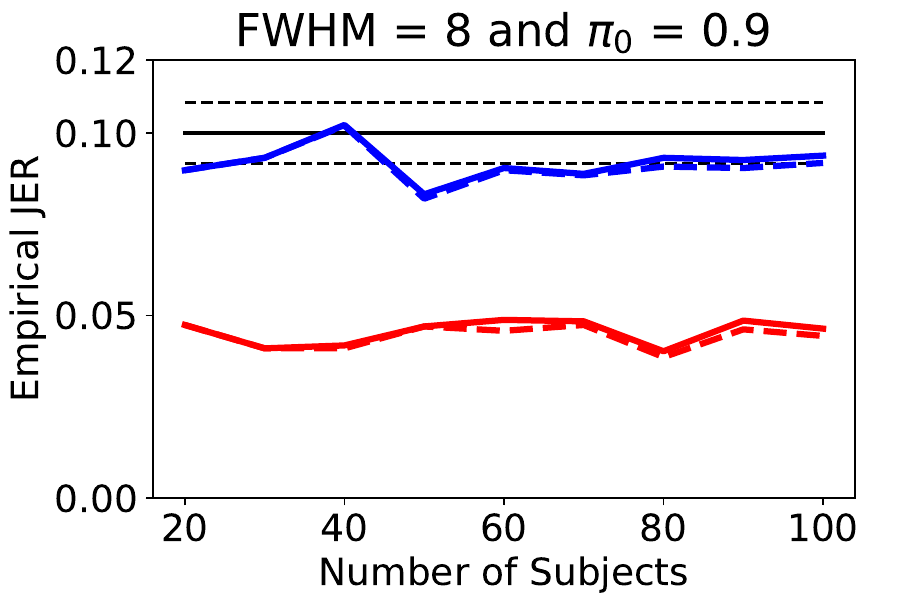}
	\end{subfigure}
	\fi
	\begin{subfigure}[t]{0.32\textwidth}
		\centering
		\includegraphics[width=\textwidth]{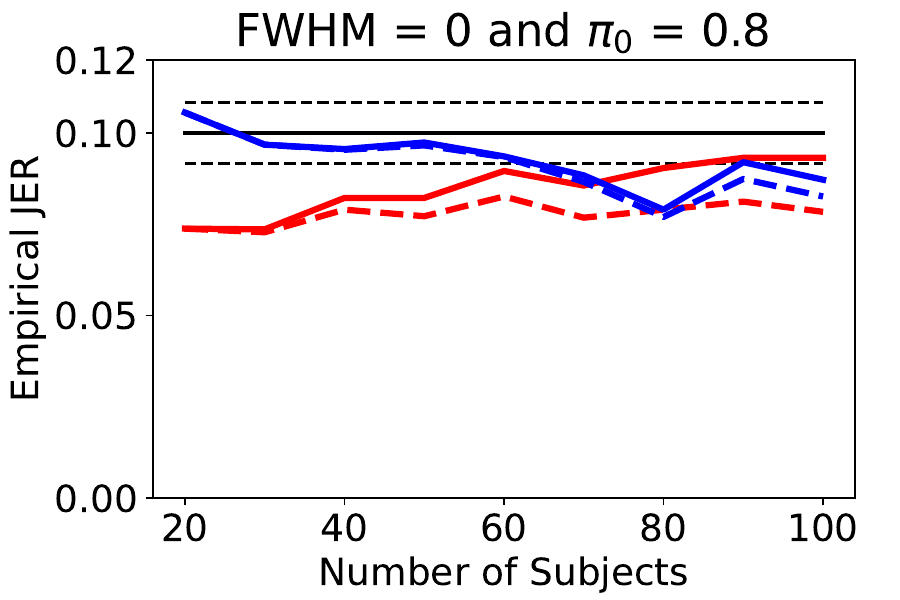}
	\end{subfigure}
	\hfill
	\begin{subfigure}[t]{0.32\textwidth}  
		\centering 
		\includegraphics[width=\textwidth]{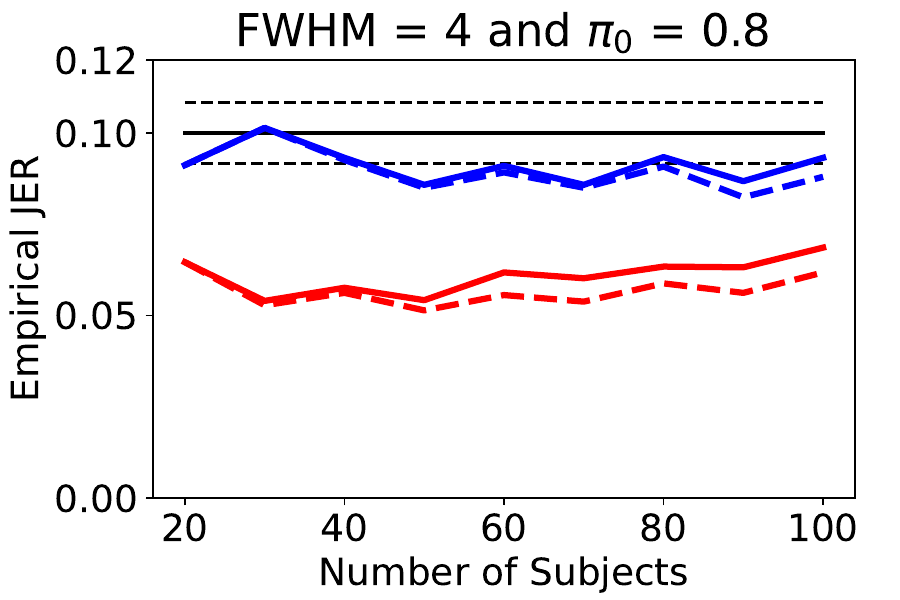}
	\end{subfigure}
	\hfill
	\begin{subfigure}[t]{0.32\textwidth}  
		\centering
		\includegraphics[width=\textwidth]{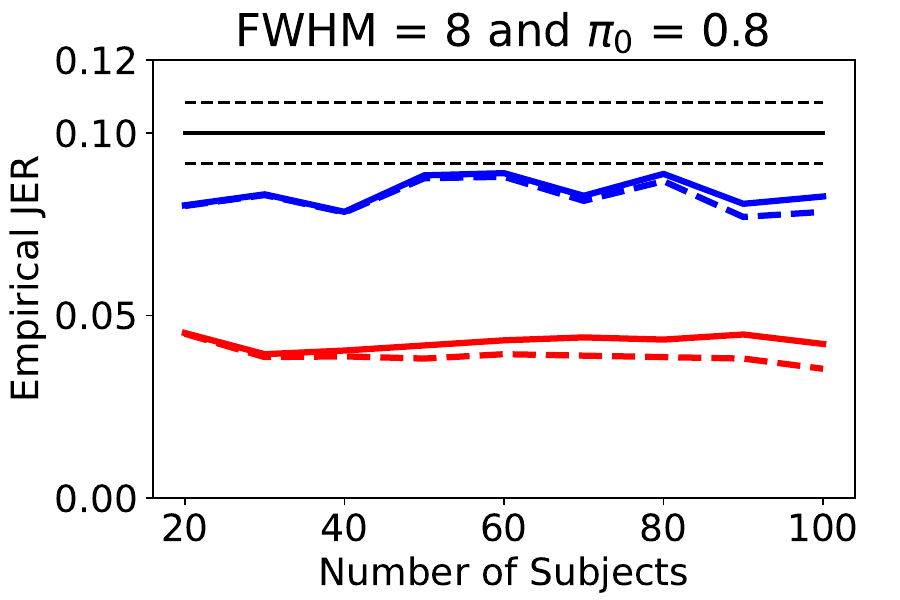}
	\end{subfigure}
	\begin{subfigure}[t]{0.32\textwidth}
		\centering
		\includegraphics[width=\textwidth]{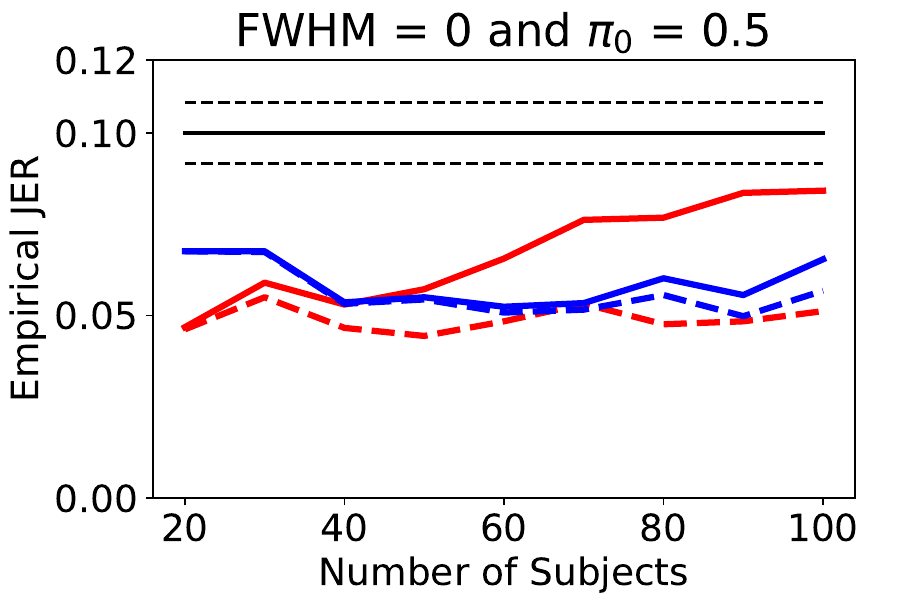}
	\end{subfigure}
	\hfill
	\begin{subfigure}[t]{0.32\textwidth}  
		\centering 
		\includegraphics[width=\textwidth]{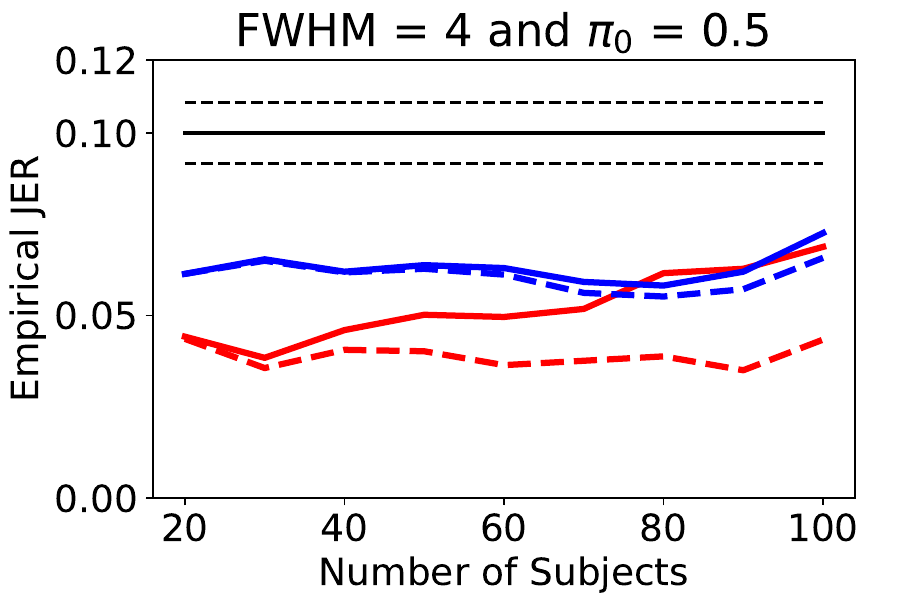}
	\end{subfigure}
	\hfill
	\begin{subfigure}[t]{0.32\textwidth}  
		\centering
		\includegraphics[width=\textwidth]{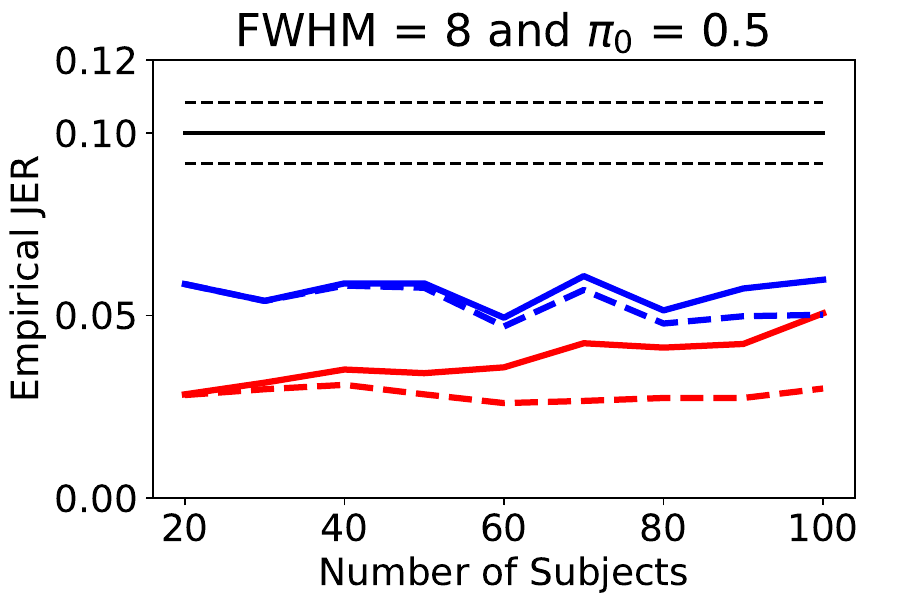}
	\end{subfigure}
	\caption{Empirical joint error rate for the $100$ by $100$ simulations.
	}\label{fig:fpr100}
\end{figure}

\newpage
\subsection{Additional power plots}\label{S:app}
\begin{figure}[H]
	\begin{subfigure}[t]{0.32\textwidth}
		\centering
		\includegraphics[width=\textwidth]{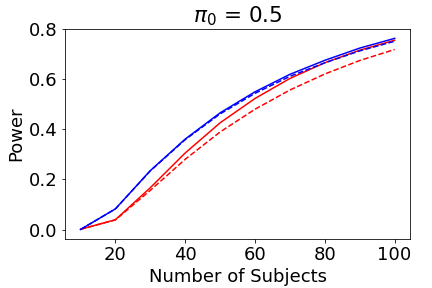}
	\end{subfigure}
	\hfill
	\begin{subfigure}[t]{0.32\textwidth}  
		\centering 
		\includegraphics[width=\textwidth]{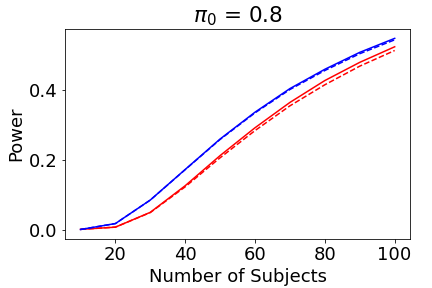}
	\end{subfigure}
	\hfill
	\begin{subfigure}[t]{0.32\textwidth}  
		\centering 
		\includegraphics[width=\textwidth]{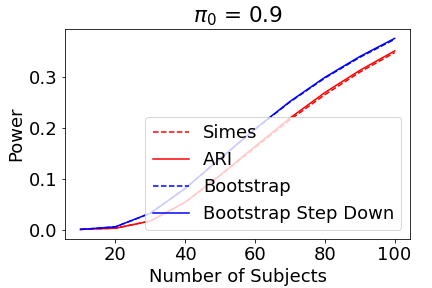}
	\end{subfigure}
	\caption{Plotting the power of the different methods against the numbers of a subjects for setting 3, i.e. taking $ R = \left\lbrace (l, v): p_{n, l}(v) \leq 0.05 \right\rbrace$ in \eqref{eq:power}.}\label{fig:power}
\end{figure}

\if\biometrika1
\subsection{Zoomed in Post-hoc curve plots }
\begin{figure}[H]
	\begin{center}
	\centering\includegraphics[width=0.9\textwidth]{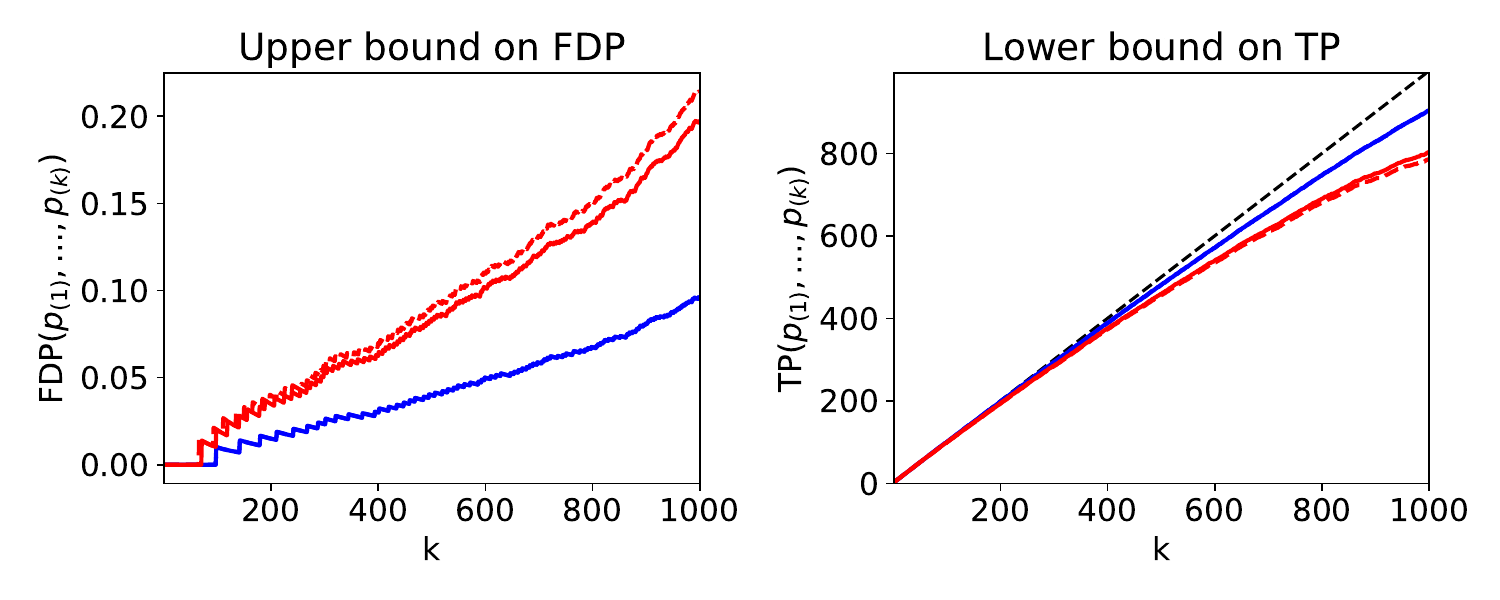}
	\end{center}
\caption{A zoomed in version of Figure \ref{fig:fdpplot} for the 1000 smallest $ p $-values.}\label{fig:lowerpanel}
\end{figure}

\fi

\if\biometrika1
\bibliographystyle{biometrika}
\else
\bibliographystyle{plainnat}
\fi
\bibliography{MHT,fMRI,RFT,Statistics,Other}

\bibliographystyle{plainnat}
\bibliography{MHT,fMRI,RFT,Statistics,Other}

\end{document}